\documentclass{article}
\usepackage{graphicx}
\usepackage{subcaption}
\usepackage{amsfonts}
\usepackage{amsmath}
\usepackage{amssymb}
\usepackage{url}
\usepackage{fancyhdr}
\usepackage{indentfirst}
\usepackage{threeparttable}
\usepackage{epstopdf}
\usepackage{enumerate}
\usepackage{array}
\usepackage{booktabs}
\usepackage{tikz}
\usetikzlibrary{patterns, positioning, arrows}
\usepackage{slashed}
\usepackage{dsfont}
\usepackage[colorlinks=true,citecolor=blue]{hyperref}
\usepackage{amsthm}
\usepackage{color}
\definecolor{r}{rgb}{1,0,0}
\usepackage{natbib}
\usepackage{comment}
\usepackage{capt-of}
\def\d{\mathrm{d}}

\newcommand{\ES}{\mathrm{ES}}

\newcommand{\X}{\mathcal {X}}

\newcommand{\VaR}{\mathrm{VaR}}

\newcommand{\RVaR}{\mathrm{RVaR}}

\newcommand{\E}{\mathbb{E}}
\newcommand{\F}{\mathcal{F}}

\newcommand{\R}{\mathbb{R}}

\newcommand{\p}{\mathbb{P}}

\newcommand{\id}{\mathds{1}}

\renewcommand{\ge}{\geqslant}
\renewcommand{\le}{\leqslant}

\renewcommand{\epsilon}{\varepsilon}
\newcommand{\esssup}{\mathrm{ess\mbox{-}sup}}
\newcommand{\essinf}{\mathrm{ess\mbox{-}inf}}

\theoremstyle{plain}
\newtheorem{theorem}{Theorem}
\newtheorem{corollary}{Corollary}
\newtheorem{lemma}{Lemma}
\newtheorem{proposition}{Proposition}
\theoremstyle{definition}
\newtheorem{definition}{Definition}
\newtheorem{example}{Example}

\newtheorem{assumption}{Assumption}

\theoremstyle{remark}
\newtheorem{remark}{Remark}

\newcommand{\cet}{\begin{center}}
\newcommand{\ecet}{\end{center}}

\usepackage{setspace}

\begin{document}
\title{Calibrating distribution models from PELVE}
\date{\today }
\author{Hirbod Assa\thanks{ Kent Business School, UK.
    \texttt{h.assa@kent.ac.uk}.} \and
    Liyuan Lin\thanks{Department of Statistics and Actuarial Science, University of Waterloo,
    Canada. \texttt{l89lin@uwaterloo.ca}.} \and
    Ruodu Wang\thanks{Department of Statistics and Actuarial Science, University of Waterloo,
    Canada. \texttt{wang@uwaterloo.ca}.} }

\maketitle

\begin{abstract}
The Value-at-Risk (VaR) and the Expected Shortfall (ES)
are the two most popular risk measures in banking and insurance regulation.
To bridge between the two regulatory risk measures, the Probability Equivalent Level of VaR-ES (PELVE) was recently proposed to convert a level of VaR to that of ES. It is straightforward to compute the value of  PELVE for a given distribution model.
In this paper, we study the converse problem of PELVE calibration, that is, to find a distribution model that yields a given PELVE, which may either be obtained from data or from expert opinion. We discuss separately the cases when one-point, two-point, $n$-point and curve constraints are given. In the most complicated case of a curve constraint, we convert the calibration problem to that of an advanced differential equation.  We apply the model calibration techniques to estimation and simulation for   datasets used in insurance.
We further study some technical properties of PELVE by offering a few new results on monotonicity and convergence.
\end{abstract}


\newcommand\blfootnote[1]{%
  \begingroup
  \renewcommand\thefootnote{}\footnote{#1}%
  \addtocounter{footnote}{-1}%
  \endgroup
}

\section{Introduction}
Value-at-Risk (VaR) and Expected Shortfall (ES, also known as TVaR and CVaR) are the most widely used risk measures for regulation in finance and insurance. The former has gained its popularity due to its simplistic approach toward risk as the risk quantile, and the second one is perceived to be useful as a modification of VaR with more appealing properties, such as tail-sensitivity and subadditivity, as studied in the  seminal work of \cite{ADEH99}.

In the Fundamental Review of the Trading Book (FRTB), the Basel Committee on Banking Supervision (\citet{BASEL19}) proposed to replace VaR at 1\% confidence with ES with a 2.5\% confidence interval for the internal model-based approach.\footnote{In this paper, we   use  the ``small $\alpha$" convention for $\VaR$ and $\ES$. Hence,   ``VaR at 1\% confidence" and ``ES at 2.5\% confidence"  correspond to  $\VaR_{99\%}$ and $\ES_{97.5\%}$ in \cite{BASEL19}, 
respectively.} The main reason, as mentioned in the FRTB, was that $\ES$ can better capture tail risk; see  \cite{ELW18} for a concrete risk sharing model where tail risk is captured by ES and ignored by VaR. On the other hand, VaR also has advantages that ES does not have, such as elicitability (e.g., \cite{G11} and \cite{KP16}) or backtesting tractability (e.g., \cite{AS14}), and the two risk measures admit different axiomatic foundations (see \cite{C09} and \cite{WZ21}). We refer to  the reviews of \citet{EPRWB14} and \citet{EKT15} for general discussions on VaR and ES, and \citet{MFE15} for a standard treatment on risk management including the use of VaR and ES.
The technical contrasts of the two risk measures and their co-existence in regulatory practice give rise to  great interest from both researchers and practitioners to explore the relationship between them.

To understand the balancing point of VaR and ES during the transition in the FRTB, \citet{LW22} proposed the Probability Equivalent Level of VaR-ES (PELVE). The value of PELVE is the multiplier to the tail probability when replacing VaR with ES such that the capital calculation stays unchanged.
More precisely, the PELVE of $X$ at level $\epsilon$ is the multiplier $c$ such that $\ES_{c\epsilon}(X)=\VaR_{\epsilon}(X)$; such $c$ uniquely exists under mild conditions.
For instance, if $\VaR_{1\%}(X)=\ES_{2.5\%}(X)$ for a future portfolio loss $X$, then PELVE of $X$  at probability level 0.01 is the multiplier $2.5$. In this case, replacing $\VaR_{1\%}$ with  $ \ES_{2.5\%} $ in FRTB does not have much effect on the capital requirement for the bank bearing the loss $X$. Instead, if $\ES_{2.5\%}(X)>\VaR_{1\%}(X)$, then the bank has a larger capital requirement under the new regulatory risk measure; this is often the case for financial assets and portfolios as shown by the empirical studies in \cite{LW22}. The PELVE enjoys many convenient properties, and it has been extended in a few  ways. In particular,
\citet{FG2021} defined generalized PELVE by replacing $\VaR$ and $\ES$ with another pair of monotone risk measures $(\rho,\tilde{\rho})$, and
\cite{B22} extended  PELVE by replacing ES with a   higher-order ES.

For a given distribution model or a data set, its PELVE can be computed or estimated in a straightforward manner.
As argued by \cite{LW22}, the PELVE for a small $\epsilon$ may be seen as a summarizing index measuring  tail heaviness in a non-limit sense. As such, one may like to generate models for a given PELVE, in a way similar to constructing models for other given statistical information; see e.g., \cite{EMS02, EHW16} for constructing multivariate models with a given correlation or tail-dependence matrix. Such statistical information may be obtained either from data or from expert opinion, but there is no a priori guarantee that a corresponding model exists.
Since PELVE involves a parameter $\epsilon\in (0,1)$, its information is represented by a curve. The calibration problem, that is, to find a distribution model for given PELVE values or a given PELVE curve, turns out to be highly non-trivial, and it is the main objective of the current paper.

From now on,   suppose that we receive some information on the PELVE of a certain random loss from an expert opinion, and we aim to build a distribution model consistent with the supplied information.
Since PELVE is location-scale invariant, such a distribution, if it exists, is not unique.

The calibration problem is trivial if we are supplied with only one point on the PELVE curve.
As the PELVE curve of the generalized Pareto distribution is a constant when the PELVE is well defined, we can use the generalized Pareto distribution to match the given PELVE value, which has a tail index implied from the expert opinion.
The calibration problem becomes more involved if we are supplied with two points on the PELVE curve, because the value of the PELVE at two different probability levels interact with each other.
The situation becomes more complicated as the number of points increases, and we further turn to the problem of calibration from a fully specified PELVE curve. Calibrating distribution from the PELVE curve can be reformulated as solving for a function $f$ via the integral equation $\int_{0}^{y}f\left(s\right)\d s=yf\left(z\left(y\right)y\right)$, where the curve $z$ is computed from the PELVE curve. This integral equation can be
further converted to an advanced differential equation (see \cite{BC63}).
For the case that $z$ is a constant curve, we can explicitly obtain all solutions for $f$. We find other distributions that also have constant PELVE curves other than the simple ones with a Pareto or exponential distribution. As a consequence, a PELVE curve does not characterize a unique location-scale family of distributions; this provides a negative answer to a question posed by \citet[Section 7, Question (iv)]{LW22}. For general function $z$,
we develop a numerical method to compute $f$.


The calibrated distribution can be used to estimate the value of other risk measures such as VaR and ES at different levels. We illustrate by an empirical example that two points of PELVE give a quite good summary of  the tail distribution of risk.  Daily log-losses (negative log-returns) of Apple (AAPL) from Yahoo Finance are collected for the period from January 3, 2012  to December 31, 2021 within total of 2518 observations.
We calculate the empirical PELVE at levels 0.01 and 0.05 using the empirical PELVE estimator provided by \citet[Section 5]{LW22} with a moving window of  500 trading days. 
For each pair of two points of PELVE at levels 0.01 and 0.05, we produce a quantile curve from the two empirical PELVE points by our calibration model in Section \ref{sec:32}, which is scaled such that $\VaR_{0.01}$ and $\VaR_{0.05}$ are equal to their empirical values.\footnote{Recall that PELVE is location-scale free, and hence we need to pick two free parameters to specify a  distribution calibrated from PELVE.}  
Figure \ref{fig:VaR} presents the empirical and calibrated quantile curves on December 31, 2021 using 500 trading days prior to that date. The two quantile curves are close to each other, with our calibrated curve being more smooth.  We also report the  values of $\ES_{0.025}$ of the calibrated distribution, which we call the calibrated $\ES_{0.025}$, and compare it with empirical $\ES_{0.025}$.  The left panel of Figure \ref{fig:ES} shows the curves of empirical  and calibrated $\ES_{0.025}$. In the right panel of Figure \ref{fig:ES}, we create a scatter plot using empirical  and calibrated $\ES_{0.025}$.
Both figures show that the empirical and calibrated $\ES_{0.025}$ curves are quite close.
 
 \begin{figure}[h]
    \centering
       \caption{Empirical $\VaR$ and calibrated $\VaR$}\label{fig:VaR}
    \includegraphics[scale=0.45]{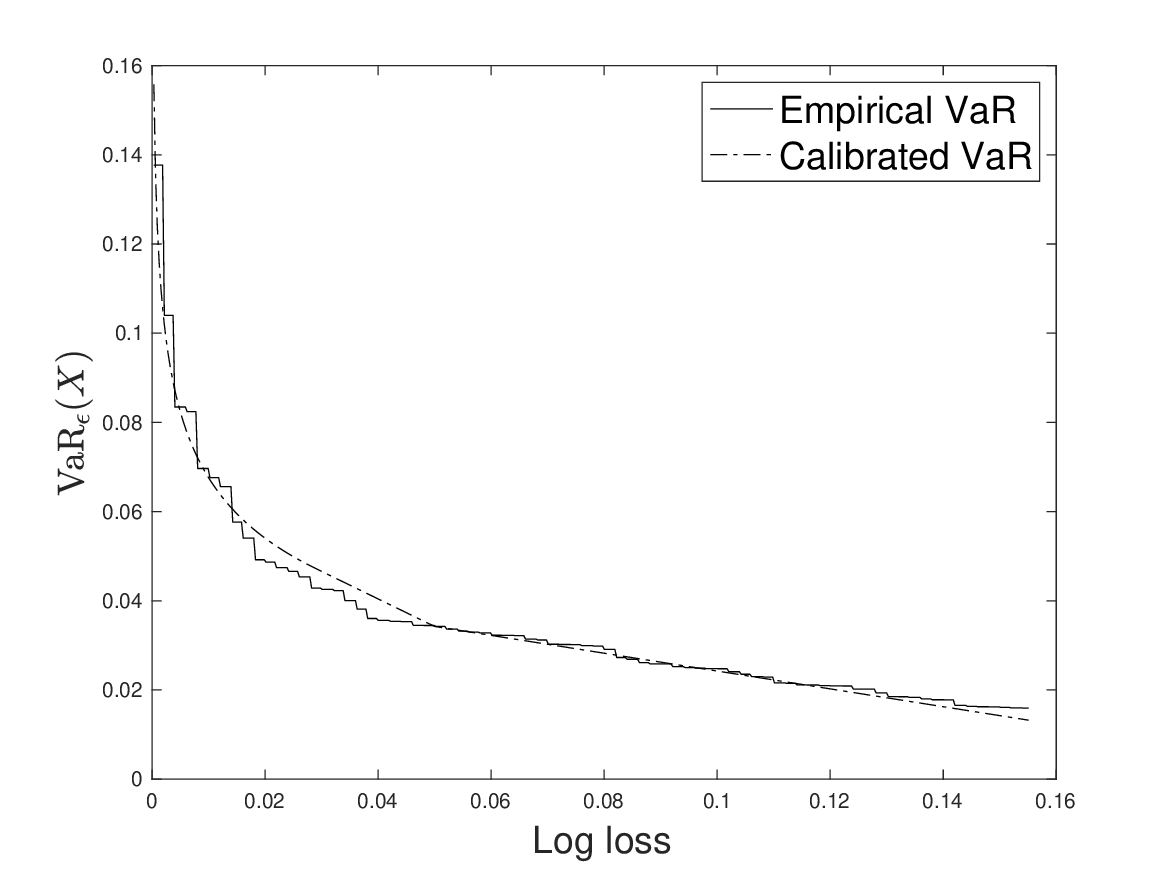}
    \end{figure}

 \begin{figure}[h]
    \centering
      \caption{Empirical $\ES_{0.05}$ and calibrated  $\ES_{0.05}$}\label{fig:ES}
    \includegraphics[scale=0.37]{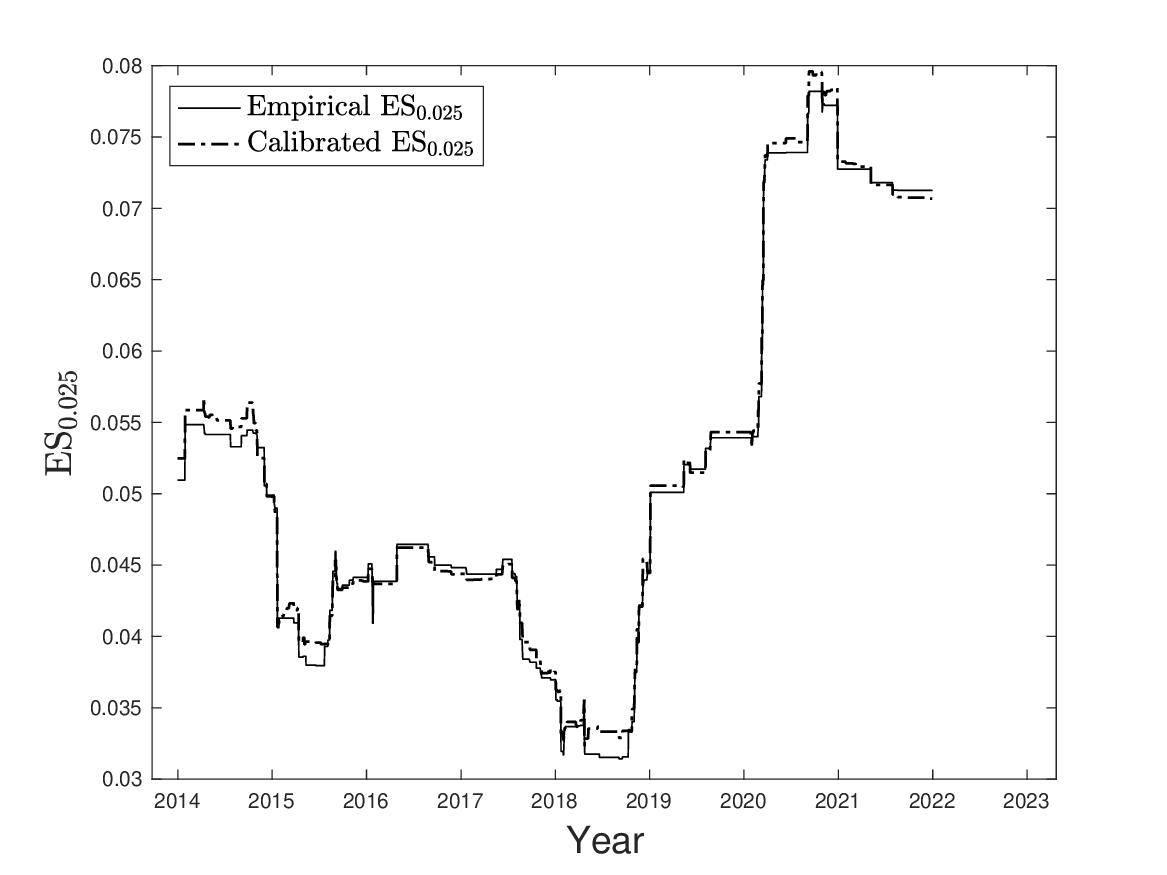}
        \includegraphics[scale=0.37]{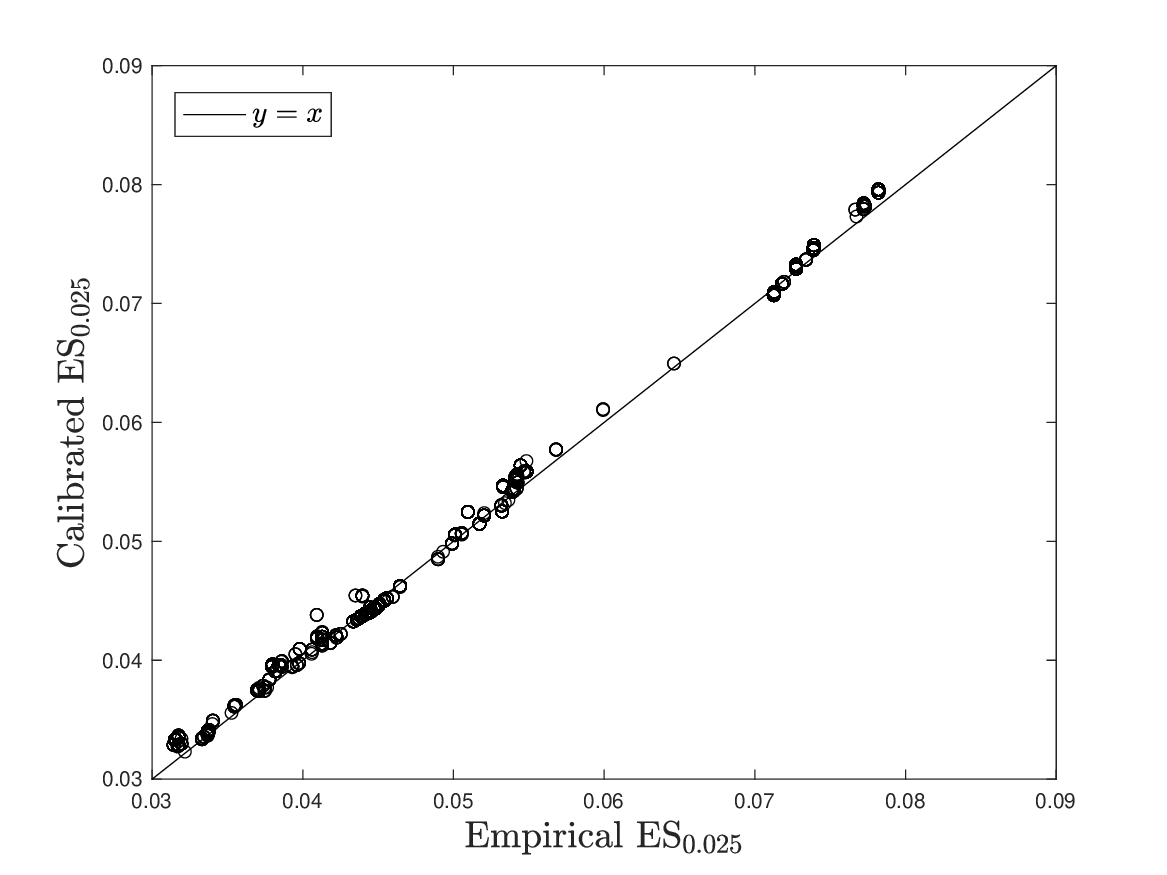}
    \end{figure}

To further enrich the theory of PELVE,  we study a few technical properties of PELVE, such as monotonicity and convergence as the probability level goes to 0.   A decreasing PELVE indicates a relatively larger impact of $\ES$ in risk assessment than $\VaR$ moving towards the tail. As we will see, while for the most known parametric distributions the PELVE is decreasing, there exist  some examples at some risk levels it is not decreasing. This means that for those examples $\VaR$ becomes a stricter risk measure when moving towards the tail. To obtain conditions for monotonicity, we define the dual PELVE by moving the multiplier $c$ from the $\ES$ side to  the $\VaR$ side.
PELVE can be seen as a functional measure   of tail heaviness in the sense that a heavier-tailed distribution has a higher PELVE curve (\citet[Theorem 1]{LW22}). The hazard rate, on the other hand, is another functional measure of tail heaviness.
We show that the   PELVE is decreasing (increasing) if the inverse of the hazard rate is convex (concave).
Monotonicity also leads to conditions for the PELVE to have a limit at the tail, which from the risk management perspective, identifies the ultimate relative positions of $\ES$ and $\VaR$ in the tail region. From a mathematical perspective,
  the limit of PELVE at $0$ allows us to extend the domain of PELVE to include $0$ as a measure of tail heaviness.

The  rest of this paper is organized as follows.
Section \ref{sec:definition} introduces the background and examples of the PELVE.
In Section \ref{sec:points} we calibrate a distribution from finitely many points in the PELVE curve. 
Section \ref{sec:function} calibrates the distribution from given PELVE curves, where we give a class of explicit solutions for constant PELVE functions  and  numerical solutions for  general PELVE functions.
In Section \ref{sec:property}, we study the monotonicity and convergence of the PELVE. Section  \ref{sec:6} presents two examples of the model calibration techniques applied to   datasets used in insurance.
A conclusion is given in Section \ref{sec:conclusion}. Some technical proofs of results in  Sections \ref{sec:points}, \ref{sec:function} and \ref{sec:property} are provided in the Appendices.

\section{Definitions and background}\label{sec:definition}

Let us consider an atomless probability space $\left(\Omega,\F,\p\right)$, where $\mathcal{F}$ is the set of the measurable sets and $\p$ is the probability measure. Let $L^{1}$ be the set of integrable random variables, i.e., $L^{1}= \{X: \E [\vert X\vert]<\infty \}$, where $\E$ is the expectation with respect to $\p$. 

 We first define $\VaR$ and $\ES$ in $L^{1}$, the two most popular risk measures. The $\VaR$   at probability level $p\in(0,1)$ is defined as
\begin{equation}\label{eq:left_VaR}
\VaR_{p}(X)  =\inf\{x\in\R:\p(X\le x)\ge 1-p\}=F^{-1}(1-p),~~~~X\in L^{1},
\end{equation}
where $F$ is the distribution of $X$. The $\ES$ at probability level $p\in[0,1)$ is defined as
$$
\ES_{p}(X) =\frac{1}{p}\int_{0}^{p}\mathrm{VaR}_{q}(X)\d q,~~~~X\in L^{1}.
$$
Note that we use the ``small $\alpha$" convention for  $\VaR$ and $\ES$, which is different from \cite{LW21}.
Let $\VaR_{0}(X)=\ES_{0}(X)=\esssup(X)$ and $\VaR_{1}(X)=\essinf(X)$. We have that $\ES_{1}(X)$ is the mean of $X$. We will also call $p\mapsto \VaR_p(X)$ the quantile function of $X$, keeping in mind that in our convention this function is decreasing.\footnote{Throughout the paper, all terms like ``increasing" and ``decreasing" are in the non-strict sense.} 

For $\epsilon\in(0,1)$, the PELVE at level $\epsilon$,  proposed by \citet{LW22}, is defined as
$$
\Pi_{X}(\epsilon)  =\inf\left\{ c\in[1,1/\epsilon]:\ES_{c\epsilon}(X)\le\VaR_{\epsilon}(X)\right\} ,\quad X\in L^{1},
$$
where $\inf (\varnothing)  =\infty$.
 \cite{LW22} used $\Pi_\epsilon (X)$ for our $\Pi_X(\epsilon)$, and our choice of notation is due to the fact that the curve $\epsilon \to \Pi_X(\epsilon)$ is the main quantity of interest in this paper.

The PELVE of $X$ is finite if and only if $\VaR_{\epsilon}(X)\ge \E[X]$.
The value of the PELVE is the multiplier $c$ such that $\ES_{c\epsilon}(X)=\VaR_{\epsilon}(X).$
If $\VaR_{p}(X)$ is not a constant for $p \in (0,\epsilon]$, then the PELVE is the unique solution for the multiplier. By Theorem 1 in \cite{LW22}, the PELVE is location-scale invariant. The distribution with a heavy tail will have a higher PELVE value.

If $X$ is a normal distributed random variable and $\epsilon=1\%$, we have $\Pi_{X}(\epsilon)\approx2.5$. It means that $\ES_{2.5\%}(X)\approx\VaR_{1\%}(X)$. That is, the replacement suggested by BCBS is fair for normally distributed risks. In other words, a higher PELVE will result in a higher capital requirement after the replacement.

In this paper, we are generally interested in the question of which distributions have a specified  or partially specified PELVE curve.
We first look at a few simple examples.

\begin{example}[Constant PELVE]
We first list some distributions that have constant PELVE  curves.
From the definition of the PELVE, we know that the PELVE should be larger than 1.
As we can see from Table \ref{constant-PELVE}, the PELVE for the generalized Pareto distribution takes  values on $(1, \infty)$. For $X \sim \text{GPD}(\xi)$, we have  $1<\Pi_{X}(\epsilon)<e$ when $\xi <0$,  $\Pi_{X}(\epsilon)=e$ when $\xi =0$ and  $\Pi_{X}(\epsilon)>e$ when $\xi >0$. Furthermore, if $X$ follows the point-mass distribution $\delta_{c}$ or the Bernoulli distribution, we have $\Pi_X(\epsilon)=1$.
\end{example}

\begin{table}[htbp]
\def\arraystretch{1.4}
    \centering{}
    \caption{Example of constant PELVE}\label{constant-PELVE}
    \begin{threeparttable}
    \begin{tabular}{m{2cm}<{\centering}| m{7cm}<{\centering}|m{4cm}<{\centering}}
    \hline\hline
    Distribution &  Distribution or probability function of $X$ & $\Pi_{X}(\epsilon)$\\
    \hline
    $\delta_c$ &$\p(X=c)=1$ &$\Pi_{X}(\epsilon)=1$ for $\epsilon \in (0,1)$\\
    \hline
    $\mathrm{B}(1,p)$& $\p(X=1)=p$ and $\p(X=0)=1-p$ & $\Pi_{X}(\epsilon)=1$ for $\epsilon \in (0,p)$\\
    \hline
    $\mathrm{U}(0,1)$ & $F(t)=t$ for $t \in (0,1)$ & $\Pi_{X}(\epsilon)=2$  for $0<\epsilon<1/2$\\
    \hline
    $\text{EXP}(\lambda)$ & $F(t)=1-\exp(-\lambda t),~~\lambda>0$  & $\Pi_{X}(\epsilon)=e$ for $0<\epsilon<1/e$\\
    \hline
    $\text{GPD}(\xi)$\tnote{1}&
    $F(x)=\left\{\begin{aligned}
            &1-\left(1+\xi x\right)^{-\frac{1}{\xi}}~~~& \xi \neq 0\\
            &1-\exp(-x)~~~& \xi=0
    \end{aligned}\right.$
    &
    $\Pi_{X}(\epsilon)=(1-\xi)^{-\frac{1}{\xi}}$ for $0<\epsilon<(1-\xi)^{\frac{1}{\xi}}$\\
    \hline\hline
    \end{tabular}
    \begin{tablenotes}
     \footnotesize
      \item[1] The distribution $\text{GPD}(\xi)$ is called the standard generalized Pareto distribution. As $\E[X]<\infty$ when $\xi<1$, the PELVE exists only when $\xi<1$.
          The support of $\text{GPD}(\xi)$ is $[0, \infty)$ when $\xi>0$ and $[0,-\frac{1}{\xi}]$ when $\xi<0$. When $\xi=0$, the $\text{GPD}(\xi)$ is exactly exponential distribution with $\lambda=1/\sigma$. There is a three-parameter $\text{GPD}(\mu,\sigma,\xi)$, which is a location-scale transform of standard GPD. Therefore, $\text{GPD}(\mu,\sigma,\xi)$ has the same PELVE as $\text{GPD}(\xi)$.
    \end{tablenotes}
    \end{threeparttable}
\end{table}

\begin{example}\label{ex:decreasing-PELVE}
Here we present some non-constant PELVE examples. We write $\mathrm{t}(v)$ for the t-distribution with parameter $(0,1,v)$, and $\mathrm{LN}(\sigma)$ for the log-normal distribution with parameter $(0,\sigma^2)$.
    \begin{figure}[h]
    \centering
        \caption{PELVE for normal distribution, t-distribution and lognormal distribution}
    \label{PELVE}
    \includegraphics[scale=0.32]{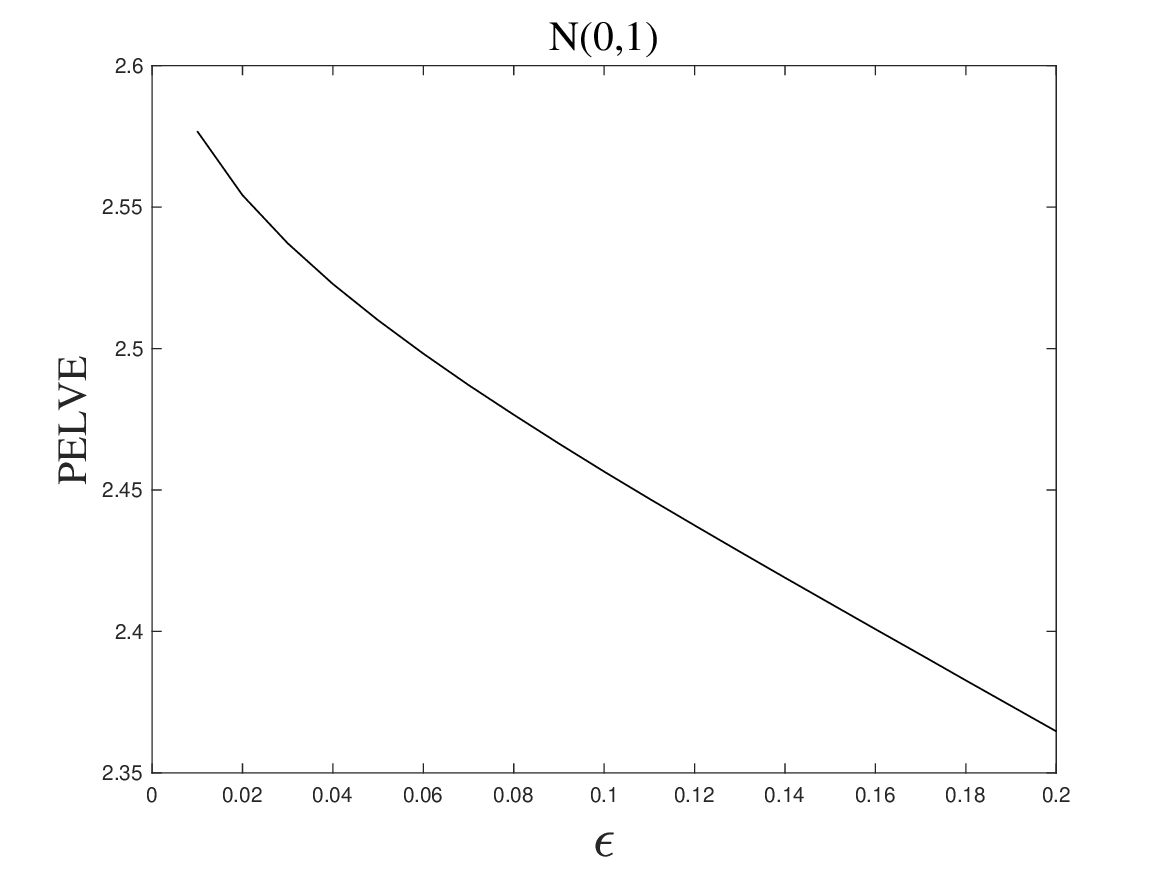}
    \includegraphics[scale=0.32]{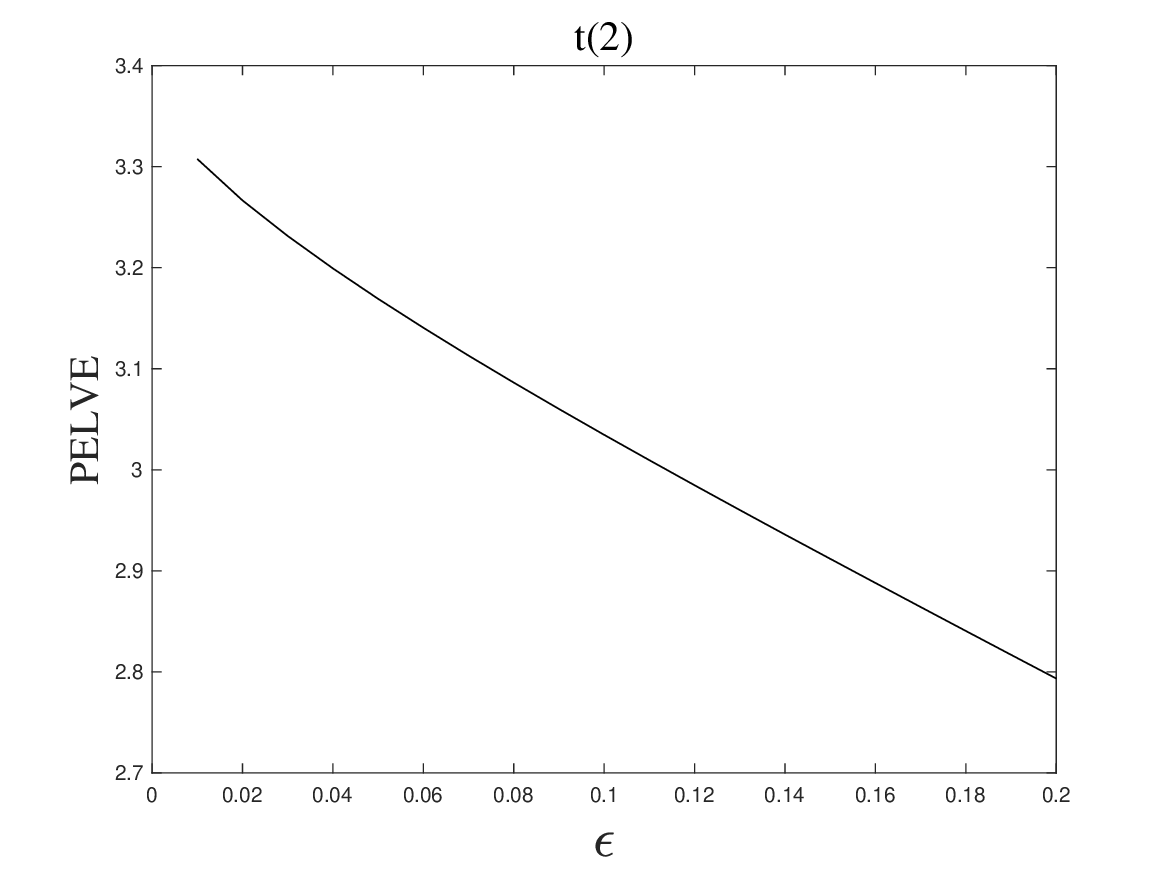}
    \includegraphics[scale=0.32]{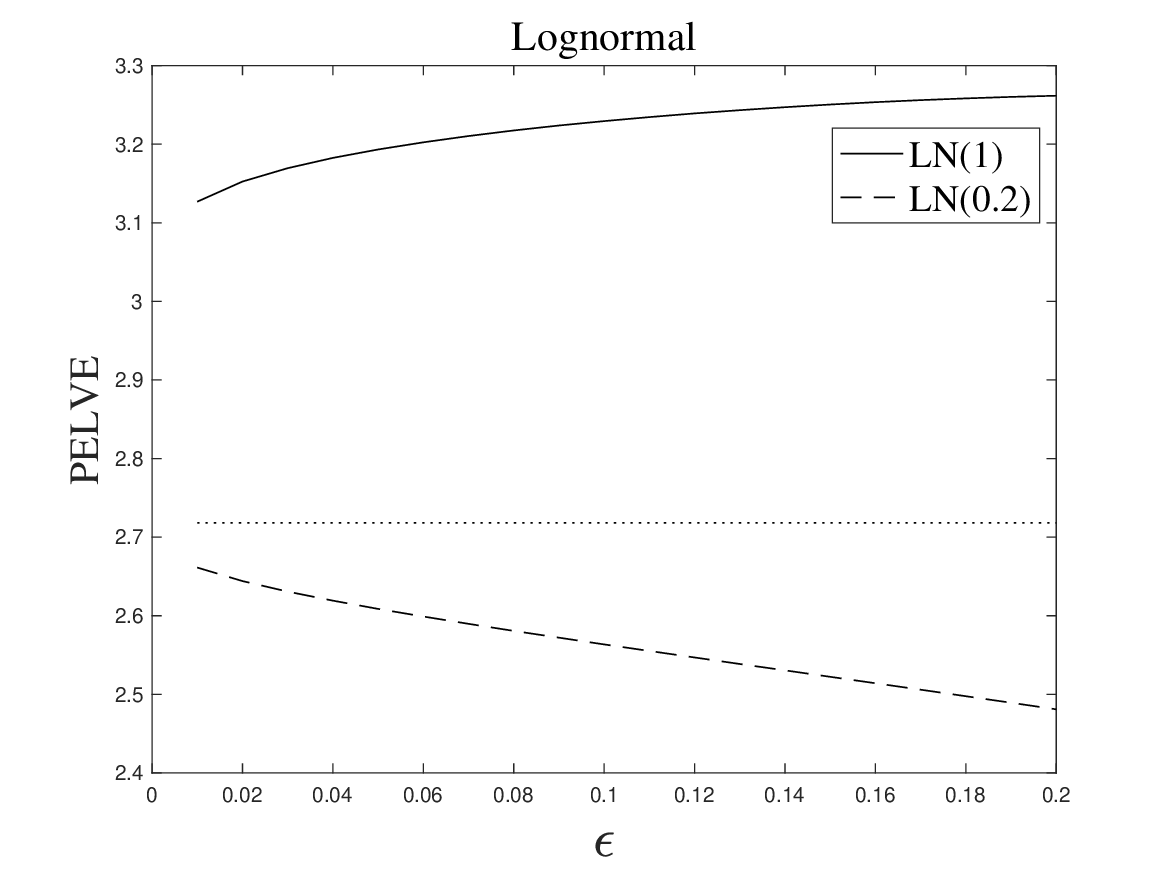}
    \end{figure}
As we can see, for normal distribution and t-distribution, the PELVE curve is decreasing as $\epsilon$ increasing. The monotonicity of the PELVE of the lognormal distribution depends on the value of $\sigma$. The monotonicity of the PELVE will be further discussed  in Section \ref{sec:property}.
For more PELVE examples, see \cite{LW22}.
\end{example}

\section{Calibration from finite-point constraints}\label{sec:points}
In this section, we discuss the calibration problem when some points of the PELVE are given.
We will focus on the case where one point or two points on the PELVE curve are specified, for which we can explicitly construct a corresponding quantile function.

We first note that the calibrated distribution is not unique.
For example, if we are given $\Pi_{X}(0.01)=2.5$, we can assume the distribution of $X$ is the Normal distribution or the generalized Pareto distribution with tail parameter $\xi$ satisfying $(1-\xi)^{-1/\xi}=2.5$ from Table \ref{constant-PELVE}.  Therefore, the distributions  obtained in our results are only some possible choices, which we choose to have a
generalized Pareto tail, as Pareto tails are standard in risk management applications.
\subsection{Calibration from a one-point constraint}
Based on Table \ref{constant-PELVE}, we can calibrate the distribution for $X$ from one given PELVE point $(\epsilon_1, c_1)$ such that $\Pi_{X}({\epsilon_1})=c_1$.
A simple idea is to take the generalized Pareto distribution when $c_1>1$ and $\delta_c$ when $c_1=1$. We summarize the idea in the following Proposition.

\begin{proposition}\label{pro:one-point}
  Let $\epsilon_1 \in (0,1)$ and $c_1 \in [1,\infty)$ such that $c_1\epsilon_1\le 1$.
  If $c_1>1$, let $\xi \in \R$  such that $(1-\xi)^{-\frac{1}{\xi}}=c_1$.
  Then, $X \sim \text{GPD}(\xi)$ has $\Pi_{X}({\epsilon_1})=c_1$.
  If $c_1=1$, then $X=k$ for some constant $k \in \R$ has $\Pi_{X}({\epsilon_1})=c_1$.

\end{proposition}
The proof can be directly derived from Table \ref{constant-PELVE} and it is omitted.
By Proposition \ref{pro:one-point}, if we have the value of PELVE at point $\epsilon_1$, we can find a distribution of $X$ which has the same PELVE value at $\epsilon_1$. If we also have the value of $\VaR$ at $\epsilon_1$, we can determine the scale parameter ($\sigma$) for the GPD distribution or the value of $k$ to match the value of $\VaR$. For Table \ref{constant-PELVE}, we can see that the calibrated generalized Pareto distribution can also serve as a solution for a more prudent condition $\Pi_X(\epsilon)\ge c_1$ when $\epsilon \in (0, \epsilon_1)$.

\subsection{Calibration from a  two-point constraint}
\label{sec:32} 
The calibration problem would be much more difficult when we are given two points of the PELVE curve. Given two points $(\epsilon_1, c_1)$ and $(\epsilon_2, c_2)$ such that $\epsilon_1<\epsilon_2$, we want to find a distribution for $X \in L^1$ such that $\Pi_{X}({\epsilon_1})=c_1$ and $\Pi_{X}(\epsilon_2)=c_2$. Nevertheless, the choices of $(\epsilon_1, c_1)$ and $(\epsilon_2, c_2)$ are not arbitrary. First, we need $1 \le c_1\le 1/\epsilon_1$ and $1 \le c_2\le 1/\epsilon_2$ by the definition of the PELVE. Then, we will show that the value of $c_2$ will be restricted if  $(\epsilon_1, c_1)$ and $\epsilon_2$ are given.

\begin{lemma}\label{lem:bound_1}
For any $X \in L^1$, let $\epsilon_1,\epsilon_2 \in (0,1)$ be such that $\E[X]\le \VaR_{\epsilon_2}(X)$ and  $\epsilon_1<\epsilon_2$. Then, we have $\epsilon_1 \Pi_{X}(\epsilon_1)\le \epsilon_2 \Pi_{X}(\epsilon_2)$.
\end{lemma}

By Lemma \ref{lem:bound_1},  for
given $\epsilon_1, \epsilon_2$ and $c_1$, the value of $c_2$ is bounded below by both 1 and $ {c_1\epsilon_1}/{\epsilon_2}$.
We also note that if $c_2=1$, then $p\mapsto \VaR_{p}(X)$ is constant on $(0,\epsilon_2)$, which implies $c_1=1$. In Appendix \ref{app:proof_sec2},
Proposition \ref{lem:low_bound} shows that the above lower bound is achieved if and only if $\VaR_{\epsilon_1}(X)=\VaR_{\epsilon_2}(X)$.

From the definition of the PELVE and Lemma \ref{lem:bound_1}, for $\epsilon_1<\epsilon_2$, the possible choices of $(\epsilon_1,c_1)$ and $(\epsilon_2,c_2)$ should satisfy $1\le c_1\le1/\epsilon_1$, $1\le c_2\le 1/\epsilon_2$ and $c_1\epsilon_1\le c_2\epsilon_2$.
We denote by $\Delta$ the admissible set for $(\epsilon_1,c_1,\epsilon_2 ,c_2) $, that is,
$$\Delta=\{(\epsilon_1,c_1,\epsilon_2 ,c_2) \in ((0,1)\times[1, \infty))^2: \epsilon_1<\epsilon_2,~ c_1\epsilon_1\le 1,~ c_2\epsilon_2\le 1,~c_1\epsilon_1\le c_2\epsilon_2  \}.
$$
We illustrate the possible region of $(c_1,c_2)$ with given $\epsilon_1$ and $\epsilon_2$ in Figure \ref{fig:region}. We divide the region into 5 cases and calibrate the distribution for each case.

\begin{figure}[!htb]
\centering
\caption{Admissible region of $(c_1,c_2)$}
\label{fig:region}
\begin{tikzpicture}
[scale=.4,description/.style=auto]

\draw[-latex] (0,-1) -- (0,16);
\draw[-latex] (-1,0) -- (18,0);
\node[left] at (0,16){{\small $c_2$}};
\node[below] at (18,0){{\small $c_1$}};

\draw [top color=gray!50, bottom color=gray!50, white, opacity =0.5] (6,2)--(2,2)--(2,12)--(6,12)--(6,2);
\fill [white,opacity=0.3,postaction={pattern=north east lines}] (6,2)--(6,12)--(14,12)--(6,2);

\draw[thin] (2,0)--(2,14);
\draw[very thick] (2,2)--(2,12);
\draw[thin] (14,0)--(14,14);
\node[below] at (2,0){{\small 1}};
\node[below] at (14,0){{\small $1/\epsilon_1$}};

\draw[thin] (0,2)--(2,2);
\draw[dash pattern={on 0.84pt off 2.51pt}] (2,2)--(14,2);
\draw[thin] (0,12)--(16,12);
\node[left] at (0,2){{\small 1}};
\node[left] at (0,12){{\small $1/\epsilon_2$}};

\draw[very thick] (6,2)--(14,12);
\node[below] at (6,0){{\small $\epsilon_2/\epsilon_1$}};
\draw[dash pattern={on 0.84pt off 2.51pt}] (6,0)--(6,12);

\node at (2,2){$\bullet$};
\path[<-, draw,  thick, dashed] (2,2)
      to[out=270, in=180] (-1,1)
      node[left] {\small Case 1};
\path[<-, draw,  thick, dashed] (2,7)
      to[out=200, in=45] (-1,5)
      node[left] {\small Case 2};
\path[<-, draw,  thick, dashed] (10,7)
      to[out=300, in=90] (12,4)
      node[below] {\small Case 3};
\node at (4,7){{\small Case 4}};
\node at (8.7,9){{\small Case 5}};
\end{tikzpicture}

\end{figure}
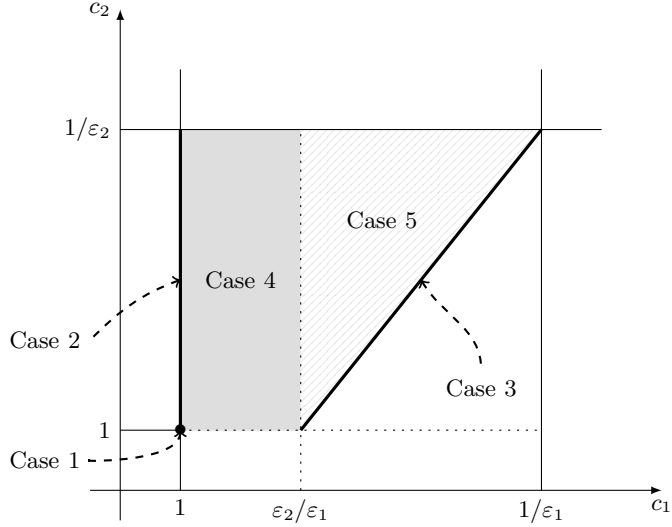

The calibration process is to construct a continuous and decreasing  quantile function that can satisfy two equivalent conditions between $\VaR$ and $\ES$, which are
\begin{equation}\label{eq:VaR-ES}\ES_{c_1\epsilon_1}(X)=\VaR_{\epsilon_1}(X) \text{~~~ and ~~~} \ES_{c_2\epsilon_2}(X)=\VaR_{\epsilon_2}(X).\end{equation}
As we can see,  only  the values of $\VaR_{\epsilon}(X)$ for $\epsilon \in (0,c_2\epsilon_2]$ matters for the equivalent condition \eqref{eq:VaR-ES}. Therefore, we focus on constructing $\VaR_{\epsilon}(X)$ for $\epsilon \in (0,c_2\epsilon_2]$. In addition, we want a continuous calibrated quantile function.

The case $c_1=1$ or $c_2=1$ is special, which means that $\VaR_{\epsilon}(X)$ is a constant on the tail part. If $c_1>1$, we can set the tail distribution as the generalized Pareto distribution from Table \ref{constant-PELVE} such that $\Pi_{X}(\epsilon_1)=c_1$.

For  $\mathbf z=(\epsilon_1,c_1,\epsilon_2,c_2) \in \Delta$,
we will construct a class of functions, denoted by $G_{\mathbf z}$, in five different cases according to Figure \ref{fig:region}.
The function $t\mapsto G_{\mathbf z}(t)$  will be our desired quantile function. If $c_1=1$, let $\hat k,\tilde k \in \R$ be any two constants satisfying $\tilde k<\hat k$.
If $c_1>1$, let $\xi \in (-\infty,1)$ be such that $(1-\xi)^{-{1}/{\xi}}=c_1$,
$$
k(\epsilon)=\left\{\begin{aligned}
&\frac{1}{\xi}(\epsilon^{-\xi}-1), &\xi\neq0,\\
&- \log(\epsilon), &\xi=0,
\end{aligned}\right.$$
and $k=\int_0^{\epsilon_1} k(\epsilon) \d\epsilon$. 
We first claim that the function $G_{\mathbf z}$   can be any arbitrary continuous and decreasing function on $[c_2\epsilon_2,1)$ since the values of $\VaR_{t}(X)$ for $t\in [c_2\epsilon_2,1)$ do not affect its PELVE at $\epsilon_1$ and $\epsilon_2$.  The value of $G_{\mathbf z}$ on $(0,c_2\epsilon_2]$ is given by
\begin{enumerate}[(i)]
 \item \underline{Case 1}, $c_2=1$ (which implies $c_1=1$):
     $G_{\mathbf z}(\epsilon)=\hat k;$
 \item \underline{Case 2}, $c_1=1$ and $1<c_2\le 1/\epsilon_2$:
    $$
G_{\mathbf z}(\epsilon)=\left\{\begin{aligned}
&\hat k, &\epsilon \in (0,\epsilon_1),\\
&a_1\epsilon+b_1, &\epsilon \in [\epsilon_1,\epsilon_2),\\
&a_2\epsilon+b_2, &\epsilon \in [\epsilon_2,c_2\epsilon_2],
\end{aligned}\right.
~~~\text{where} ~~~
\left\{
\begin{aligned}
&a_1=\frac{\tilde k-\hat k}{\epsilon_2-\epsilon_1},\\
&b_1=\hat k-a_1\epsilon_1,\\
&a_2=\frac{(\tilde k-\hat k)(\epsilon_1+\epsilon_2)}{(c_2\epsilon_2-\epsilon_2)^2},\\
&b_2=\tilde k-a_2\epsilon_2;
\end{aligned}\right.$$

 \item \underline{Case 3}, $\epsilon_2/\epsilon_1
 <c_1\le 1/\epsilon_1$ and $c_2={c_1\epsilon_1}/{\epsilon_2}$:
     $$
G_{\mathbf z}(\epsilon)=\left\{\begin{aligned}
&k(\epsilon), &\epsilon \in (0,\epsilon_1),\\
&k(\epsilon_1), &\epsilon \in [\epsilon_1,\epsilon_2),\\
&a\epsilon+b, &\epsilon \in  [\epsilon_2,c_2\epsilon_2],
\end{aligned}\right.
\text{~~~where~~~}
\left\{\begin{aligned}
&a=\frac{2(k(\epsilon_1)\epsilon_1-k)}{(c_2\epsilon_2-\epsilon_2)^2},\\
&b=k(\epsilon_1)-a\epsilon_2;
\end{aligned}\right.$$

  \item \underline{Case 4}, $1<c_1\le \epsilon_2/\epsilon_1$ and $1<c_2\le 1/\epsilon_2$:
     $$
G_{\mathbf z}(\epsilon)=\left\{\begin{aligned}
&k(\epsilon), &\epsilon \in (0,c_1\epsilon_1),\\
&a_1\epsilon+b_1, &\epsilon \in [c_1\epsilon_1, \epsilon_2),\\
&a_2 \epsilon+b_2, &\epsilon \in [\epsilon_2,c_2\epsilon_2],
\end{aligned}\right.
~~~\text{where}~~~
\left\{\begin{aligned}
&a_1=-(c_1\epsilon_1)^{-\xi-1},\\
&b_1=k(c_1\epsilon_1)-a_1c_1\epsilon_1,\\
&a_2=\frac{a_1(\epsilon_2^2-(c_1\epsilon_1)^2)+2(k(c_1\epsilon_1)-k(\epsilon_1))c_1\epsilon_1}{(c_2\epsilon_2-\epsilon_2)^2},\\
&b_2=a_1\epsilon_2+b_1-a_2\epsilon_2;\\
\end{aligned}\right.$$

 \item \underline{Case 5}, $\epsilon_2/\epsilon_1<c_1\le 1/\epsilon_1$ and ${c_1\epsilon_1}/{\epsilon_2}<c_2\le1/\epsilon_2$:
     $$
G_{\mathbf z}(\epsilon)=\left\{\begin{aligned}
&k(\epsilon), &\epsilon \in (0,\epsilon_1),\\
&a_1\epsilon+b_1, & \epsilon \in [\epsilon_1, \epsilon_2),\\
&a_1\epsilon_2+b_1, &\epsilon \in [\epsilon_2,c_1\epsilon_1),\\
&a_2\epsilon+b_2, &\epsilon \in [c_1\epsilon_1,c_2\epsilon_2],
\end{aligned}\right.~~~\text{where}~~~
\left\{
\begin{aligned}
&a_1=\frac{k(\epsilon_1)\epsilon_1-k}{(\epsilon_2-\epsilon_1)(c_1\epsilon_1-1/2(\epsilon_1+\epsilon_2))},\\
&b_1=k(\epsilon_1)-a_1\epsilon_1,\\
&a_2=\frac{2c_1\epsilon_1(a_1\epsilon_2+b_1-k(\epsilon_1))}{(c_1\epsilon_1-c_2\epsilon_2)^2},\\
&b_2=a_1\epsilon_2+b_1-a_2c_1\epsilon_1.
\end{aligned}\right.~~~$$
\end{enumerate}

An illustration of the functions $G_{\mathbf z}$ on $[0, c_2\epsilon_2]$ in Case 2 to Case 5 is presented  in Figure \ref{fig:calibration},
and we omit Case 1 in which $G_{\mathbf z}$ is a constant function on $[0, c_2\epsilon_2]$.
\begin{figure}[t]
\centering
\caption{An illustration of $G_{\mathbf z}$ in cases 2 to 5}
\label{fig:calibration}
\begin{subfigure}[b]{0.47\textwidth}
\centering
\begin{tikzpicture}[scale=.93,description/.style=auto]
\draw[<->] (0,5) -- (0,0) -- (5,0);
\node[below] at (0,0) {$0$};
\node[below] at (1.5,0) {$\epsilon_1$};
\node[below] at (3,0) {$\epsilon_2$};
\node[below] at (4.5,0) {$c_2\epsilon_2$};
\node[left] at (0,5) {$G_{\mathbf z}(\epsilon)$};
\node[left] at (0,4) {$G_{\mathbf z}(\epsilon_1)$};
\draw  (0,4)--(1.5,4);
\draw  (1.5,4)--(3,3);
\draw  (3,3)--(4.5,0.5);
\draw[dash pattern={on 0.84pt off 2.51pt}] (1.5,0)--(1.5,4);
\draw[dash pattern={on 0.84pt off 2.51pt}] (3,0)--(3,3);
\draw[dash pattern={on 0.84pt off 2.51pt}] (4.5,0)--(4.5,0.5);
\end{tikzpicture}

\caption{\footnotesize The function $G_{\mathbf z}$ in Case 2}
\end{subfigure}
~
\begin{subfigure}[b]{0.47\textwidth}
\centering
\begin{tikzpicture}[scale=.93,description/.style=auto]
\draw[<->] (0,5) -- (0,0) -- (5,0);
\node[below] at (0,0) {$0$};
\node[below] at (1.5,0) {$\epsilon_1$};
\node[below] at (3,0) {$\epsilon_2$};
\node[below] at (4.5,0) {$c_2\epsilon_2$};
\node[left] at (0,5) {$G_{\mathbf z}(\epsilon)$};
\node[left] at (0,3) {$\begin{aligned} &G_{\mathbf z}(\epsilon_1)\\ &G_{\mathbf z}(\epsilon_2)\end{aligned}$};
\draw (0.2,5) .. controls (0.5,3.7) and (1,3.5).. (1.5,3);
\draw  (1.5,3)--(3,3);
\draw  (3,3)--(4.5,0.5);
\draw[dash pattern={on 0.84pt off 2.51pt}] (1.5,0)--(1.5,3);
\draw[dash pattern={on 0.84pt off 2.51pt}] (3,0)--(3,3);
\draw[dash pattern={on 0.84pt off 2.51pt}] (4.5,0)--(4.5,0.5);
\draw[dash pattern={on 0.84pt off 2.51pt}] (0,3)--(1.5,3);
\end{tikzpicture}
\caption{\footnotesize The function $G_{\mathbf z}$ in Case 3}
\end{subfigure}
~
\begin{subfigure}[b]{0.47\textwidth}
\centering
\begin{tikzpicture}[scale=.93,description/.style=auto]
\draw[<->] (0,5) -- (0,0) -- (5,0);
\node[below] at (0,0) {$0$};
\node[below] at (1.5,0) {$\epsilon_1$};
\node[below] at (2.5,0) {$c_1\epsilon_1$};
\node[below] at (3.5,0) {$\epsilon_2$};
\node[below] at (4.5,0) {$c_2\epsilon_2$};
\node[left] at (0,5) {$G_{\mathbf z}(\epsilon)$};
\node[left] at (0,3) {$ G_{\mathbf z}(\epsilon_1)$};
\draw (0.2,5) .. controls (1,3.5) and (2,3.1).. (2.5,3);
\draw  (2.5,3)--(3.5,2.7);
\draw  (3.5,2.7)--(4.5,0.5);
\draw[dash pattern={on 0.84pt off 2.51pt}] (1.5,0)--(1.5,3.3);
\draw[dash pattern={on 0.84pt off 2.51pt}] (2.5,0)--(2.5,3);
\draw[dash pattern={on 0.84pt off 2.51pt}] (3.5,0)--(3.5,2.7);
\draw[dash pattern={on 0.84pt off 2.51pt}] (4.5,0)--(4.5,0.5);
\draw[dash pattern={on 0.84pt off 2.51pt}] (0,3)--(2.5,3);
\end{tikzpicture}
\caption{\footnotesize The function $G_{\mathbf z}$ in Case 4}
\end{subfigure}
~
\begin{subfigure}[b]{0.47\textwidth}
\centering
\begin{tikzpicture}[scale=.93,description/.style=auto]
\draw[<->] (0,5) -- (0,0) -- (5,0);
\node[below] at (0,0) {$0$};
\node[below] at (1.5,0) {$\epsilon_1$};
\node[below] at (2.5,0) {$\epsilon_2$};
\node[below] at (3.5,0) {$c_1\epsilon_1$};
\node[below] at (4.5,0) {$c_2\epsilon_2$};
\node[left] at (0,5) {$G_{\mathbf z}(\epsilon)$};
\node[left] at (0,2) {$\begin{aligned} &G_{\mathbf z}(c_1\epsilon_1)\\ &G_{\mathbf z}(\epsilon_2)\end{aligned}$};
\draw (0.2,5) .. controls (0.5,3.5) and (1,3.1).. (1.5,3);
\draw  (1.5,3)--(2.5,2);
\draw  (2.5,2)--(3.5,2);
\draw  (3.5,2)--(4.5,0.5);
\draw[dash pattern={on 0.84pt off 2.51pt}] (1.5,0)--(1.5,3);
\draw[dash pattern={on 0.84pt off 2.51pt}] (2.5,0)--(2.5,2);
\draw[dash pattern={on 0.84pt off 2.51pt}] (3.5,0)--(3.5,2);
\draw[dash pattern={on 0.84pt off 2.51pt}] (4.5,0)--(4.5,0.5);
\draw[dash pattern={on 0.84pt off 2.51pt}] (0,2)--(2.5,2);
\end{tikzpicture}
\caption{ \footnotesize The function $G_{\mathbf z}$ in Case 5}
\end{subfigure}
\end{figure}
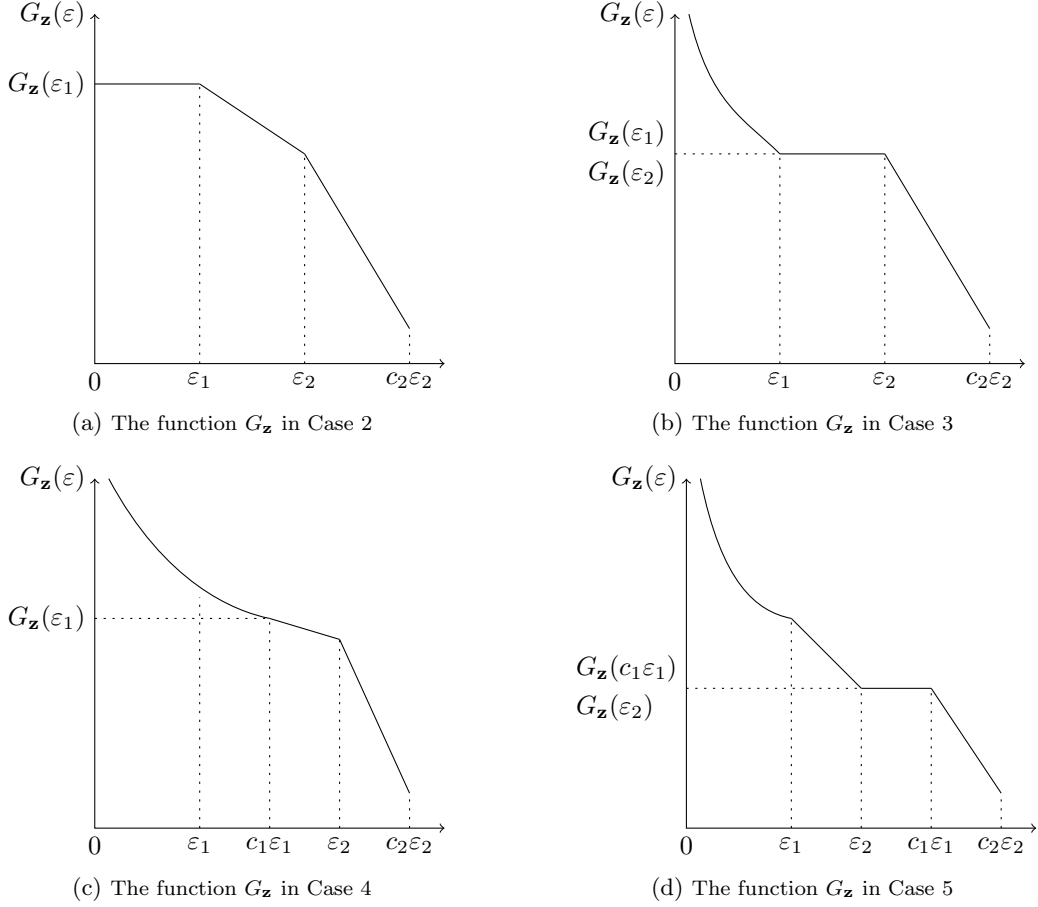

\begin{theorem}\label{thm:two_points}
For $\mathbf z=(\epsilon_1,c_1,\epsilon_2,c_2) \in \Delta$, the random variable $X$ with a continuous quantile function given by $t\mapsto \VaR_{t}(X) =G_{\mathbf z}(t)$  satisfies $\Pi_{X}({\epsilon_1})=c_1$ and $\Pi_{X}(\epsilon_2)=c_2$.
\end{theorem}

\begin{remark}\label{rem:cannot}
As we can see from Figure \ref{fig:calibration}, some parts of the calibrated quantile function may be flat, corresponding to the existence of  atoms in the distribution.
This may be considered as undesirable from a modeling perspective, and indeed it  is forced by the boundary cases of $(\epsilon_1,c_1,\epsilon_2,c_2) \in \Delta$ in Figure \ref{fig:region}. The flat parts in Cases 1 to 3 are necessary due to Propositions \ref{lem:low_bound}.  
On the other hand, the flat part  in Case 5 can be replaced by a strictly decreasing function. For instance, we can replace the flat part with a strictly decreasing linear segment as long as $c_2$ satisfies the bounds shown in Propositions  \ref{pro:c_2_bound} in Appendix \ref{app:proof_sec2}. Another way is to set $\VaR_\epsilon(X)$ as $k(\epsilon)$ for $\epsilon \in (0, c_1\epsilon_1)$ if $c_2\le \left(c_1\epsilon_1(\epsilon_1^{-\xi}-(c_1\epsilon_1)^{-\xi})\right)/\left(\epsilon_2^{-\xi}-(c_1\epsilon_1)^{-\xi}\right)$, and this choice is applied  in the numerical examples in the Introduction and Section \ref{sec:6}. The interested reader can see Propositions \ref{lem:low_bound} and \ref{pro:c_2_bound} in Appendix \ref{app:proof_sec2}, where we show that a strictly decreasing quantile function cannot attain the boundary cases $(\epsilon_1,c_1,\epsilon_2,c_2) $, and hence the flat parts are necessary to include and unify these cases.
\end{remark}

We can easily get the distribution of $X$ from $\VaR_{\epsilon}(X)$. As the PELVE is scale-location invariant, we can scale or move the distribution we get to match more information. For example, if $\VaR_{\epsilon_1}(X)$ and $\VaR_{\epsilon_2}(X)$ are given, we can choose two constants $\lambda$ and $\mu$ such that $\lambda X+\mu$  matches the specified VaR values.
In a similar spirit, the calibration problem can be extended to calibrate the distributions from some given $\ES$ and $\VaR$ values. The two points calibration problem can be regarded as given two $\ES$ and $\VaR$ values.  Calibrating from only $\ES$ or $\VaR$ would be easy. However, the choices of $\ES$ values will also be limited by $\VaR$ values if we consider them at the same time, which is the same as the choice of $c_1, c_2$ as we discussed in this section.

\subsection{Calibration from an n-point constraint}

As we see above, the PELVE calibration problem is quite technical even when only two points on the PELVE curve are given. By extending the constraint to more than two points, the problem will in general become much more complicated. We briefly discuss this problem in this section.

 For the $n$-point constraint problem, we first need to figure out the admissible set for $(\epsilon_i,c_i)_{i=1,\dots, n}$. By Lemma \ref{lem:bound_1}, the admissible set for the  $n$-point calibration problem is a subset of $$\{(\epsilon_i, c_i)_{i=1, \dots, n}: 0<\epsilon_1<\dots< \epsilon_n< 1,~ c_1,\dots,c_n\ge 1, ~0<c_1\epsilon_1 \le \dots\le c_n\epsilon_n\le 1\}.$$
 However, it is not clear whether each point in the above set  is admissible. There are other constraints for the admissible points such as  Proposition \ref{pro:c_2_bound}. Once the admissible set is determined, we need to divide the admissible set according to the position of $\epsilon_i$ and  $c_i\epsilon_i$, $i=1,\dots, n$. Furthermore, the case $c_i=1$ and $c_i\epsilon_i=c_j\epsilon_j$ for $i,j =1, \dots, n$ need special attention as Cases 1, 2 and 3 in the two-point constraint problem. For instance, in the three-point constraint problem, we need to discuss over 10 separate cases.
 
Below, we only discuss some special cases of $(\epsilon_i,c_i)_{i=1,\dots, n}$. First, 
if  $c_n=1$,  then the problem becomes trivial, as the calibrated quantile functions satisfy $\VaR_t(X)=\hat k$ for some $\hat k \in \R$ in $[0,c_n \epsilon_n]$.

For the case $c_k\epsilon_k>\epsilon_k\ge c_{k-1}\epsilon_{k-1}$ for $k=3, \dots, n$, we can set the calibrated quantile function in $(0, c_{n}\epsilon_{n}]$ recursively. 
 This is because such a configuration of $(\epsilon_i,c_i)_{i=1,\dots, n}$ allows for separation of the constraints, in the sense that we can adjust the values of $\VaR_{t}$ for $ t\in [\epsilon_k,c_k\epsilon_k]$ to match PELVE at $\epsilon_{k}$ without disturbing  $\VaR_{t}$ for $t\le c_{k-1}\epsilon_{k-1}$.
  Let $\VaR^{k}_t(X)$ be the calibrated quantile function from the $k$-point constraint problem for $k=2,\dots,n$ where $\VaR^2_t(X)$  follows Theorem \ref{thm:two_points}. The calibrated quantile function for the $n$-point constraint problem is 
$$\VaR^{k}_t(X)=\left\{\begin{aligned}
&\VaR^{k-1}_t(X), ~ &t \in [0, c_{k-1}\epsilon_{k-1}],\\
&a_{k-1}t+b_{k-1}, ~ &t \in (c_{k-1}\epsilon_{k-1}, \epsilon_{k}],\\
&a_kt+b_k,~ &t \in (\epsilon_{k}, c_{k}\epsilon_{k}],
\end{aligned}\right.~~ \mbox{where}~~ \left\{
\begin{aligned}
&a_{k}=\frac{a_{k-1}(\epsilon_k^2+c_{k-1}^2\epsilon_{k-1}^2-2c_{k-1}\epsilon_{k-1}^2)}{(c_{k}\epsilon_{k}-\epsilon_{k})^2},\\
&b_k=a_{k-1}\epsilon_k+b_{k-1}-a_k\epsilon_{k}.
\end{aligned}\right.$$
In particular, for $n=3$, and assuming  $c_3\epsilon_3>\epsilon_3\ge c_2\epsilon_2$, the calibrated function is given by, with  $\mathbf z=(\epsilon_1, c_1\epsilon_1,\epsilon_2,c_2\epsilon_2)\in \Delta$,
$$\VaR_t(X)=\left\{\begin{aligned}
&G_{\mathbf z} (t), ~ &t \in [0, c_2\epsilon_2],\\
&a_2t+b_2, ~ &t \in (c_2\epsilon_2, \epsilon_3],\\
&a_3t +b_3,~ &t \in (\epsilon_3, c_3\epsilon_3],
\end{aligned}\right.~~ \mbox{where}~~\left\{
\begin{aligned}
&a_2=\frac{a_1(\epsilon_2^2-(c_1\epsilon_1)^2)+2(k(c_1\epsilon_1)-k(\epsilon_1))c_1\epsilon_1}{(c_2\epsilon_2-\epsilon_2)^2},\\
&b_2=-(c_1\epsilon_1)^{-\xi-1}(\epsilon_2-c_1\epsilon_1)-k(c_1\epsilon_1)\\
&a_{3}=\frac{a_{2}(\epsilon_3^2+c_{2}^2\epsilon_{2}^2-2c_{2}\epsilon_{2}^2)}{(c_{3}\epsilon_{3}-\epsilon_{3})^2},\\
&b_3=a_{2}\epsilon_3+b_{2}-a_3\epsilon_{3}.
\end{aligned}\right.$$  In Figure \ref{fig:three-point}, we show the calibrated quantile function for the case $(\epsilon_1, \epsilon_2, \epsilon_3)=(0.005, 0.025, 0.1)$ and $(c_1, c_2, c_3)=(4, 3, 2.5)$. Note that the condition $c_2\epsilon_2 \le \epsilon_3$ is needed here.

\begin{figure}[htbp]
   \caption{ Calibrated quantile function when $(\epsilon_1, \epsilon_2, \epsilon_3)=(0.005, 0.025, 0.1)$ and $(c_1, c_2, c_3)=(4, 3, 2.5 )$}\label{fig:three-point}
    \centering
    \includegraphics[scale=0.4]{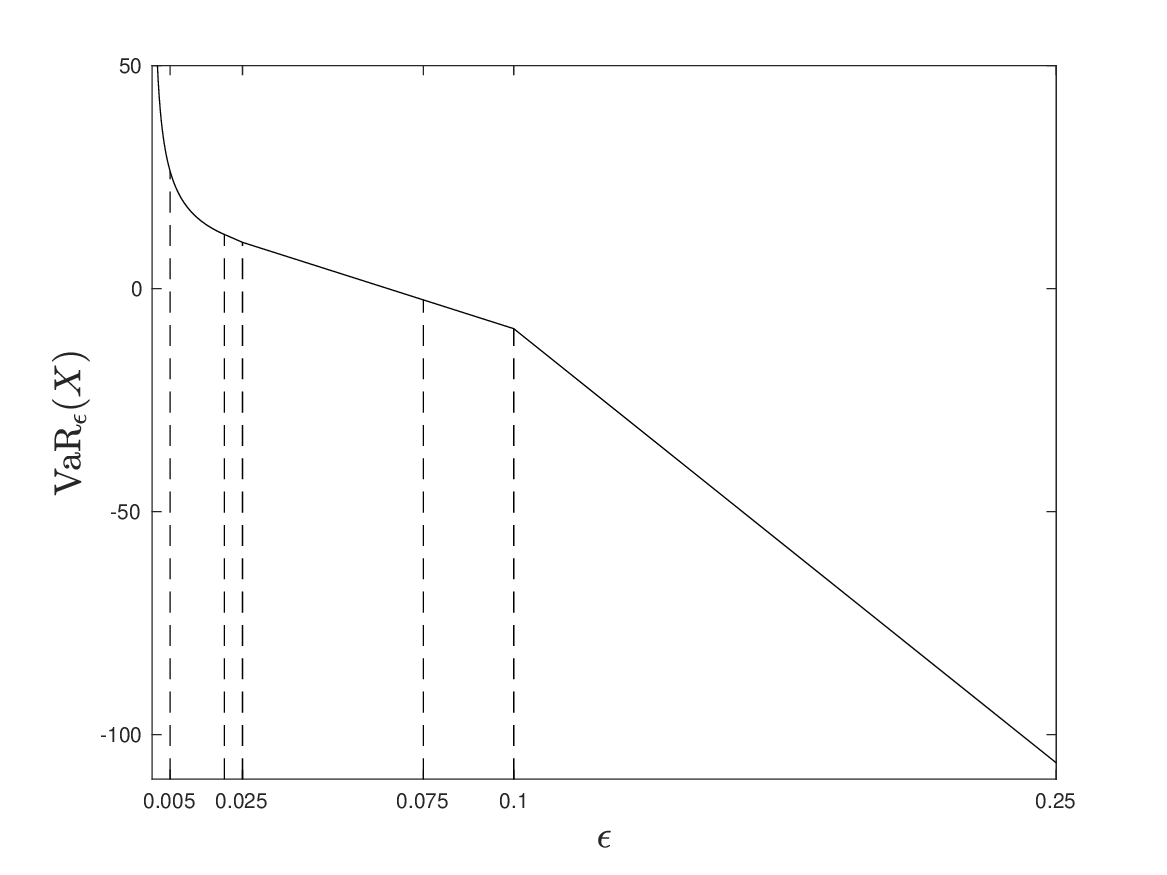}
\end{figure}

Although we cannot solve the $n$-point constraint problem in general, we can instead discuss   calibration  from a given PELVE curve, which is the problem addressed  in the next section.

\section{Calibration from a curve constraint}\label{sec:function}

By the location-scale invariance properties of the PELVE, we know that the solution cannot be unique. Conversely, it would be interesting to ask whether all solutions can be linearly transformed from a particular solution; that is, for a given function $\epsilon \mapsto \Pi(\epsilon)$, whether the set
    $\left\{ X\in\X:\Pi_{X} =\Pi \right\}$
is a location-scale class.
This question, as well as identifying $X$ satisfying $\Pi_{X}=\Pi$, is the main objective of this section.

\subsection{PELVE and dual PELVE}\label{var-es}

First, we note that calibrated distributions from an entire PELVE curve $\epsilon \mapsto \Pi(\epsilon)$ on $ (0,1)$ would be unnatural, because the existence of the PELVE requires $\E[X]\le \VaR_{\epsilon}(X)$ which may not hold  for $\epsilon$ not very small. Thus, the PELVE curve $\Pi_{X}$ does not behave well on  some parts of $(0,1)$.
To address this issue, we  introduce a new notion called the dual PELVE and an integral equation which can help us to calibrate the distribution by differential equations.
The dual PELVE is defined by moving the multiplier  in PELVE from the $\ES$ side to the $\VaR$ side.
\begin{definition}
For $X\in L^{1}$, the dual PELVE function of $X$ at level $\epsilon\in(0,1]$ is defined as
$$
\pi_{X}(\epsilon) =\inf\left\{ d\ge1:\ES_{\epsilon}(X)\le\VaR_{\epsilon/d}(X)\right\} ,~~\epsilon\in(0,1].
$$
\end{definition}
The existence and uniqueness of $\pi_{X}(\epsilon)$ can be shown in the same way as the existence and uniqueness of the PELVE. There are advantages and disadvantages of working with both notions; see \citet[Remark 2]{LW22}. In our context, the main advantage of using the dual PELVE is that $\pi_{X}(\epsilon)$ is finite for all $\epsilon\in(0,1]$, while $\Pi_{X}(\epsilon)$  is finite only when $\E[X]\le\VaR_{\epsilon}(X)$.
 
Note that for $X$ with a discontinuous quantile function,
there may not exist $d $ such that $\ES_{\epsilon}(X)= \VaR_{\epsilon/d}(X)$.
In order to guarantee the  above equivalence, we make the following assumption for the quantile function, represented by general function $f$.
\begin{assumption}\label{ass:1}
The function $f$ is strictly decreasing and continuous, and $\int_{0}^{1}\left|f(s)\right|\mathrm{d}s<\infty$.
\end{assumption}
Let $\X$ be the set of $X\in L^1$ with quantile function satisfying Assumption \ref{ass:1}. The  requirement that  the quantile function of $X$ is continuous and strictly decreasing is equivalent to that the distribution function is continuous and strictly increasing in $(\essinf(X), \esssup(X))$; see \cite{EH13}.
We limit our discussion to random variables $X \in \X$, which include the most common models in risk management. 
 \begin{proposition}\label{pro:dual}
 For $X$ with quantile function satisfying Assumption \ref{ass:1} and $\epsilon\in (0,1)$, we have $\Pi_{X}(\epsilon/\pi_{X}(\epsilon))=\pi_{X}(\epsilon)$   and  $\pi_{X}(\Pi_{X}(\epsilon)\epsilon)=\Pi_{X}(\epsilon)$  if $\E[X]\le\VaR_{\epsilon}(X)$. 
Furthermore,
$\pi_{X}(\epsilon)$ is the unique solution $d\ge1$ to the equation
$$
    \ES_{\epsilon}(X)=\VaR_{\epsilon/d}(X).
$$
\end{proposition}
It is straightforward to verify Proposition \ref{pro:dual}. By Proposition \ref{pro:dual}, we can calibrate the distribution functions from dual PELVE instead of PELVE,  and the calibrated distributions should satisfy the equation $\ES_{\epsilon}(X)=\VaR_{\epsilon/d}(X)$.

\subsection{An integral equation associated with dual PELVE}\label{abstract}
In order to calibrate distributions from the dual PELVE, we can equivalently focus on quantile functions. Let us consider $X\in \X$  and  $f(s)=\VaR_{s}(X)$.  Then, solving $\pi_{X}(\epsilon)$ is the same as solving $z$ in following equation:
\begin{equation}\label{eq:abstract}
    \int_{0}^{y}f\left(s\right)\mathrm{d}s=yf\left(zy\right)
\end{equation}
for   $y=\epsilon$. The solution is $z=1/\pi_{X}(y)$.
 As $f(s)=\VaR_{s}(X)$,  $f$ satisfies Assumption  \ref{ass:1}.
Denote by $\mathcal{C}$  the set of all $f$ satisfying Assumption \ref{ass:1}.
For any $f\in \mathcal{C}$, the existence of the solution $z$ is guaranteed by the mean-value theorem and its uniqueness is obvious.
For $y\in(0,1]$, let $z_{f}(y)$ be the solution to \eqref{eq:abstract} associated with $f$. Clearly, $z_{f}(y)\le1$ and $y\mapsto yz_{f}(y)$ is strictly increasing.
This is similar to Lemma \ref{lem:bound_1} for the two-point case.
Obviously, $z_{f}(y)$ is also location-scale invariant under linear transformation on $f \in \mathcal{C}$. That is, $z_{\lambda f+b}=z_{f}$ for $\lambda>0$ and $b\in\R$. Furthermore, $z_f$ is continuous as $f$ is continuous and strictly decreasing.
The next proposition is a simple connection between $z_f$ and $\pi_X$.

\begin{proposition}\label{prop:abstract}
     For any  $f$ satisfying Assumption \ref{ass:1}, $X=f(U)$ for some $U \sim \mathrm{U}(0,1)$ has the dual PELVE $\pi_X(y)=1/z_f(y)$ for all $y \in (0,1)$ where $z_f$ is solution to \eqref{eq:abstract}.
      For  $X$ with quantile function satisfying Assumption \ref{ass:1}, there exists  $f$ satisfying Assumption \ref{ass:1} such that $X=f(U)$ for some $U \sim \mathrm{U}(0,1)$ and the solution to \eqref{eq:abstract} is $z_f(y)=1/\pi_X(y)$ for all $y \in (0,1)$.
\end{proposition}
\begin{proof}
  For any $f$ satisfying Assumption \ref{ass:1}, let $F(x)=1-f^{-1}(x)$.
    Hence, $F$ is a continuous and strictly increasing distribution function 
    and $F^{-1}(s)=f(1-s)$ for $s \in (0,1)$.
    Let $U\sim \mathrm{U}(0,1)$ and $X=F^{-1}(U)=f(1-U)$.
    Then $X \in \X$ and $X\sim F$.
    As $F^{-1}(1-s)=f(s)$, we have $\pi_X(y)=1/z_f(y)$.
    Take $U'=1-U$. We have $X=f(U')$ and $U'\sim \mathrm{U}(0,1)$.

    For $X\in \X$, let $f(s)=\VaR_{s}(X)$.
    Then, we have $z_f(y)=1/\pi_X(y)$ for $y \in (0,1]$.
    Furthermore, we have $F^{-1}(s)=f(1-s)$.
    Therefore, there exists $U \sim \mathrm{U}(0,1)$ such that $X=f(1-U)$.
    Let $U'=1-U$. Then, we have $X=f(U')$ and $U'\sim \mathrm{U}(0,1)$.
\end{proof}

Proposition \ref{prop:abstract} allows us to study  $z$ instead of $\pi$ for the calibration problem. The integral equation \eqref{eq:abstract} can be very helpful in characterizing the distribution from the dual PELVE.

Some examples of $\pi_X$ and $z_{f}$ are listed in Table \ref{ex:zf}, which is corresponding to the PELVE presented in Table \ref{constant-PELVE}.

\begin{table}[htbp]
\def\arraystretch{1.4}
    \centering{}
    \caption{Example of $\pi_X$ and $z_{f}$}\label{ex:zf} %
    \begin{tabular}{m{1.5cm}<{\centering} | m{3cm}<{\centering}|m{5cm}<{\centering}|m{2.8cm}<{\centering}}
    \hline \hline
    $X$&  $\pi_X(\epsilon)$ & $f$ & $z_f$ \\
    \hline
    $\mathrm{U}(0,1)$ & $\pi_X(\epsilon)=2$ & $f(x)=1-x$ & $z_{f}(y)=1/2$ \\
    \hline
    $\mathrm{Exp}(\lambda)$& $\pi_X(\epsilon)=e$ &   $f(x)=-\log(x)/\lambda$&$z_{f}(y)=1/e$ \\
    \hline
    $\mathrm{GPD}(\xi)$ & $\pi_X(\epsilon)=(1-\xi)^{-\frac{1}{\xi}}$ & $f(x)=\left\{\begin{aligned}
    &1/\xi\left(x^{-\xi}-1\right) &\xi\neq 0\\
    &-\log(x) &\xi=0 \\\end{aligned}\right.$ &$z_{f}=(1-\xi)^{\frac{1}{\xi}}$\\
    \hline \hline
    \end{tabular}
\end{table}

 For a given  dual PELVE curve $\pi$, we find the solution to the integral equation by the following steps.
\begin{enumerate}
    \item Let $z(y)=\frac{1}{\pi(y)}$ for all $y\in(0,1]$.
    \item Find $f\in \mathcal{C}$  that satisfies  $\int_{0}^{y}f(s)\d s=yf\left(z(y)y\right)$ for all $y\in(0,1]$.
    \item By Proposition \ref{prop:abstract}, $X=f(U)$ for some $U \sim \mathrm{U}(0,1)$ will have the given dual PELVE $\pi$.
\end{enumerate}

Therefore, we will focus on characterizing $f$ from a given $z:(0,1]\to(0,1]$ below. Generally, it is hard to characterize $f$ explicitly.
We first formulate the problem as an advanced differential equation, which helps us to find solutions.

\subsection{Advanced differential equations}

In this section, we show that the main objective \eqref{eq:abstract} can be represented by a differential equation. The use of differential equations in computing risk measures has not been actively developed. The only paper we know is \cite{B17} which addresses a different problem.

Let us recall the integral equation \eqref{eq:abstract} from Section \ref{abstract}.
For a function $f\in\mathcal{C}$, we solve the function $z_{f}:(0,1)\to\R$ from \eqref{eq:abstract}.  We represent \eqref{eq:abstract} by an advanced differential equation using the following steps.
\begin{enumerate}
    \item Let $\omega_{f}(y)=yz_{f}\left(y\right)$. It is easy to see that $z_{f}\left(y\right)\le1$. Hence, $\omega_{f}$  is strictly increasing and continuous 
    on $(0,1]$ and $\omega_{f}(y)\le y$.
    \item Let $u_{f}$ be the inverse function of $\omega_{f}$.
We have that $u_{f}:(0,z_{f}(1)]\mapsto (0,1]$ is a continuous and  strictly increasing function
and  $u_{f}\left(w\right)\ge w$.
\item Replacing $y$ with $u_f(w)$ in \eqref{eq:abstract}, we have 
$ f(w)=\int_{0}^{u_f(w)} f(w) \d s/u_f(w).$
\item Assume $u_f$ is continuously differentiable. It is clear that $f$ is continuously differentiable on $(0,z_{f}(1))$.  Hence, we can represent \eqref{eq:abstract} by the following advanced differential equation
$$
f'\left(w\right)+\frac{u_{f}'\left(w\right)}{u_{f}\left(w\right)}\left(f\left(w\right)-f\left(u_{f}\left(w\right)\right)\right)=0.
$$
\end{enumerate} 

For a given function $z:(0,1]\to\R$, let $u=\omega^{-1}$ such that $\omega(y)=yz(y)$ for $y \in (0,1]$.  Then, we solve the function $f$  by the following differential equation
\begin{equation}
    f'\left(w\right)+\frac{u'\left(w\right)}{u\left(w\right)}f\left(w\right)-\frac{u'\left(w\right)}{u\left(w\right)}f\left(u\left(w\right)\right)=0.\label{eq:DDE}
\end{equation}
If $z=1/\pi_X$ for some $X \in \X$, then  $u$ is a strictly increasing and continuous function such that 
$u\left(w\right)\ge w$. Furthermore, if $z$ is continuously differentiable, then we can characterize all $X \in \X$ with $\pi_X=1/z$ by \eqref{eq:DDE}.
As $u'\left(w\right)/u\left(w\right)\ge 0$ and $u(w)\ge w$, \eqref{eq:DDE} is a linear advanced differential equation which is well studied in the literature. In \citet{BB11},
it is shown that  there exists a non-oscillatory solution for \eqref{eq:DDE}. 

\subsection{The constant PELVE curve}\label{sec:constant}


We first solve the case that $z(y)=c$ for all $y\in(0,1]$ and some constant $c \in\left(0,1\right)$. As we can see from Table \ref{ex:zf}, the power function and logarithm function have constant $z_f$. If $f(x)=\lambda x^{\alpha}+b$ for $\alpha>-1$, we can see that  $ (\alpha+1)^{-1/\alpha}=c$.
In this section, we can characterize all the other  solutions which can not be expressed as a linear transformation of the power function. That is, we will see that the set
$$\left\{ f\in \mathcal{C}:z_{f}(y)=z(y), ~y \in (0,1]\right\}$$
is not a location-scale class. Hence, we can answer the question at the beginning of the section; that is, in the case the PELVE is a constant, the set  $\left\{ X\in\X:\Pi_{X} =c \right\}$ is not a location-scale class.

\begin{theorem}\label{thm:constant}
For $c \in (0,1)$,  any $X$ with quantile function satisfying Assumption \ref{ass:1} and  $\pi_{X}(\epsilon)=1/c$ for $\epsilon \in (0,1)$ can be written as $X=f(U)$ for some $U \sim \mathrm{U}(0,1)$ and  $f$ satisfying Assumption \ref{ass:1}. Furthermore, such $f$ has the form
$$f\left(y\right) =C_{1}+C_{2}y^{\alpha}+O\left(y^{\zeta}\right),$$
where $\alpha$ is the root of $(\alpha+1)^{-1/\alpha}=c$, $\zeta>\max\{0,\alpha\}$, $C_1, C_2\in \R$, $C_2\alpha<0$ and  $O (y^{\zeta} )$ is a function such that $\lim \sup_{y \to 0} O (y^{\zeta} )/y^{\zeta}$ is a constant.
\end{theorem}

The proof of Theorem \ref{thm:constant} is provided in Appendix \ref{app:proof_sec5}. As we can see, Theorem \ref{thm:constant}  characterizes all $X \in \X$ such that $\pi_X(\epsilon)=1/c$.
If $c \in (0,1/e)$, $\alpha$ is negative. As $\zeta>0$, we can see that $X=f(U)$ is regularly varying of index $\alpha$.
Hence, one can then consider the Pareto distribution with survival function $S(x)=x^{\alpha}$ as a representative solution for the tail behavior. An open question is that, in the general case that the PELVE is not necessarily constant, whether all the solutions behave similarly regarding their tail behavior.

Another interesting implication of the theorem and its proof is that one can give a non-trivial solution for $z$ is a constant.
\begin{example}\label{f}
For $c \in (0,1)$, let $(\theta,\eta)$ be a solution of
    $$
    \begin{cases}
    c\log c=-\frac{\eta\exp\left(-\frac{\eta}{\tan\left(\eta\right)}\right)}{\sin\left(\eta\right)},\\
    \theta=-\frac{\eta}{\tan\left(\eta\right)}.
    \end{cases}
    $$
Then, the function $f$, given by
    \begin{equation}\label{eq:sol_f}
    f(y)=C_1+C_2y^{\alpha}+C_3y^{\zeta}\sin(-\sigma\log(y)),\quad 0<y<1,
    \end{equation}
satisfies $\int_{0}^{y}f(s)\d s=yf(cy)$ and  Assumption \ref{ass:1},
where $\alpha$ solves $(\alpha+1)^{-1/\alpha}=c$,  $\zeta=\theta/\log c-1$, $\sigma=-\eta/\log c$, $C_2$ is a constant such that $C_2\alpha<0$
and $0<C_3<-C_2\alpha/(\zeta+|\sigma|)$.
\end{example}

If we take $C_3=0$, we get the simplest power function for $z(x)=c$.
If $C_3\neq0$, the solution \eqref{eq:sol_f} is not a linear transformation of the power function solution.

Let us look at the example where $\pi(\epsilon)=2$ for all $\epsilon\in(0,1]$, which means $z(y)=1/2$ for $y\in(0,1]$.
As we have seen in Table \ref{ex:zf},  $f(y)=1-y$ can be a solution that leads to $X \sim \mathrm{U}(0,1)$.
Furthermore, according to Example \ref{f}, we can have another solution
$$f(y)  =1-y^{\alpha}+Cy^{\zeta}\sin(-\sigma\log(y)),$$
where $\alpha=1$, $C=0.05096$, $\zeta=4.0184$ and $\sigma=-15.4090$.
In the left of Figure \ref{fig:Non-unique-answer}, we have depicted the two solutions for $f$.
We can see they are quite different when $y$ goes to 1.
In the right of Figure \ref{fig:Non-unique-answer}, we numerically calculate $z_{f}$ for $f(y)=1-y^{\alpha}+Cy^{\zeta}\sin(-\sigma\log(y))$. We can see its numerical value is almost 1/2 and the discrepancy is due to limited computational accuracy.

\begin{figure}[htbp]
    \caption{\label{fig:Non-unique-answer}Non-unique calibrated functions for $z(y)=1/2$.}
    \centering{}\includegraphics[scale=0.38]{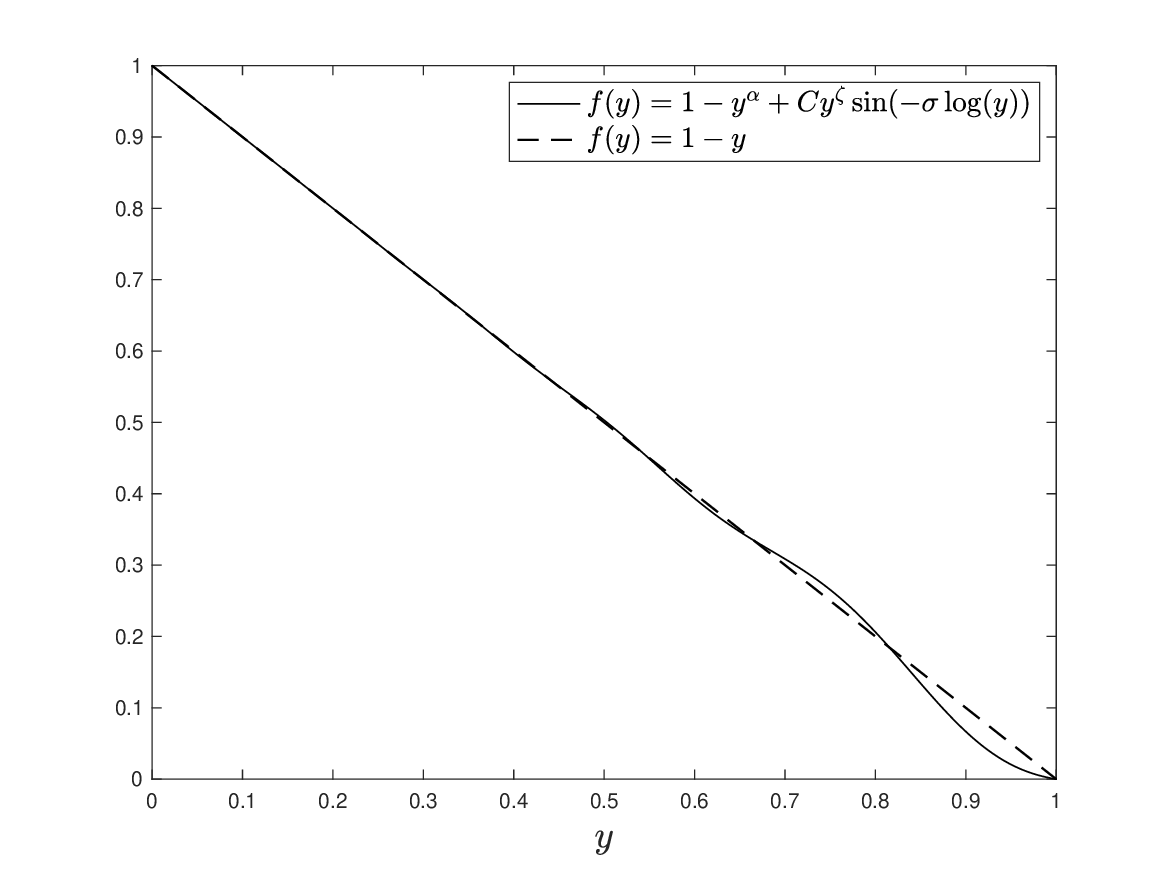} \includegraphics[scale=0.38]{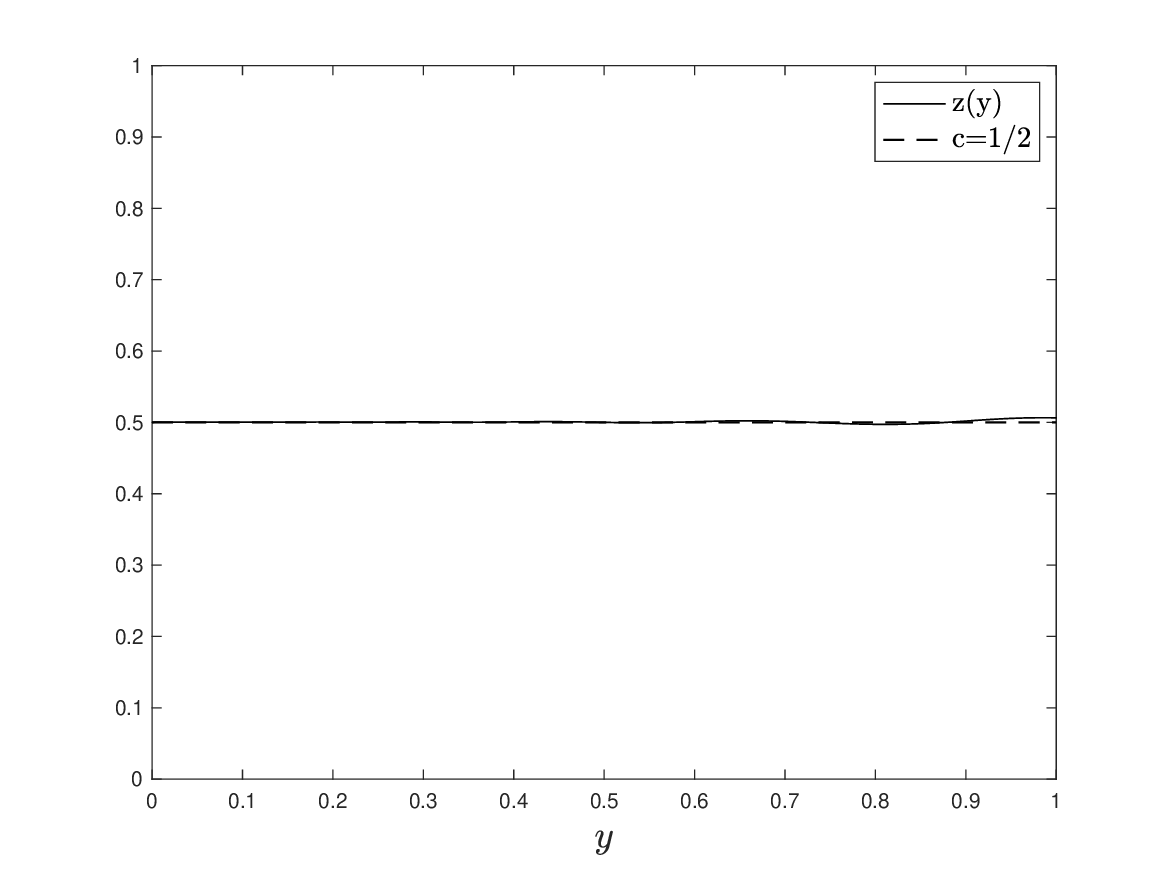}
\end{figure}

By letting $X=f(U)$, we get  $\pi_{X}(\epsilon)=2$ for all $\epsilon\in(0,1]$ and such $X$ does not follow the uniform distribution.

\subsection{A numerical method}
In general, it is hard to get an explicit solution to \eqref{eq:DDE}.
Here we present a numerical method to solve \eqref{eq:DDE}.
Let us introduce the following process.
    \begin{enumerate}
       \item Let $a_{0}=1$, $a_{1}=a$, ...,$a_{n}=u^{-1}\left(a_{n-1}\right)$.
       \item
       For $a \in (0,1)$, let $\xi$ be the solution to  $(1-\xi)^{\frac{1}{\xi}}=a$.
       Let
       \begin{equation}\label{eq:GPD}
       f_0(x)=\left\{\begin{aligned}
    &\frac{1}{\xi}\left(x^{-\xi}-1\right), &\xi\neq 0,\\
    &-\log(x), &\xi=0, \\\end{aligned}\right.
    \end{equation} on $[a,1]$.

       \item We can solve the following ODE on $\left[a_{2},a_{1}\right]$:
            $$ f_{1}'\left(w\right)+\frac{u'\left(w\right)}{u\left(w\right)}f_{1}\left(w\right) =\frac{u'\left(w\right)}{u\left(w\right)}f_{0}\left(u\left(w\right)\right),\quad w\in\left[a_{2},a_{1}\right].
            $$
       \item Now we can repeat step 3 by induction on $\left[a_{n+1},a_{n}\right]$ for $n>1$ by solving
             $$
             f_{n}'\left(w\right)+\frac{u'\left(w\right)}{u\left(w\right)}f_{n}\left(w\right)  =\frac{u'\left(w\right)}{u\left(w\right)}f_{n-1}\left(u\left(w\right)\right),\quad w\in\left[a_{n+1},a_{n}\right].
             $$
       \item In general, the solution for differential equation $\frac{dy}{dx}+P(x)y=Q(x)$ is
             $$
             y  =e^{-\int^{x}P(\lambda)\,\d\lambda}\left[\int^{x}e^{\int^{\lambda}P(\varepsilon)\d\varepsilon}Q(\lambda)\d\lambda+C\right].
             $$
            So, we get the following solution for $f_{n}$:
             $$
             f_{n}\left(w\right) =e^{\int_{w}^{a_{n}}\frac{u'(\lambda)}{u(\lambda)}\d\lambda}\left[f_{n-1}(a_{n})-\int_{w}^{a_{n}}e^{-\int_{\lambda}^{a_{n}}\frac{u'(\varepsilon)}{u(\varepsilon)}\d
             \varepsilon}\frac{u'(\lambda)}{u(\lambda)}f_{n-1}\left(u\left(\lambda\right)\right)\d\lambda\right],~w\in[a_{n+1},a_{n}].\label{fn}
             $$
       \item Finally, let $f=f_{n}$ on $[a_{n+1},a_{n}]$.
    \end{enumerate}
Note that since we start with a strictly decreasing function, then from equation (\ref{eq:DDE}) we have
$$
f'(w)  =\frac{u'\left(w\right)}{u\left(w\right)}\left(f\left(u(w)\right)-f\left(w\right)\right)<0,
$$
so $f$ remains strictly decreasing.

 The solution produced by the numerical method   heavily relies on $f_{0}$.
The equation \eqref{eq:DDE} does not have a unique solution, but the solution from the above process is unique. We set $f_{0}$ as \eqref{eq:GPD} by assuming $z$ can be extended from $(0,1]$ to $\R^{+}$ and set $z(y)=a$ for all $y>1$. We use this assumption for simplification as we can know that \eqref{eq:GPD} satisfies \eqref{eq:DDE} for a constant $z$ from Section \ref{sec:constant}. This choice of $f_0$ is the same as the choice of $k(\epsilon)$ in the two-point calibration problem, and this reflects our subjective view of the importance of the Pareto distribution in risk management.
Especially, when $z(y)=c$ for some constant $c$, we have $u(x)=x/c$.
Therefore, \eqref{fn} gives
$$
f_n(w)=\frac{a_n}{w}\left[f_{n-1}(a_n)-\int_{w}^{a_n}\frac{1}{a_n}f_{n-1}\left(\frac{\lambda}{c}\right)\d\lambda\right].
$$
If we set $f_{0}$ as \eqref{eq:GPD}, we can have $f_{1}$ also in the form of \eqref{eq:GPD}.
Then, it is obvious that $f_n$ is also in the form of \eqref{eq:GPD}.
Therefore, the numerical method gives the simplest power function or logarithm function when $z(y)$ is a constant on $(0,1]$ as Table \ref{ex:zf}, which  leads to the generalized Pareto distribution for $X$.

\label{sec:44}
\subsection{Numerical calibrated quantile function}

Now let us explore the method in Section \ref{sec:44} with simulation. Here we present the results for a few cases. In Figures \ref{linear} to \ref{exp}, we compare the solution from the numerical method with the standard formula in Table \ref{ex:zf} in the left panel, and compare $\int_{0}^{y}f(s)\d s$ with $yf(z(y)y)$ to validate  the equation \eqref{eq:abstract} in the right panel.

We first try some examples where $z$ is constant as shown in Table \ref{ex:zf}, i.e. $z(x)=1/2$ (Figure \ref{linear}), $z(x)=1/e$ (Figure \ref{log}) and $z(x)=0.9^{10}$ (Figure \ref{power}). For Figure \ref{linear} to \ref{power}, we can see that the numerical method provides exactly the same function $f$ as Table \ref{ex:zf}.

In Figure \ref{exp}, we check the case  $z(x)=\log\left(x/(1-e^{-x})\right)/x$. The function $f(x)=e^{-x}$ satisfies \eqref{eq:abstract}. We can see that the solution from the numerical method is close to  a function of the form $f(x)=\lambda e^{-x}+b$, which is known to satisfy the integral equation.

        \begin{figure}[h]
        \caption{Calibrated function and validation  for $z(x)=1/2$}\centering\includegraphics[scale=0.39]{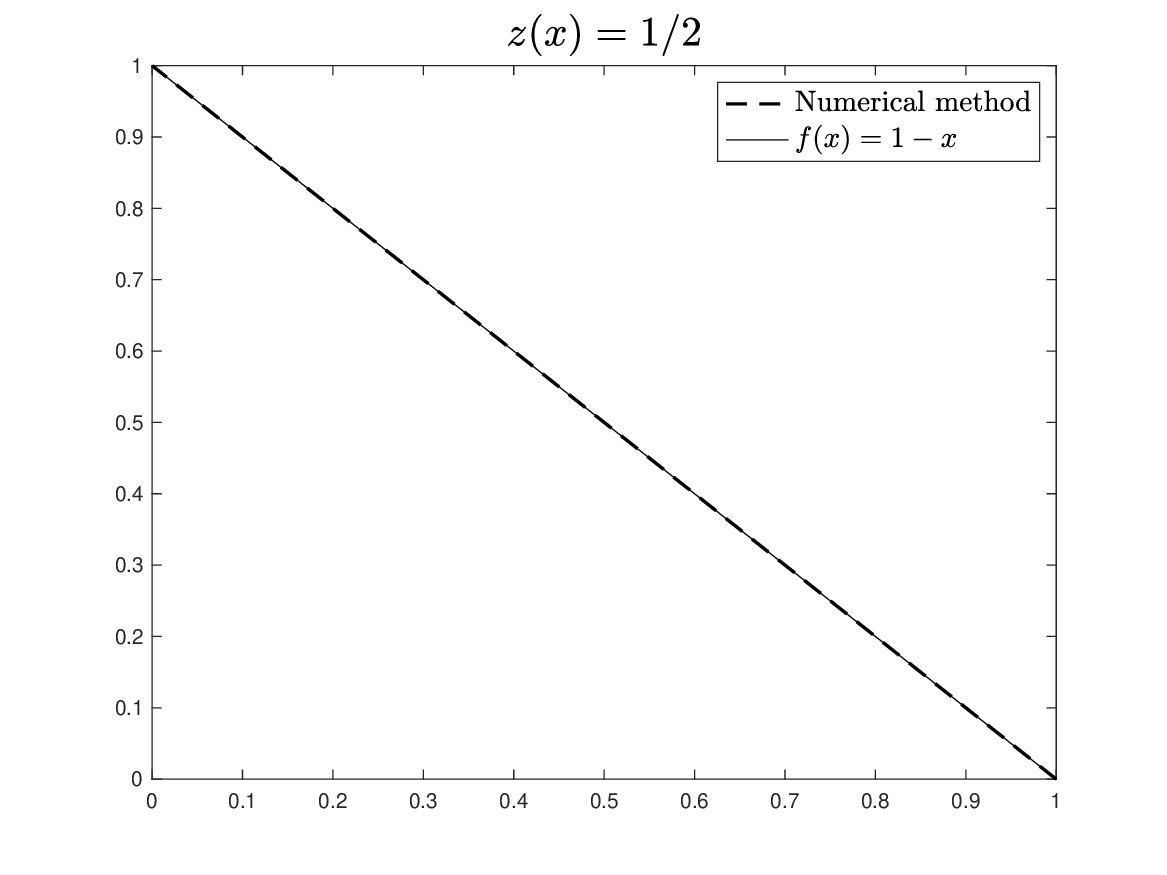}
        \includegraphics[scale=0.39]{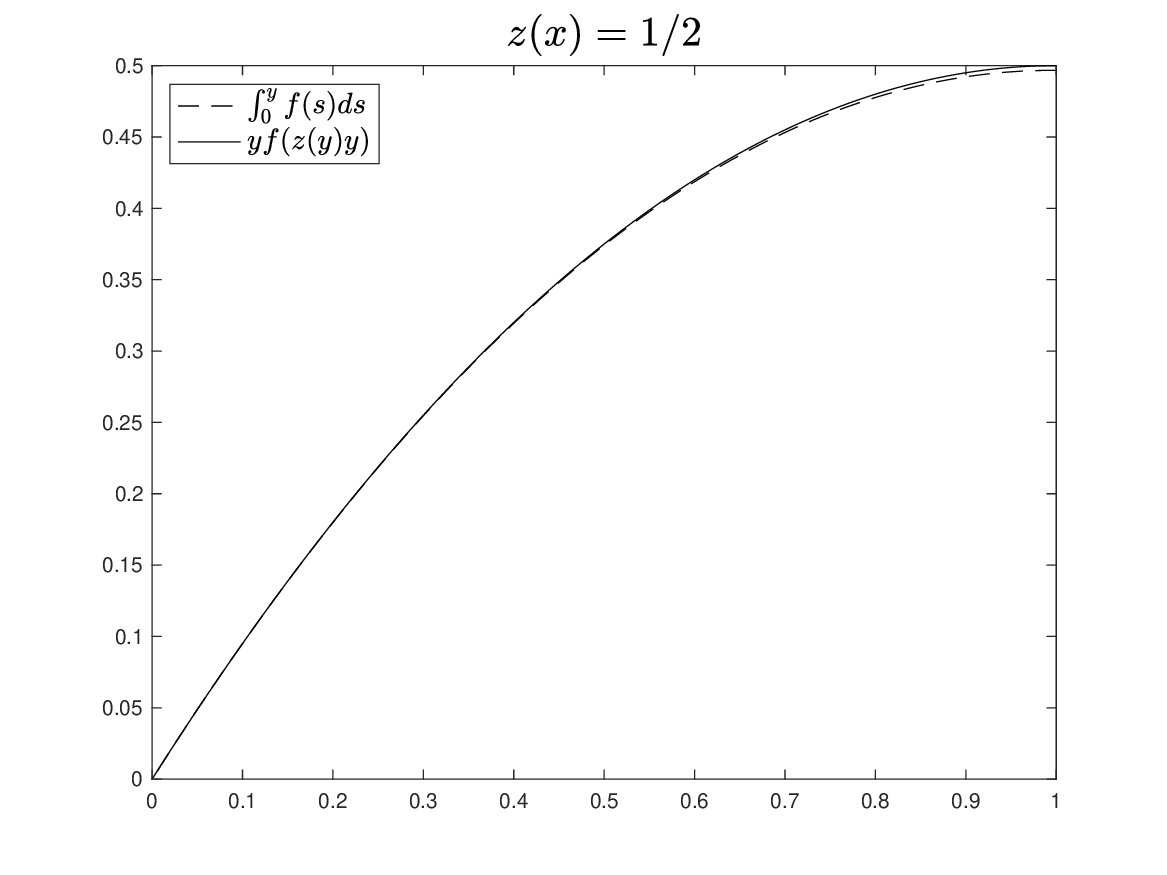}
        \label{linear}
        \end{figure}

     \begin{figure}[h]
     \caption{Calibrated function and validation for $z(x)=1/e$}
    \centering\includegraphics[scale=0.39]{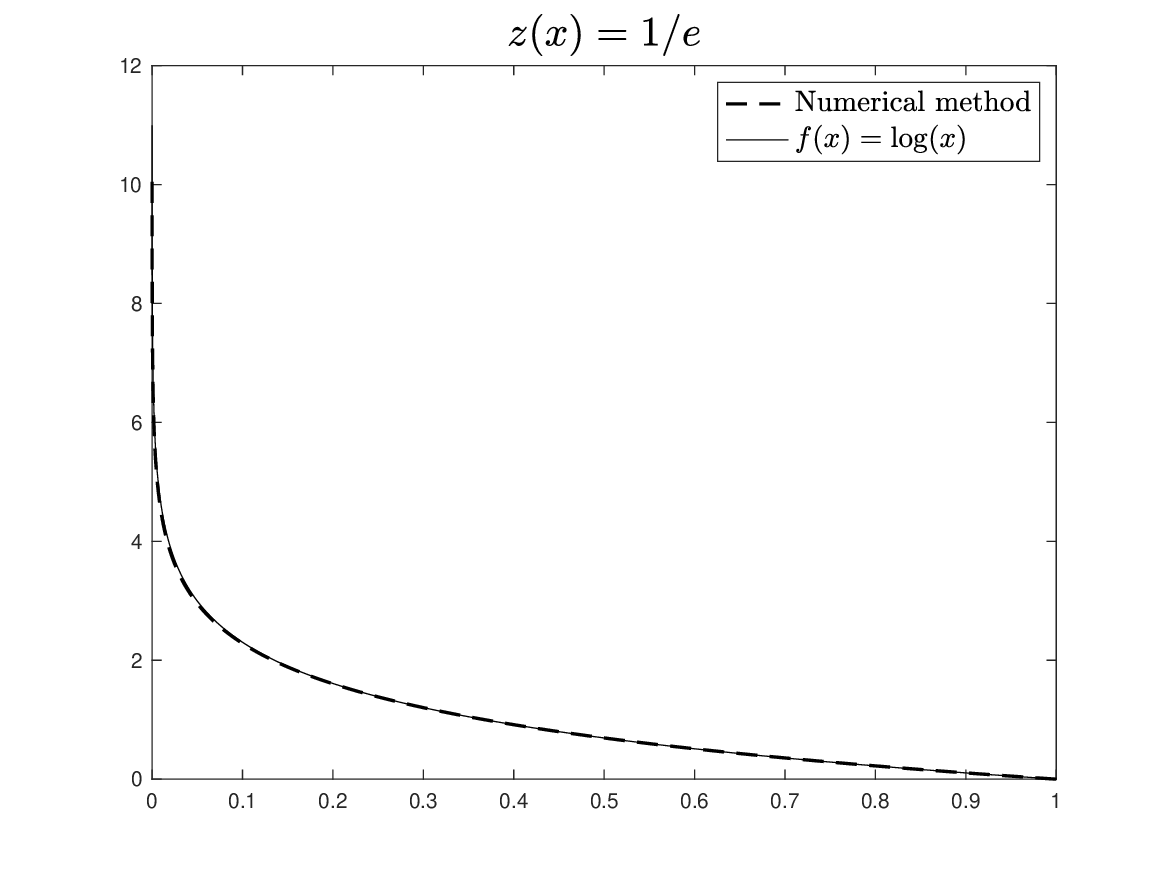}
    \includegraphics[scale=0.39]{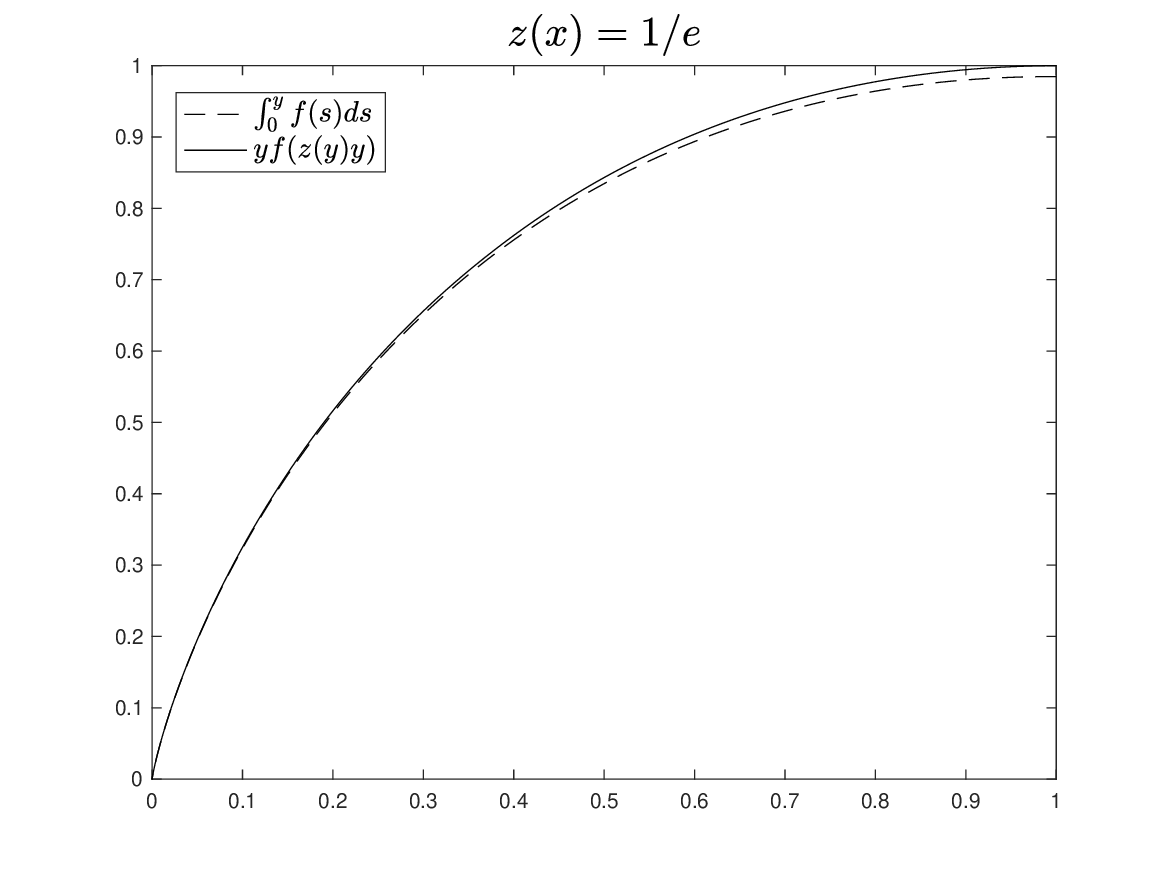}
    \label{log}
    \end{figure}

    \begin{figure}[h]
      \caption{Calibrated function and validation for $z(x)=0.9^{10}$}
    \label{power}
    \centering \includegraphics[scale=0.39]{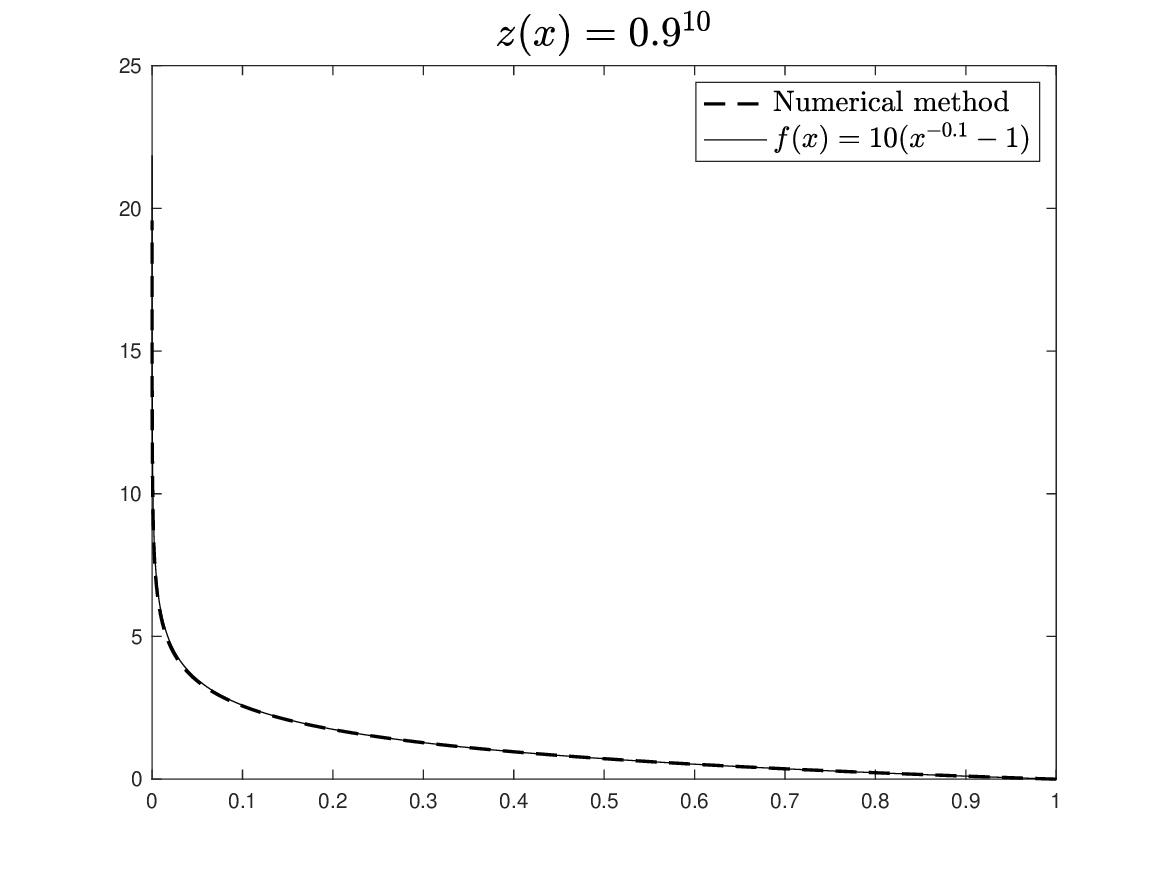}
    \includegraphics[scale=0.39]{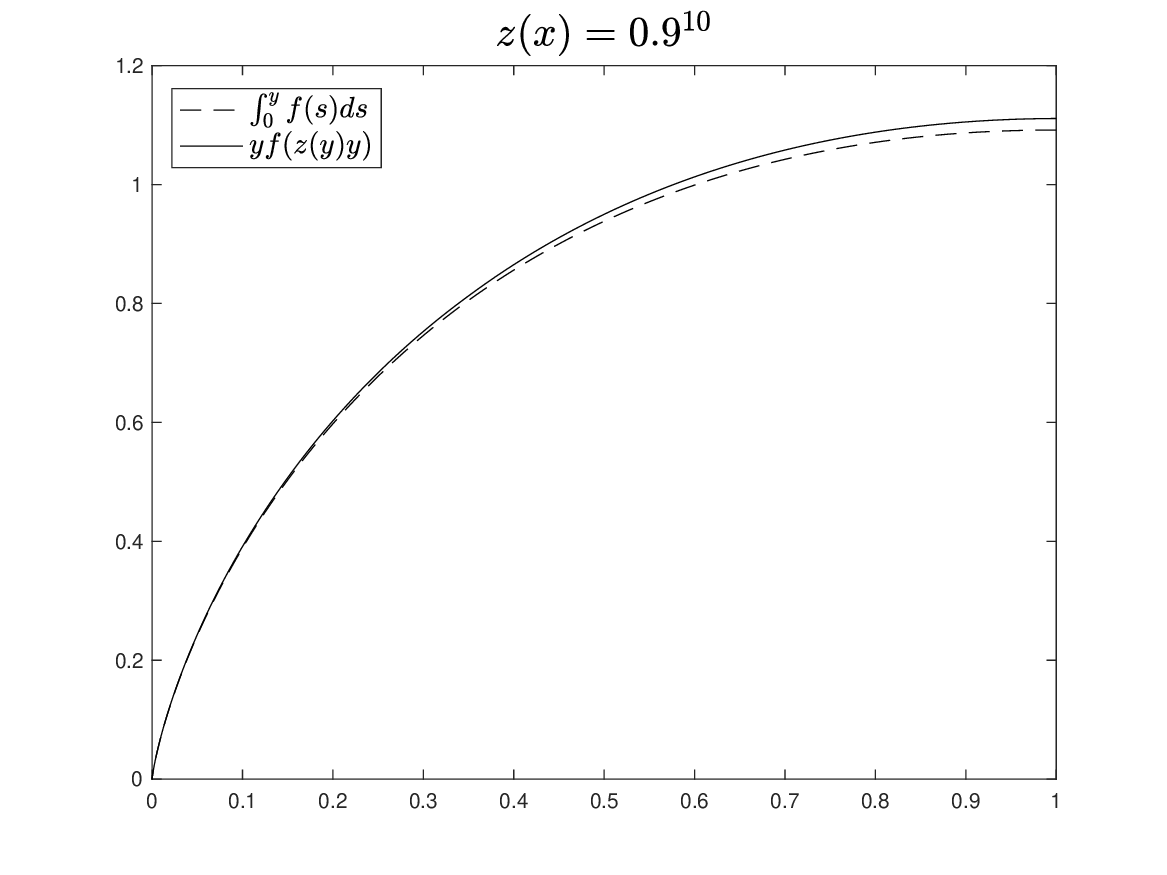}
    \end{figure}

    \begin{figure}[h]
        \caption{Calibrated function and validation for $z(x)=\log\left(x/(1-e^{-x})\right)/x$}
    \label{exp}
    \centering \includegraphics[scale=0.39]{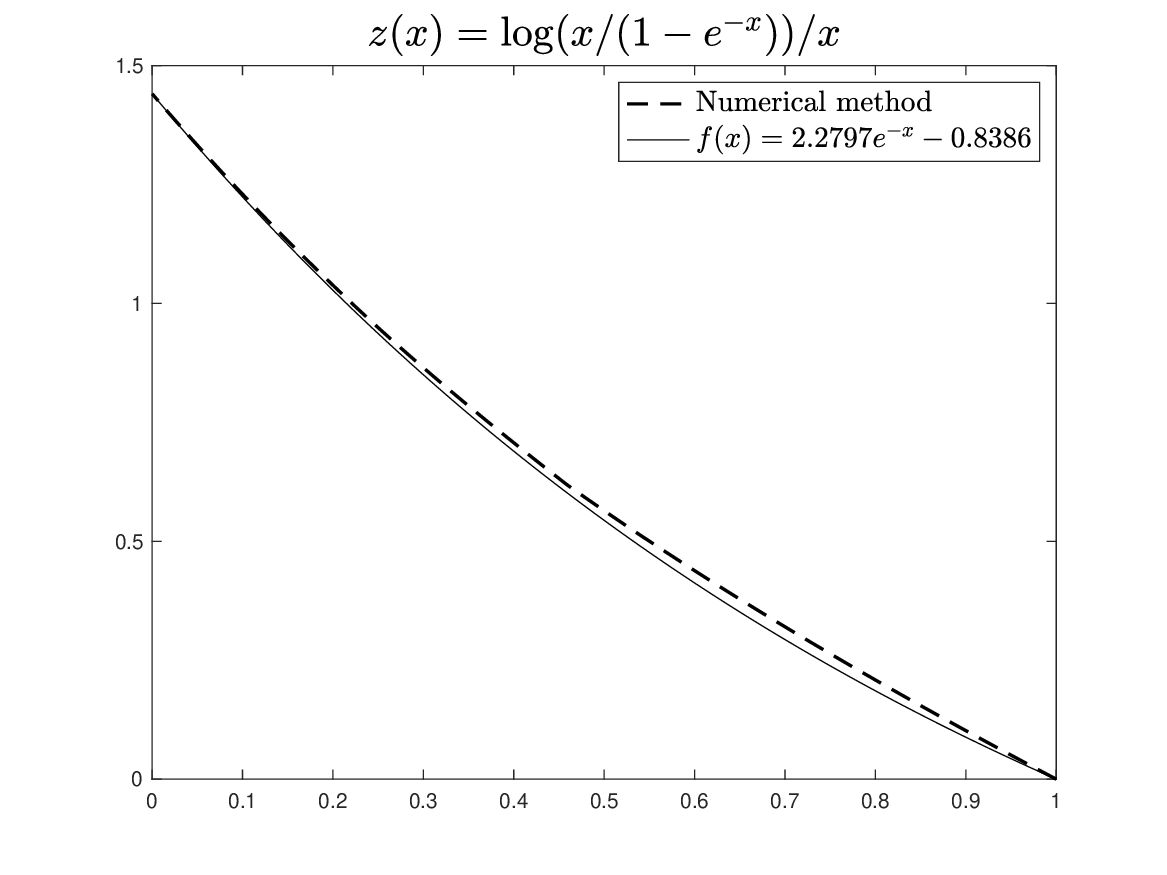} \includegraphics[scale=0.39]{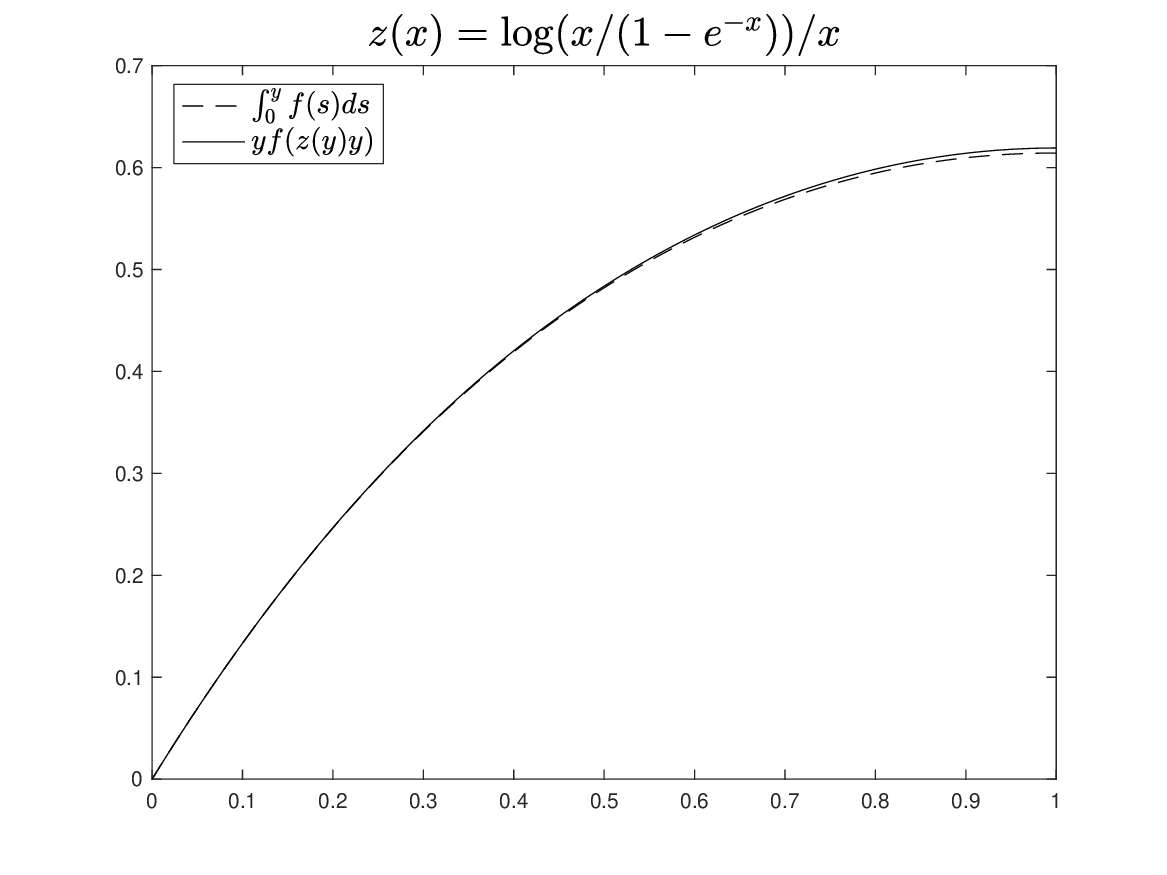}
    \end{figure}

\section{Technical properties of the PELVE}\label{sec:property}

We now take a turn to study several additional properties of PELVE. In particular, we will obtain results on  the monotonicity  and convergence of the dual PELVE as well as the PELVE.

\subsection{Basic properties of dual PELVE}

The following proposition that shows the PELVE and dual PELVE share some basic properties such as monotonicity (i), location-scale invariance (ii) and shape relevance (iii)-(iv) below.
\begin{proposition}\label{property}
    Suppose the quantile function  of $X$ satisfies Assumption \ref{ass:1} and $\epsilon\in(0,1]$.
    \begin{enumerate}[(i)]
    \item $\Pi_{X}(\epsilon)$ is increasing (decreasing) in $\epsilon$ if and only if so is $\pi_X(\epsilon)$.
    \item For all $\lambda>0$ and $a\in\R$, $\pi_{\lambda X+a}(\epsilon)=\pi_{X}(\epsilon)$.
    \item $\pi_{f(X)}(\epsilon)\le\pi_{X}(\epsilon)$ for all strictly increasing concave functions: $f:\R\to\R$ with $f(X)\in\X$.
    \item $\pi_{g(X)}(\epsilon)\ge\pi_{X}(\epsilon)$ for all strictly increasing convex functions: $g:\R\to\R$ with $g(X)\in\X$.
    \end{enumerate}
\end{proposition}

The statements (ii)-(iv) are parallel to the corresponding statements in Theorem 1 of \cite{LW22} on PELVE.
The proof of Proposition \ref{property} is put in Appendix \ref{app:proof_sec6}. Proposition \ref{property}  allows us to study the monotonicity and convergence of the PELVE by analyzing the corresponding properties of  the dual PELVE, which is more convenient in many cases. In the following sections, we focus on finding the conditions which make the dual PELVE monotone and convergent at 0. By Proposition \ref{property}, those conditions can also apply to the PELVE.
\subsection{Non-monotone and non-convergent examples}

In this section, we study   the monotonicity and convergence of dual PELVE. For monotonicity, we have shown some well-known distributions such as normal distribution, t-distribution and lognormal distribution have monotone PELVE curves in Example \ref{PELVE}.
However, the PELVE is not monotone for all $X \in \mathcal{X}$. Below we provide an example.
\begin{example}
[Non-monotone PELVE]
    Let us consider the following density function $g$ on $[-2,2]$,
       $$
    g\left(x\right)
    = \frac{1}{2}\left((x+2)\id_{\{x\in[-2,-1]\}} -x\id_{\{x\in(-1,0]\}}+ x\id_{\{x\in(0,1]\}} +  (2-x) \id_{\{x\in (1,2]\}} \right) .
    $$
    For $X$ with density function $g$,
    Figure \ref{nonmontone} presents the value of $\Pi_{X}(\epsilon)$ for $\epsilon\in(0,0.5)$.
As one can see, the PELVE is not necessarily decreasing, and so is the dual PELVE.
    \begin{figure}[htpb]
    \centering
        \caption{PELVE for $X$ with density $g$}
    \label{nonmontone}
    \includegraphics[scale=0.6]{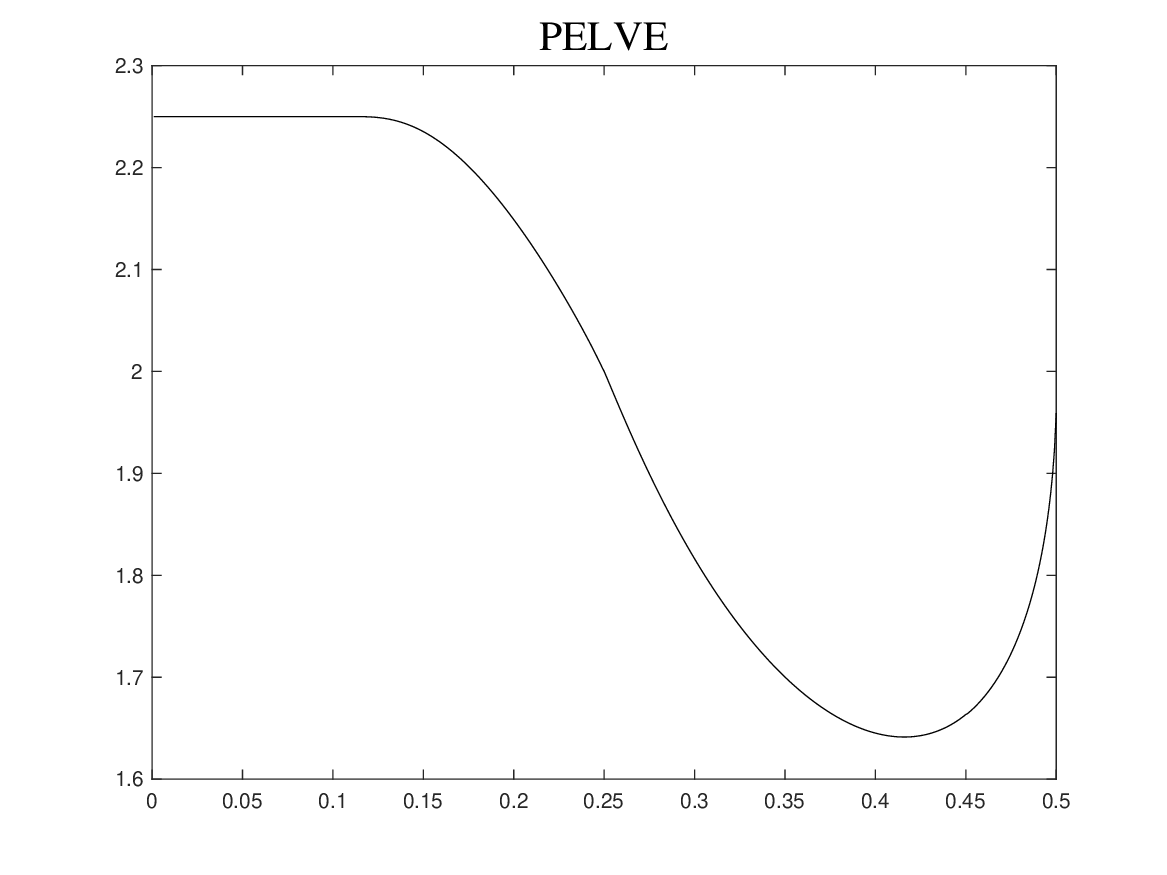} 
    \end{figure}
\end{example}
For the convergence, it is clear that $\pi_{X}(\epsilon)$ is continuous in $(0,1)$ for $X\in \mathcal{X}$.
Therefore, $\lim_{\epsilon\to p}\pi_{X}(\epsilon)$ exists for all $p\in(0,1)$. However, both $\Pi_{X}(\epsilon)$ and $\pi_{X}(\epsilon)$ are not well defined at $\epsilon=0$. If $\lim_{\epsilon\to0} \pi_{X}(\epsilon)$ exists, we can define $\pi_{X}(0)$ as the limit, and $\Pi_{X}(0)$ similarly.  However, the following example shows that the limit does not exist for some $X\in \mathcal{X}$.
\begin{example}[No limit at 0]
 We can construct a random variable $X\in\X$ such that $\lim_{\epsilon\to0}\pi_{X}(\epsilon)$ does not exist from the integral equation \eqref{eq:abstract} in Section \ref{abstract}. Equivalently, we will find a continuous and strictly decreasing function $f\in\mathcal{C}$ such that $\lim_{y\to0}z_{f}(y)$ does not exist.
    Let $c$ be the Cantor ternary function on $[0,1]$. Note that $x\mapsto c(x)$ is continuous and increasing on $(0,1)$ and $c(x/3)=c(x)/2$. Let $f(x)=-c(x)-x^{\log2/\log3}$.
    It is clear that $f\in\mathcal{C}$ and $f(x/3)=f(x)/2$. For each $y\in(0,1]$, we have
    $$
    \begin{aligned}
    yf\left(z_{f}(y)y\right)
     & =\int_{0}^{y}f(x)\d x
     \\ & =2\int_{0}^{y}f\left(\frac{1}{3}x\right)\d x
  =6\int_{0}^{\frac{1}{3}y}f(x)\d x=2yf\left(\frac{1}{3}yz_{f}\left(\frac{1}{3}y\right)\right)  =yf\left(yz_{f}\left(\frac{1}{3}y\right)\right).
    \end{aligned}
    $$
    Since $f$ is strictly decreasing, $z_{f}(y)=z_{f}(y/3)$ for $y\in(0,1]$.
    It means that $z_{f}(y)$ is a constant on $(0,1]$ if $\lim_{y\to0}z_{f}(y)$ exists.
    Now, let us look at two particular points of $z_{f}(y)$. We can show that $z_{f}(1)\neq z_{f}(4/9)$.
    Let $z=(\log2/\log3+1)^{-(\log3/\log2)}$. Then, we have $1/3<z\approx 0.46<1/2$.
    For $y=1$, we have $\int_{0}^{1}c(s)\d s=c(z)=1/2$ and $\int_{0}^{c}s^{\log2/\log3} \d s=z^{\log2/\log3}$. Therefore, we get $z_{f}(1)=z<1/2$.
    For $y=4/9$, we have
    $$
    \begin{aligned}
    f\left(\frac{4}{9}z_{f}\left(\frac{4}{9}\right)\right)
    & =\frac{9}{4}\int_{0}^{4/9}f(s)\d s\\
    & =-\frac{9}{4}\left(\frac{1}{\frac{\log2}{\log3}+1}\left(\frac{4}{9}\right)^{\frac{\log2}{\log3}+1}+\frac{1}{12}+\frac{1}{2}\left({\frac{4}{9}-\frac{1}{3}}\right)\right)
    <-0.68<f\left(\ensuremath{\frac{2}{9}}\right)\approx-0.64.
    \end{aligned}
    $$
    As $f$ is strictly decreasing, we have $(4/9)z_{f}(4/9)>2/9$ which implies $z_{f}(4/9)>1/2>z_{f}(1)$.
    As a result, $\lim_{y\to0}z_{f}(y)$ does not exist. Therefore, we have a continuous and strictly decreasing $f$ such that  $\lim_{y\to0}z_{f}(y)$ does not exist.
\end{example}

\subsection{Sufficient condition for monotonicity and convergence}\label{monotonicity}
 In risk management applications,
for a random variable  $X$ modeling a random loss, the behavior of its tail is the most important.
Let $F^{[p,1]}$ be the upper $p$-tail distribution of $F$ (see e.g., \cite{LW21}), namely
    $$
    F^{[p,1]}(x)  =\frac{(F(x)-p)_{+}}{1-p},~~~x\in \R.
    $$
We will see that the dual PELVE of $F^{[p,1]}$ is a part of the dual PELVE of $F$.

\begin{lemma}\label{lem:tail}
 Let $F$ be the distribution function of $X$ with  quantile function  satisfying Assumption \ref{ass:1}. For  $p\in(0,1)$ and $X'\sim F^{[p,1]}$, it holds
    $$\pi_{X'}(\epsilon) =\pi_{X}(\epsilon(1-p)).$$
\end{lemma}
\begin{proof}
    It is clear that $\VaR_{\epsilon}(X')=\VaR_{{\epsilon(1-p)}}(X)$ and $\ES_{\epsilon}(X')=\ES_{{\epsilon(1-p)}}(X)$. Therefore,
    $$
    \begin{aligned}
     \pi_{X'}(\epsilon)
     & =\inf\{d\ge1:\ES_{\epsilon}(X')\le\VaR_{\epsilon/d}(X')\}\\
     & =\inf\{d\ge1:\ES_{\epsilon(1-p)}(X')\le\VaR_{\epsilon(1-p)/d}(X')\}  =\pi_{X}(\epsilon(1-p)).
    \end{aligned}
    $$
    Thus, we have the desired result.
\end{proof}

The tail distribution can provide a condition to check whether the dual PELVE is decreasing.

\begin{proposition}\label{convex}
 Let $F$ be the distribution function of $X$ with  quantile function  satisfying Assumption \ref{ass:1}.
    If $x\mapsto F^{-1}\big((1-p)F(x)+p\big)$ is convex (concave) for all $p\in(0,1)$, then $\pi_{X}$ and $\Pi_X$ are decreasing (increasing).
\end{proposition}
\begin{proof}
    For any $p\in(0,1)$, let $X'\sim F^{[p,1]}$. By Lemma \ref{lem:tail},
    we have $\pi_{X'}(\epsilon)=\pi_{X}(\epsilon(1-p))$. Furthermore,
    we have
    \begin{align*}
 {\left(F^{[p,1]}\right)}^{-1}(t) & =F^{-1}\left(\ensuremath{(1-p)t+p}\right) =F^{-1}\Big((1-p)F\big(F^{-1}(t)\big)+p\Big),\quad t\in[0,1].
    \end{align*}
    Let $U\sim\mathrm{U}(0,1)$, $X=F^{-1}(U)$ and $X'=(F^{[p,1]})^{-1}(U)$.

    We assume that $x\mapsto F^{-1}\big((1-p)F(x)+p\big)$ is a convex
    function on $(\essinf(X),\esssup(X))$ first. Let $f:\R\to\R$ be
    a strictly increasing convex function such that $f(x)=F^{-1}((1-p)F(x)+p)$
    for $x\in(\essinf(X),\esssup(X)).$ Then, we have $X'=f(X)$. By Proposition
    \ref{property}, we get $\pi_{X'}(\epsilon)\ge\pi_{X}(\epsilon)$.
    As $\pi_{X'}(\epsilon)=\pi_{X}\ensuremath{(\epsilon(1-p))}$, we have
    $\pi_{X}\ensuremath{(\epsilon(1-p))}\ge\pi_{X}(\epsilon)$ for all $p\in(0,1)$.
    Thus, $\pi_{X}$ is decreasing. By Proposition \ref{property}, we have  $\Pi_{X}$ is also decreasing.

    On the other hand, if $x\mapsto F^{-1}\big((1-p)F(x)+p\big)$ is concave,
    we have $\pi_{X}\ensuremath{(\epsilon(1-p))}\le\pi_{X}(\epsilon)$ for
    all $p\in(0,1)$ and $\pi_{X}$ is increasing. So is $\Pi_X$.
\end{proof}

The condition $x\mapsto F^{-1}\big((1-p)F(x)+p\big)$ is convex (concave) for all $p\in(0,1)$ is generally hard to check. Intuitively, this condition means that  $F^{[p,1]}$ has a less heavy tail compared to $F$.
We can   further simplify this condition by using  the hazard rate function. For $X\in\X$ with distribution function $F$ and density function $f$, let $S=1-F$ be the survival function and $\eta=f/S$ be the hazard rate function. As $F$ is continuous and strictly increasing, $S$ is continuous and strictly decreasing.

\begin{theorem}\label{hazard_rate}
    For  $X$ with quantile function satisfying Assumption \ref{ass:1}, let $\eta$ be the hazard rate function of $X$. If $1/\eta$ is second-order differentiable and convex (concave), then $\pi_{X}$ and $\Pi_{X}$ are decreasing (increasing).
\end{theorem}
The proof of Theorem \ref{hazard_rate} is provided in Appendix \ref{app:proof_sec6}.

\begin{example}
    For the normal distribution, we can give a   short proof of the convexity of $1/\eta$. Let $S$ be the survival function of the standard normal distribution and $f$ its density.
    Let $I\left(x\right)=1/\eta\left(x\right)= S (x )/f (x )=\exp\left(x^{2}/2\right)\int_{-\infty}^{-x}\exp\left(-s^{2}/2\right)\d s$.
    One can easily see that
    \begin{equation}
    I'\left(x\right)=xI\left(x\right)-1\label{eq:first-derivative}
    \end{equation}
    which gives $I''\left(x\right)=xI'\left(x\right)+I\left(x\right)$.
    This with \eqref{eq:first-derivative} implies that
    \begin{equation}
    I''\left(x\right)=\left(1+x^{2}\right)I\left(x\right)-x.\label{eq:second-derivative}
    \end{equation}
    First, consider the negative line i.e., $x<0$. In this case \eqref{eq:first-derivative} and \eqref{eq:second-derivative} imply
    $
    I' (x ) =xI (x )-1<0,
    $
    and $
    I'' (x ) = (1+x^{2} )I (x )+ (-x )>0.
    $
    The implication of the two relations is that $I$ is a convex and decreasing function on negative line. Now we consider the case $x>0$.
    In this case, let $i\left(x\right)=I'\left(-x\right)$. From what we have proved it is clear that $i$ is an increasing function on $x>0$.
    On the other hand, we have $I\left(x\right)+I\left(-x\right)=1/f \left(x\right)=\sqrt{2\pi}\exp\left(x^{2}/2\right)$.
    This combined with \eqref{eq:first-derivative} gives us
    $$
    I'\left(x\right) =x\left(I\left(x\right)+I\left(-x\right)\right)+i\left(x\right)=x\sqrt{2\pi}\exp\left(x^{2}/2\right)+i\left(x\right),x>0.
    $$
    This means $I'$ is an increasing function on $x>0$ as it is a summation of two other increasing functions, so $I$ is convex on the positive line as well.
\end{example}

Figure \ref{hazard} presents the curve   $1/\eta$ for the generalized Pareto distribution, the Normal distribution, the t-distribution and the Lognormal distribution.
For distributions $\mathrm{GPD}(1/2)$, $\mathrm{N}(0,1)$ and $\mathrm{t}(2)$, we can see that the curves $1/\eta$ are convex, and this coincides  with decreasing PELVE shown in Example \ref{ex:decreasing-PELVE}.
For the Lognormal distribution, the shape of $1/\eta$ depends on $\sigma$. As shown in Example \ref{ex:decreasing-PELVE}, the PELVE for $\mathrm{LN}(\sigma)$ is visibly decreasing for $\sigma^{2}=0.04$ and increasing for $\sigma=1$. Corresponding to the above observations, we see that $1/\eta$ is convex for $\sigma^{2}=0.04$ and concave for $\sigma^{2}=1$.

\begin{figure}[htpb]
    \centering
        \caption{$1/\eta$ for $\mathrm{GPD}(1/2)$, $\mathrm{N}(0,1)$, $\mathrm{t}(2)$, $\mathrm{LN}(0,2)$ and $\mathrm{LN}(1)$ in blue curves; in the right panel, the red curve is  linear}
    \label{hazard}
    \includegraphics[width=2.5in]{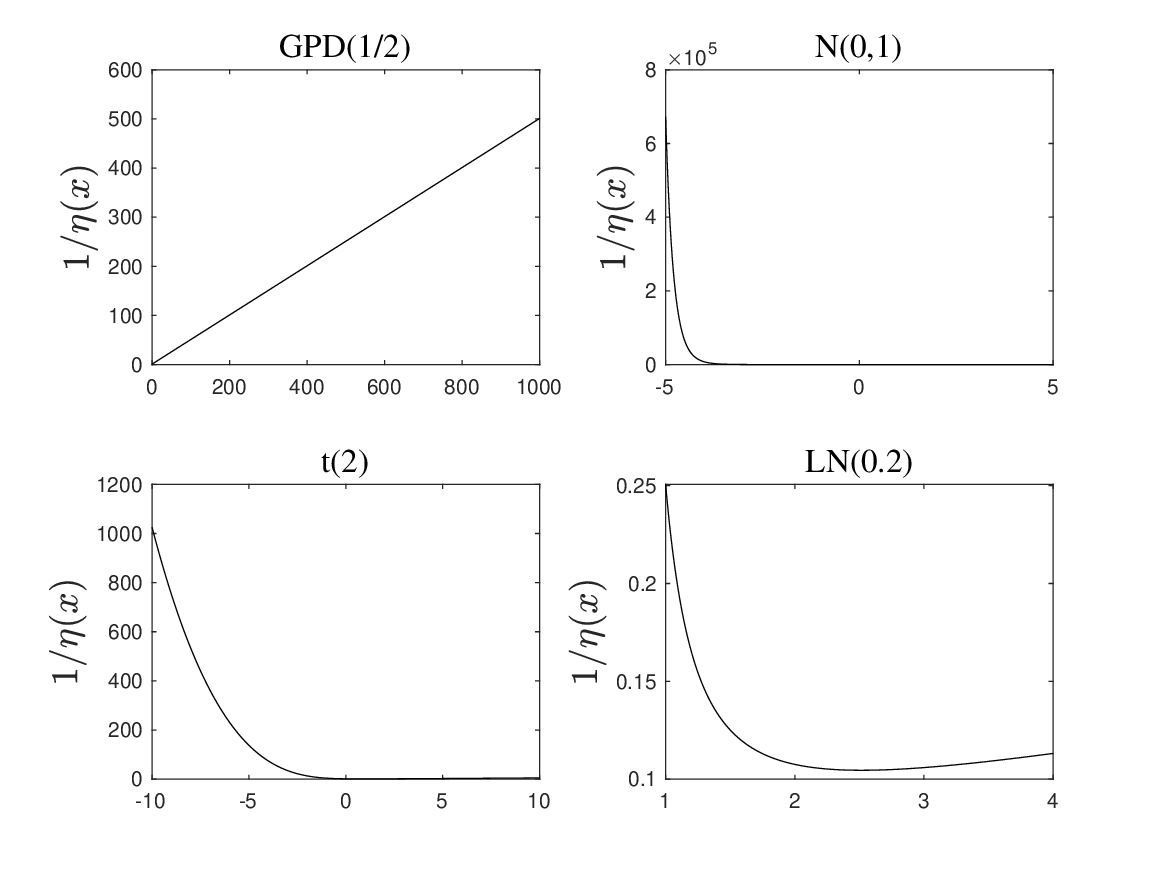} \includegraphics[width=2.5in]{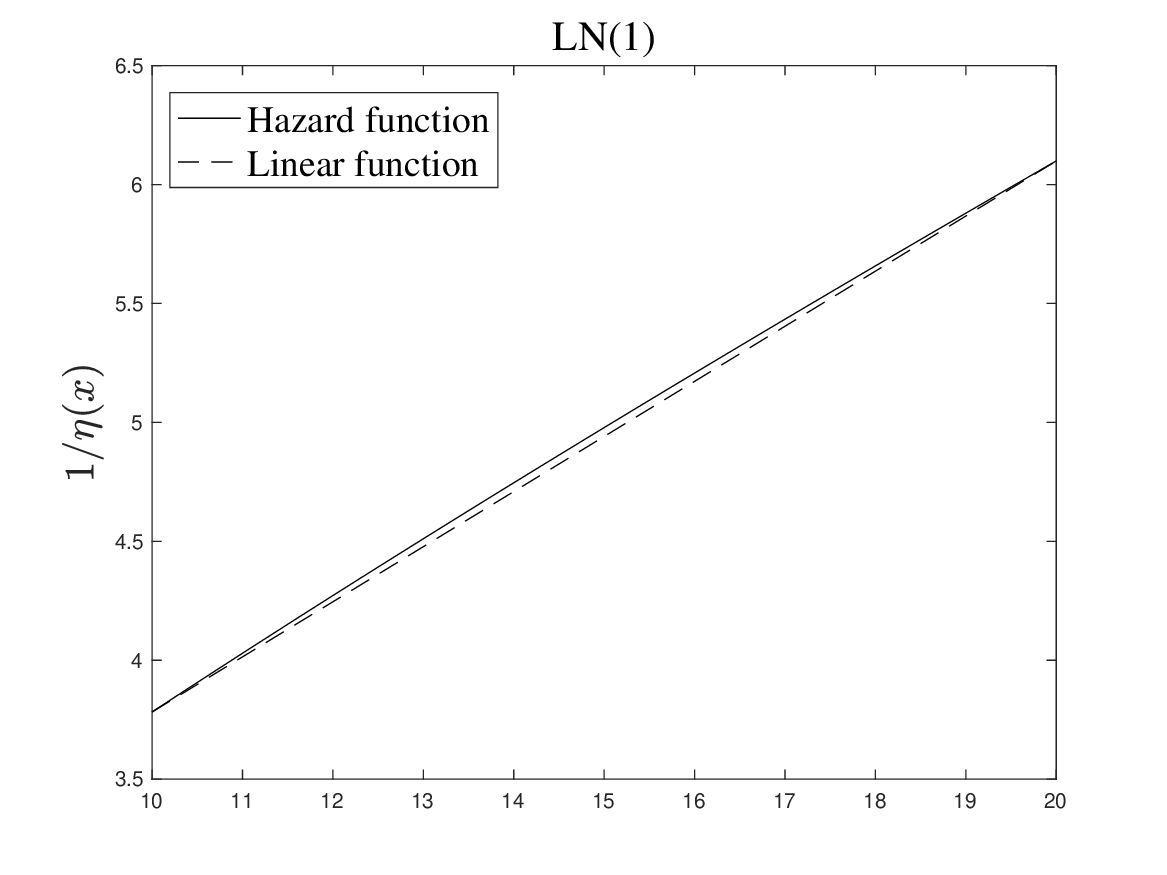}
\end{figure}

\begin{corollary}\label{cor:concave_hazard}
    If the hazard rate of a  random variable $X$ is second-order differentiable and concave, then $\pi_{X}$ and $\Pi_X$ are decreasing.
\end{corollary}
\begin{proof}
    Just note that  if $\eta$ is concave, then $\eta \eta''$ is non-positive. It follows that   $$
    \left(\frac{1}{\eta}\right)''=\left(-\frac{\eta'}{\eta^{2}}\right)'=\frac{2\left(\eta'\right)^{2}-{{\eta\eta''}}}{\eta^{3}}\ge0.
    $$
    Thus, ${1}/{\eta}$ is convex, and the desired statement follows from Theorem \ref{hazard_rate}.
\end{proof}
    The  corollary above is a result of the fact that the concavity of $\eta$ implies convexity of $1/\eta$. Therefore, concave $\eta$ always leads to decreasing PELVE. For example, the Gamma distribution $\mathrm{G}(\alpha,\lambda)$ with density $f(x)=\lambda^\alpha t^{\alpha-1}e^{-\lambda t}/\Gamma(\alpha)$ has concave hazard rate function  when $\alpha>1$.
    Furthermore, by Theorem \ref{hazard_rate}, we can easily find  more well-known distributions that have decreasing $\pi_{X}$.

As the tail distribution determines $\pi_{X}$ around $0$, we can focus on the tail distribution to discuss the convergence of $\pi_{X}$ at $0$.
Note that if the survival distribution function is regularly varying, then
its tail parameter one-to-one corresponds to the limit of $\Pi_X$ at $0$ as shown by Theorem 3 of \cite{LW22}. Hence, the limit of $\Pi_X$, if it exists, can be useful as a measure of tail heaviness, and it is well defined even for distributions that do not have a heavy tail.
 By the monotone convergence theorem, we have $\lim_{\epsilon\to0}\pi_{X}(\epsilon)$ exists if $\pi_X$ is monotone.
The limit may be finite or infinite.

\begin{corollary}\label{cor:limit}
    For $X$ with quantile function  satisfying Assumption \ref{ass:1}, let $\eta$ be the hazard rate of $X$. If $1/\eta(x)$ is second-order differentiable and convex (concave) in $\left(\ensuremath{F^{-1}(\delta),\esssup(X)}\right)$ for some $\delta\in(0,1)$, then $\lim_{\epsilon\to0}\pi_{X}(\epsilon)$ exists.
    In particular, this is true if $\eta$ is second-order differentiable and concave on $\left(\ensuremath{F^{-1}(\delta),\esssup(X)}\right)$.
\end{corollary}

\begin{proof}
    Let $X'\sim F^{[\delta,1]}$. Then, the survival function for $X'$ is $S_{X'}(x)=S(x)/(1-p)$ for $x\ge F^{-1}(\delta)$.
    The density function is $f_{X'}(x)=f(x)/(1-p)$ for $x\ge F^{-1}(\delta)$.
    Therefore, the hazard rate function is $\eta_{X'}(x)=f(x)/S(x)=\eta(x)$ for $x\ge F^{-1}(\delta)$.

    As $1/\eta(x)$ is convex (concave) when $x>F^{-1}(\delta)$, we have $1/\eta_{X'}(x)$ is convex (concave).
    By Theorem \ref{hazard_rate}, we have $\pi_{X'}(\epsilon)$ is decreasing (increasing) on $(0,1)$.
    As a result, we have $\pi_{X}(\epsilon)$ is decreasing (increasing) on $(0,\delta)$ and $\lim_{\epsilon\to0}\pi_{X}(\epsilon)$ exists.

    By Corollary \ref{cor:concave_hazard}, if $\eta$ is concave on $(F^{-1}(\delta),\esssup(X))$, $1/\eta$ is convex on  $(F^{-1}(\delta),\esssup(X))$ and $\lim_{\epsilon \to 0} \pi_X(\epsilon)$ also exists.
\end{proof}

\begin{example}
If $\lim_{\epsilon\to0} \pi_X(\epsilon)$ is a constant, we have $\lim_{\epsilon \to 0}\Pi_{\epsilon}(X)=\lim_{\epsilon\to0} \pi_X(\epsilon)$ as $\pi_{X}(\Pi_{X}(\epsilon)\epsilon)=\Pi_{X}(\epsilon)$.
We give the numerical values of $\Pi_X(\epsilon)$ at very small probability levels $\epsilon$ for normal, t, and log-normal distributions. These distributions do not have a constant PELVE curve, and using Corollary \ref{cor:limit} we can check that  their PELVE have limits.
As we can see from Table \ref{table:limit}, PELVE can still distinguish the heaviness of the tail even when $\epsilon$ is very small. The heavier tailed distributions report a higher PELVE value. For the normal distribution and the log-normal distribution with $\sigma=0.2$, the value of PELVE is close to $e\approx2.7183$ as $\epsilon \downarrow 0$.
From the numerical values, it is unclear whether $\Pi_X(\epsilon)\to e$ for all log-normal distributions, but there is no practical relevance to compute $\Pi_X(\epsilon)$ for $\epsilon<10^{-11}$ in applications.
\begin{table}[htbp]
\def\arraystretch{1.4}
    \centering{}
    \caption{ Values of $\Pi_{X}(\epsilon)$}\label{table:limit} %
    \begin{tabular}{m{1.6cm}<{\centering}  m{1.6cm}<{\centering} m{1.6cm}<{\centering} m{1.6cm}<{\centering} m{1.6cm}<{\centering} m{1.6cm}<{\centering} m{1.6cm}<{\centering}}
    \hline \hline
    Distribution& N  & LN($1$)  & LN($0.5$) &LN($0.2$)  & t($2$) & t($3$)\\
    \hline
    $\epsilon={10}^{-10}$ &2.6884  &2.9167&2.7944 & 2.7290 &4.0000&3.3750 \\
    \hline
    $\epsilon={10}^{-11}$ &2.6909 &2.9077& 2.7920 &2.7287 &4.0000&3.3750 \\
    \hline\hline
    \end{tabular}
\end{table}

\end{example}

\section{Applications to datasets used in insurance}\label{sec:6}
In this section, we apply the PEVLE calibration techniques to datasets used in insurance to show how to use the calibrated distribution in estimating risk measure values and simulation.

\subsection{Dental expenditure data}
In this example, we apply the calibration model to the 6494 complete household component's total dental expenditure data from  Medical Expenditure Panel Survey for 2020. An earlier version of the same dataset  is used by \cite{BCLPPY10}  to study the relationship between  worker absenteeism and  overweight or obesity. 
The main purpose of this experiment is to construct tractable models, with continuous and simple quantile functions, which have similar risk measure values   as the original dataset, and the same PELVE at certain levels.
We present in Figure \ref{fig:MPC} two quantile functions calibrated from $\Pi_X(\epsilon_1)$ and $\Pi_X(\epsilon_2)$, with $(\epsilon_1,\epsilon_2)=(0.01,0.05)$ and $(\epsilon_1,\epsilon_2)=(0.05,0.1)$, respectively. The two calibrated quantile functions are scaled up according to  the empirical  $\VaR_{\epsilon_1}(X)$ and $\VaR_{\epsilon_2}(X)$.
By Theorem \ref{thm:two_points}, we can calibrate the quantile functions from Case 4 when $(\epsilon_1,\epsilon_2)=(0.01,0.05)$,  and from Case 5 when $(\epsilon_1,\epsilon_2)=(0.05,0.1)$. As mentioned before, for $(\epsilon_1,\epsilon_2)=(0.05,0.1)$,  we set the  calibrated quantile function in $(0, c_1\epsilon_1)$  as the Pareto quantile function. Hence, there is no flat part in the two calibrated quantile functions shown in Figure \ref{fig:MPC}.
As we can see, both the two calibrated quantile functions fit the empirical quantile functions well.
The calibrated quantile function can be regarded as a special parameterized  model for tail distribution, which can fit the value of $\VaR$  and $\ES$ at specified levels.
    With the parameterized calibrated model, we can estimate the value of tail risk measures (see \cite{LW21}) such as ES, VaR, and Range-VaR (RVaR), amongst others. 
   In Tables \ref{tab:MPC_ES} and \ref{tab:MPC_RVaR}, we compute the values of ES and RVaR for the calibrated model and compare them with empirical ES and RVaR values, respectively, where  the risk measure RVaR is defined as $$\mathrm{RVaR}_{\alpha,\beta}(X)=\frac{1}{\beta-\alpha}\int_\beta^\alpha \VaR_\gamma(X) \d \gamma$$ for $0\le \alpha<\beta<1$; see \cite{CDS10} and \cite{ELW18}. As we scale the calibrated quantile function  to empirical $\VaR_{\epsilon_1}(X)$ and $\VaR_{\epsilon_2}(X)$, the calibrated ES and empirical ES are identical at levels $\epsilon_1\Pi_X(\epsilon_1)$ and $\epsilon_2\Pi_X(\epsilon_2)$ by the definition of PELVE. For other probability levels, the calibrated ES and RVaR  in Tables \ref{tab:MPC_ES} and \ref{tab:MPC_RVaR} are   close to their empirical counterparts. When $(\epsilon_1,\epsilon_2)=(0.01,0.05)$, it may only be useful to compute calibrated
    $\ES_p(X)$ for $p<0.05\Pi_X(0.05)=0.11591$ because the calibrated quantile function is arbitrary beyond the level $0.11591$. If we need to estimate ES or RVaR for a larger probability level, we can choose a higher $\epsilon_2$ as long as $\E[X]\le \VaR_{\epsilon_2}(X)$ is satisfied. For this dataset, the highest $\epsilon_2$ we can use is 0.1983. 

 Using the methods in Section \ref{sec:points}, for quantile levels between $(0, \epsilon_1)$, the distribution calibrated from one point $(\epsilon_1, c_1)$ is the same as the one calibrated from two points  $(\epsilon_1,c_1)$ and $(\epsilon_2, c_2)$.  Hence, the   results for $\ES_p$ of the one-point calibrated function are also shown in Tables \ref{tab:MPC_ES} and \ref{tab:MPC_RVaR} in the cells $p \le \epsilon_1$.

     \begin{figure}[h]
    \centering
       \caption{Empirical and calibrated  $\VaR_{\epsilon}$ for the dental expenditure data}\label{fig:MPC}
    \includegraphics[scale=0.5]{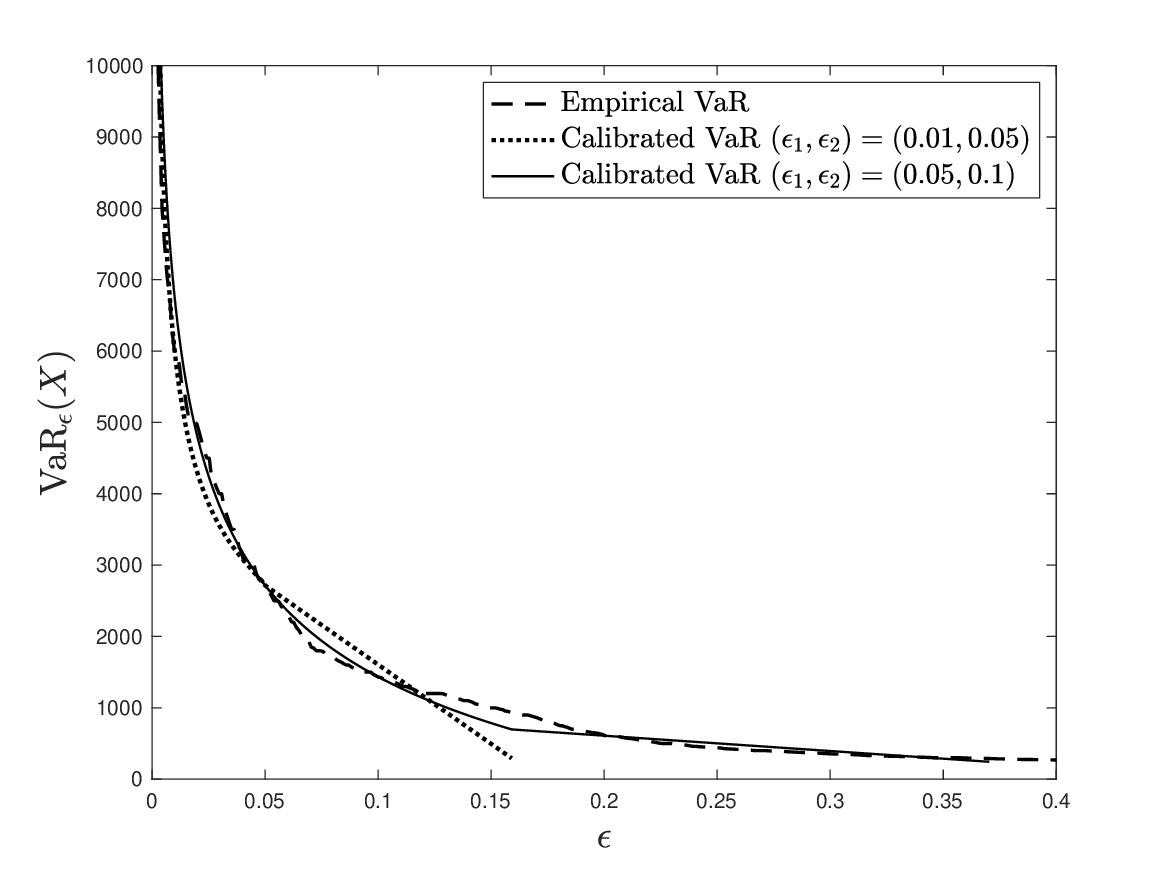}
    \end{figure}

\begin{table}[htbp]
\def\arraystretch{1.4}
    \centering{}
    \caption{ Empirical ES and calibrated ES  for  the dental expenditure data}\label{tab:MPC_ES} %
    \begin{tabular}{c c c c c c}
    \hline \hline
    $p$& 0.01  & 0.05  & 0.1 &0.2  & 0.3\\
    \hline
    Empirical $\ES_p$ & 10073.1 &5361.7&3624.8&2317.9&1696.7 \\
    Calibrated $\ES_p$ from $(\epsilon_1, \epsilon_2)=(0.01, 0.05)$  &11703.9&5357.7&3759.1& - & -\\
    Calibrated $\ES_p$ from $(\epsilon_1, \epsilon_2)=(0.05, 0.1)$ &10878.1&5439.6&3711.3&2293.7&1696.4\\
    \hline\hline
    \end{tabular}
\end{table}

\begin{table}[htbp]
\def\arraystretch{1.4}
    \centering{}
    \caption{ Empirical RVaR and calibrated RVaR  for the dental expenditure data}\label{tab:MPC_RVaR} %
    \begin{tabular}{c c c c c c}
    \hline \hline
    $(\alpha, \beta)$& $(0.01,0.02)$  & $(0.02,0.05) $ & $(0.05,0.1)$  \\
    \hline
    Empirical $\RVaR_{\alpha,\beta}$ & 5748.5 &3662.4&1887.9\\
    Calibrated $\RVaR_{\alpha,\beta}$ from $(\epsilon_1, \epsilon_2)=(0.01, 0.05)$ &5003.7&3360.2&2160.6&  \\
    Calibrated $\RVaR_{\alpha,\beta}$ from $(\epsilon_1, \epsilon_2)=(0.05,0.1)$ &5634.6&3561.8&1983.1\\
    \hline\hline
    \end{tabular}
\end{table}

\subsection{Hospital costs data}

 In this example, we apply the calibration process to the Hospital Costs data of \cite{frees2009regression} which were originally from the Nationwide Inpatient Sample of the Healthcare Cost and Utilization Project (NIS-HCUP). The data contains 500 hospital costs observations with 244 males and 256 females which can be regarded as the losses of  the health insurance policies. Using the calibration model of the two-point constraint problem, we calibrate quantile functions for females and males from PELVE at probability levels $\epsilon_1=0.05$ and $\epsilon_2=0.1$, which are shown in Figure \ref{fig:hos_VaR}. Except for estimating the risk measure, the calibrated distribution is useful in simulation. Assume the insurance company wants to know the top 10\% hospital costs; that is $X|X>\VaR_{0.1}(X)$ where $X$ is the hospital costs.
 There are only  24 available data for males and 25 available data for females, which would be not enough for making statistically solid decisions. To generate more pseudo-data points, we can simulate data from the calibrated distribution; that is, we simulate data from  $F^{[p,1]}$ where $F$ is the calibrated distribution in Figure \ref{fig:hos_VaR}. Taking $p=0.9$, we have $F^{[p,1]}(t)=\VaR_{(1-p)(1-t)}(X)$ with $\VaR_{t}(X)$ from Figure \ref{fig:hos_VaR}.  We simulate 1000 data from the calibrated distributions based on PELVE at $\epsilon_1=0.05$ and $\epsilon_2=0.1$. In Figure \ref{fig:hos_qq}, we present two QQ plots of simulated data against empirical data for females and males respectively.  As we can see, the simulated data has a similar distribution as the empirical data.  
 Those simulated pseudo-data points can be used for estimating risk measures or making other decisions.  For example, the simulated hospital cost can be used to design health insurance contrasts or set the premium in complex systems, where sometimes methods based on simulated data are more convenient to work with than methods relying on distribution functions.
 This may be seen as  an alternative, smoothed, version of bootstrap; recall that the classic bootstrap sample can only take the   values represented in the dataset. Furthermore, we compare the simulated data of hospital costs for females and males in Figure \ref{fig:FM}, which shows that the distribution of the hospital costs for females has a heavier tail than that for males.

     \begin{figure}[h]
    \centering
    \includegraphics[scale=0.5]{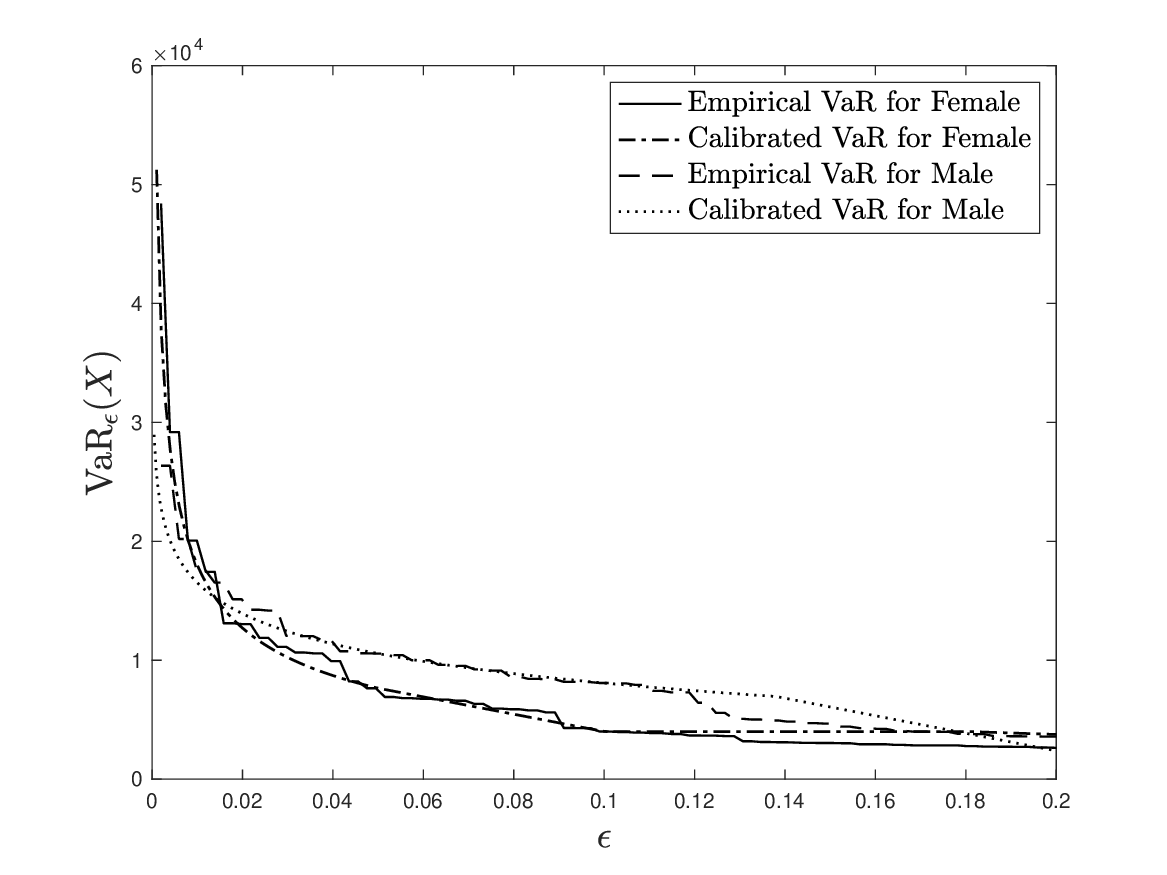}
   \caption{Empirical and calibrated  $\VaR_{\epsilon}$ for the hospital costs data}\label{fig:hos_VaR}
    \end{figure}

   \begin{figure}[h]
    \centering
    \begin{subfigure}[b]{0.5\textwidth}
    \centering
     \includegraphics[scale=0.4]{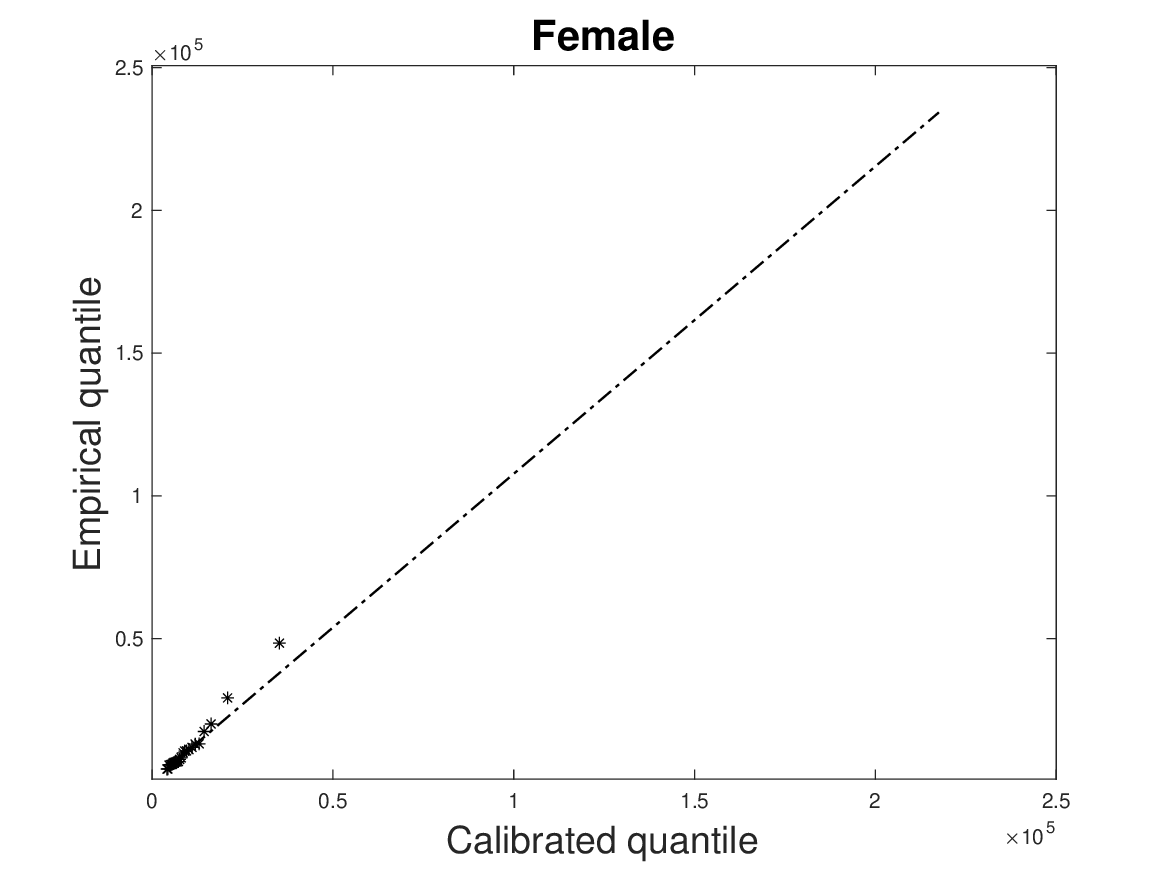}
     \subcaption{Hospital costs for female}\label{fig:F}
     \end{subfigure}~
      \begin{subfigure}[b]{0.5\textwidth}
      \centering
      \includegraphics[scale=0.4]{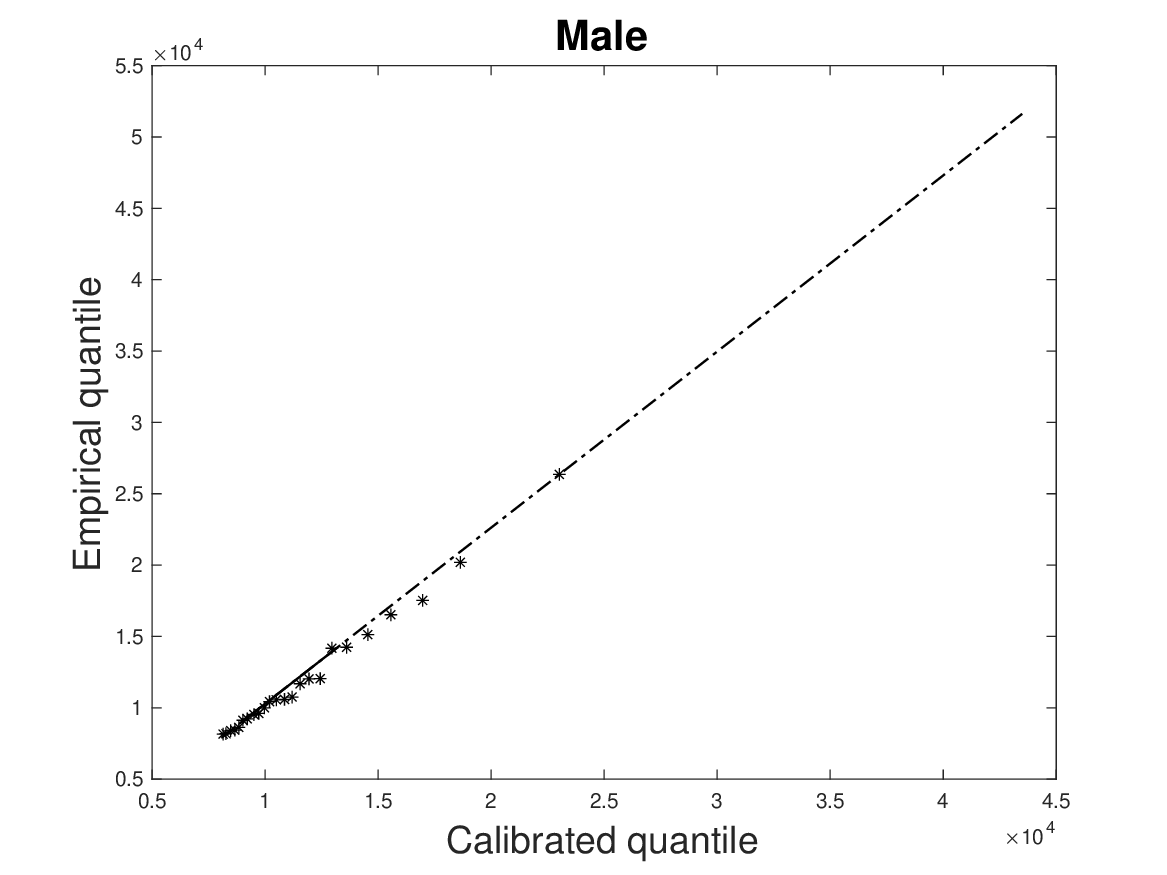}
      \subcaption{Hospital costs for male}\label{fig:M}
      \end{subfigure}
   \caption{QQ plot: simulated data VS  empirical data}\label{fig:hos_qq}
    \end{figure}
  \begin{figure}[h]
  \centering
       \includegraphics[scale=0.5]{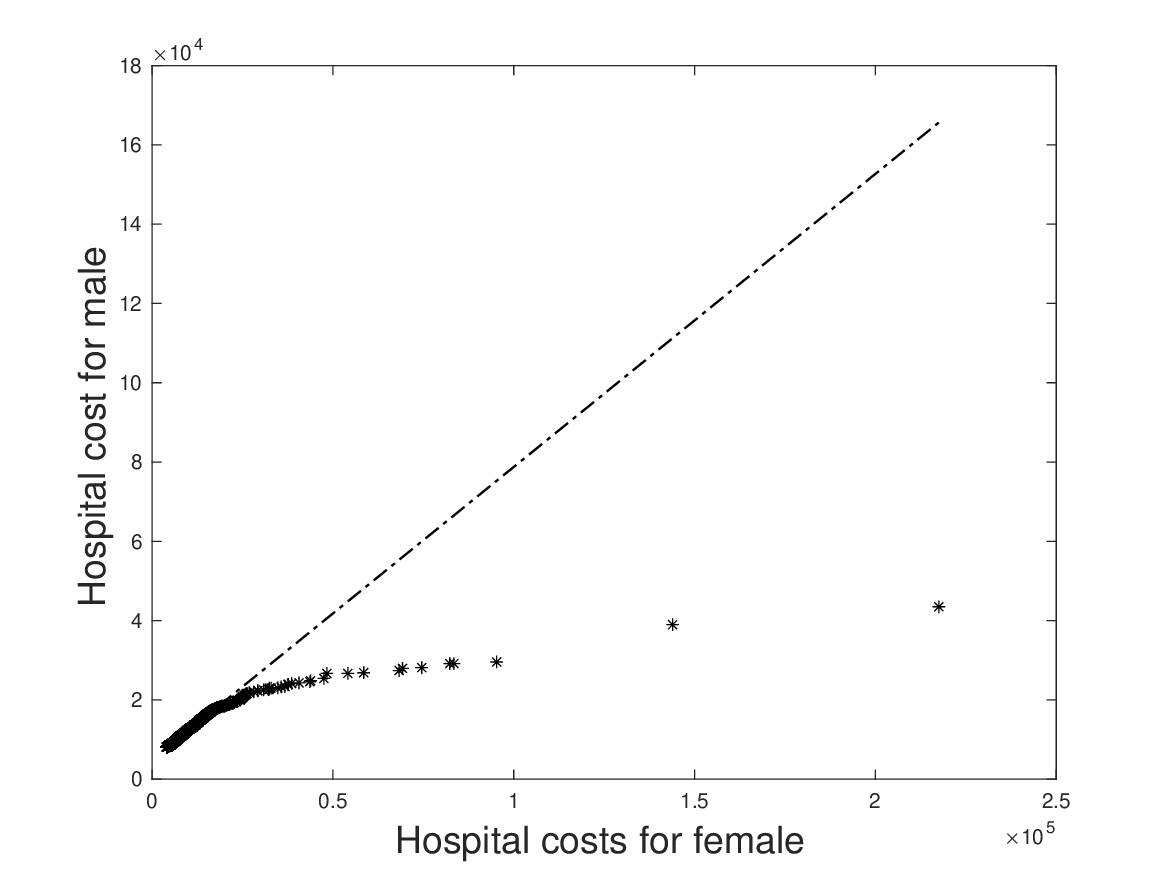}
       \caption{QQ plot of simulated data of hospital costs: female VS male}\label{fig:FM}
       \end{figure}

\section{Conclusion}\label{sec:conclusion}

In this paper, we offer several contributions to the calibration problem and properties of the PELVE.
The calibration problem concerns, with some given values from a PELVE curve, how one can build a distribution that has this PELVE.
We solve a few settings of calibration based on a one-point constraint, a two-point constraint, or the entire curve constraint.
In particular,  the calibration  for a given PELVE curve involves solving an integral equation $\int_{0}^{y}f\left(s\right)\d s=yf\left(z\left(y\right)y\right)$ for a given function  $z$, and this requires some advanced analysis and a numerical method in differential equations.
For the case that $z$ is a constant curve, we can identify all  solutions, which are surprisingly complicated. In addition, we see that if $\pi_X$ is a constant larger than $e$, which is observed from typical values in financial return data (\cite{LW22}), $X$ share the same tail behavior with the corresponding  Pareto solution.
We also applied our calibration techniques to two datasets used in insurance.

On the technical side, we study whether the PELVE is monotone and whether it converges at 0.
We show that the monotonicity of the PELVE is associated with the shape of the hazard rate.
If the inverse of the hazard rate is convex (concave), the PELVE is decreasing (increasing).
The monotonicity at the tail part of the PELVE leads to conditions to check the convergence of the PELVE at $0$.
If the inverse of the hazard rate is convex (concave) at the tail of the distribution, the limit of the PELVE at $0$ exists.

 There are several open questions related to PELVE that we still do not fully understand. One particular such question is whether the tail behavior, e.g., tail index, of a distribution is completely determined by its PELVE. We have seen that this holds true in the case of a constant PELVE (see Theorem \ref{thm:constant}), but we do not have a general conclusion. In the case of regularly varying survival functions, \citet[Theorem 3]{LW22}  showed that the limit of PELVE determines its tail parameter, but it is unclear whether this can be generalized to other distributions.
 Another challenging task is, for a specified curve $\pi$ on $[0,1]$, to determine whether there exists a model $X$ with $\pi_X=\pi$. The case of $n$-point constraints for large $n$ may require a new design of verification algorithms. This question concerns the compatibility of given information with statistical models, which has been studied, in other applications of risk management, by \cite{EMS02, EHW16} and \cite{KSSW18}.

\subsection*{Acknowledgements}
The authors thank Xiyue Han for many helpful comments. Ruodu Wang is supported by the Natural Sciences and Engineering Research Council of Canada (RGPIN-2018-03823, RGPAS-2018-522590).

\appendix
\section{Omitted proofs  in Section \ref{sec:points}}\label{app:proof_sec2}

\begin{proof}[Proof of Lemma \ref{lem:bound_1}]
As $\E[X]\le \VaR_{\epsilon_2}(X)$ and  $\epsilon_1<\epsilon_2$, $\E[X]\le \VaR_{\epsilon_2}(X)\le\VaR_{\epsilon_1}(X)$. By Proposition 1 in \cite{LW22}, $\Pi_{X}(\epsilon_1)<\infty$ and $\Pi_{X}(\epsilon_2)<\infty$.

For any $\epsilon\in (0,1)$ satisfying $\E[X]\le \VaR_{\epsilon}(X)$,
$$
\begin{aligned}
\epsilon \Pi_{X}(\epsilon)&=\epsilon\inf\{c\in [1,1/\epsilon]:\ES_{c\epsilon}(X)\le \VaR_{\epsilon}(X)\}\\
&=\inf\{\epsilon c\in [\epsilon,1]:\ES_{c\epsilon}(X)\le \VaR_{\epsilon}(X)\}\\
&=\inf\{k\in [\epsilon,1]:\ES_{k}(X)\le \VaR_{\epsilon}(X)\}.\\
\end{aligned}$$
Let $A(\epsilon)=\{k\in [\epsilon,1]:\ES_{k}(X)\le \VaR_{\epsilon}(X)\}$. For any $k \in A(\epsilon_2)$, we have $1\ge k\ge\epsilon_2>\epsilon_1$ and $\ES_{k}(X)\le \VaR_{\epsilon_2}(X)\le \VaR_{\epsilon_1}(X)$. Hence, $k \in A(\epsilon_1)$ and this gives $A(\epsilon_2) \subseteq A(\epsilon_1)$.
Therefore, $\epsilon_2\Pi_{X}(\epsilon_2)=\inf A(\epsilon_2)\ge \inf A(\epsilon_1)=\epsilon_1 \Pi_{X}(\epsilon_1)$.
\end{proof}

\begin{proof}[{Proof of Theorem \ref{thm:two_points}}]
We will check the equivalent condition \eqref{eq:VaR-ES} between $\VaR$ and $\ES$.
Note that if $t \mapsto \VaR_{t}(X)$ is a constant on $(0,\epsilon)$, then $\Pi_X(\epsilon)=1$. If  $t \mapsto \VaR_{t}(X)$ is not a constant on $(0,\epsilon)$, then $\Pi_X(\epsilon)$ is the unique solution that satisfies $\ES_{\epsilon\Pi_X(\epsilon)}(X)=\VaR_{\epsilon}(X)$.
\begin{enumerate}[(i)]
\item \underline{Case 1}, $c_2=1$.
It is clear that $\VaR_{t}(X)$ is a constant for $t \in (0,c_2\epsilon_2]$ and \eqref{eq:VaR-ES} is satisfied. Hence, $\Pi_{X}(\epsilon_2)=1.$
Moreover, $\VaR_{t}(X)$ is also a constant for $t \in (0,c_1\epsilon_1]$, which implies $\Pi_{X}(\epsilon_1)=1$.

\item \underline{Case 2}, $c_1=1$ and $1<c_2\le 1/\epsilon_2$.
For $t\in (0,\epsilon_1)$, $\VaR_{t}(X)=G_{\mathbf z}(t)$ is a constant for  $t \in (0,c_1\epsilon_1)$. Hence, $\Pi_{X}(\epsilon_1)=1$.
Next, we check whether $\ES_{c_2\epsilon_2}(X)=\VaR_{\epsilon_2}(X)$.
The value of $\ES_{c_2\epsilon_2}(X)$ is
\begin{align*}
&\ES_{c_2\epsilon_2}(X)\\
&=\frac{1}{c_2\epsilon_2}\left(\int_{0}^{\epsilon_1} \hat k \d\epsilon+\int_{\epsilon_1}^{\epsilon_2}( a_1\epsilon+b_1) \d\epsilon+\int_{\epsilon_2}^{c_2\epsilon_2}( a_2\epsilon+b_2) \d\epsilon \right)\\
&=\frac{1}{c_2\epsilon_2}\left(\epsilon_1\hat k+\frac{1}{2}a_1(\epsilon_2^2-\epsilon_1^2)+b_1(\epsilon_2-\epsilon_1)+\frac{1}{2}a_2(c_2^2\epsilon_2^2-\epsilon_2^2)+b_2(c_2\epsilon_2-\epsilon_2)\right)\\
&=\frac{1}{c_2\epsilon_2}\left(\frac{1}{2}a_1(\epsilon_2-\epsilon_1)^2+\hat k\epsilon_2+\frac{1}{2}a_2(c_2\epsilon_2-\epsilon_2)^2+\tilde k(c_2\epsilon_2-\epsilon_2)\right)\\
&=\frac{1}{c_2\epsilon_2}\left(\frac{1}{2}(\tilde k-\hat k)(\epsilon_2-\epsilon_1)+\hat k\epsilon_2+\frac{1}{2}(\tilde k-\hat k)(\epsilon_1+\epsilon_2)+\tilde k(c_2\epsilon_2-\epsilon_2)\right)=\tilde k
\end{align*}
The value of $\VaR_{\epsilon_2}(X)$ is $a_2\epsilon_2+b_2=\tilde k$.
Thus, \eqref{eq:VaR-ES} is satisfied. As $\VaR_{t}(X)$ is not a constant for $t \in (0,c_2\epsilon_2)$, we have $\Pi_{X}(\epsilon_2)=c_2$.

\item \underline{Case 3}, $1<c_1\le 1/\epsilon_1$ and $c_2=\frac{c_1\epsilon_1}{\epsilon_2}$.
In this case, we have
     $$\VaR_{\epsilon_1}(X)=G_{\mathbf z}(\epsilon_1)=k(\epsilon_1)=a\epsilon_2+b=G_{\mathbf{z}}(\epsilon_2)=\VaR_{\epsilon_2}(X)$$
     and  $\ES_{c_1\epsilon_1}(X)=\ES_{c_2\epsilon_2}(X)$ as $c_1\epsilon_1=c_2\epsilon_2$.
Thus, we only need to check whether $\ES_{c_2\epsilon_2}(X)=\VaR_{\epsilon_2}(X)$.
The value of $\ES_{c_2\epsilon_2}(X)$ is
\begin{align*}
\ES_{c_2\epsilon_2}(X)&=\frac{1}{c_2\epsilon_2}\left(\int_{0}^{\epsilon_1} k(\epsilon) \d\epsilon+\int_{\epsilon_1}^{\epsilon_2}k(\epsilon_1) \d\epsilon+\int_{\epsilon_2}^{c_2\epsilon_2} a\epsilon+b \d\epsilon\right)\\
&=\frac{1}{c_2\epsilon_2}\left(k+k(\epsilon_1)(\epsilon_2-\epsilon_1)+\frac{1}{2}a(c_2^2\epsilon_2^2-\epsilon_2^2)+(k(\epsilon_1)-a\epsilon_2)(c_2\epsilon_2-\epsilon_2)\right)\\
&=\frac{1}{c_2\epsilon_2}\left(k+k(\epsilon_1)(c_2\epsilon_2-\epsilon_1)+\frac{1}{2}a(c_2\epsilon_2-\epsilon_2)^2\right)\\
&=\frac{1}{c_2\epsilon_2}\left(k+k(\epsilon_1)(c_2\epsilon_2-\epsilon_1)+k(\epsilon_1)\epsilon_1-k\right) =k(\epsilon_1).
\end{align*}
The value of $\VaR_{\epsilon_2}(X)$ is also $k(\epsilon_1)$. Hence, \eqref{eq:VaR-ES} is satisfied and $\Pi_X(\epsilon_1)=c_1$, $\Pi_X(\epsilon_2)=c_2$ because $t\mapsto \VaR_{t}(X)$ is not a constant on $(0,\epsilon_1)$.

\item \underline{Case 4}, $1<c_1\le \epsilon_2/\epsilon_1$ and $1<c_2\le 1/\epsilon_2$.
     The first equivalent condition of \eqref{eq:VaR-ES} for  $\VaR_{\epsilon_1}(X)$ and $\ES_{c_1\epsilon_1}(X)$ is satisfied because $\VaR_{t}(X)=k(t)$ is the quantile function for $\mathrm{GPD}(\xi)$ with PELVE $c_1$ and $t \in (0, c_1\epsilon_1)$. Hence, we have $\Pi_X(\epsilon_1)=c_1$. Moreover, $\ES_{c_1\epsilon_1}(X)=\VaR_{\epsilon_1}(X)=k(\epsilon_1)$. We choose $a_1=k'(c_1\epsilon_1)$ and $b_1$ such that $a_1c_1\epsilon_1+b_1=k(c_1\epsilon_1)$.
For the equivalent condition between  $\ES_{c_2\epsilon_2}(X)$ and $\VaR_{\epsilon_2}(X)$, we can verify
\begin{align*}
& \ES_{c_2\epsilon_2}(X) =\frac{1}{c_2\epsilon_2}\left(\int_0^{c_1\epsilon_1}k(\epsilon) \d\epsilon+\int_{c_1\epsilon_1}^{\epsilon_2} a_1\epsilon+b_1\d\epsilon+\int_{\epsilon_2}^{c_2\epsilon_2} (a_2\epsilon+b_2) \d\epsilon
\right)\\
&=\frac{1}{c_2\epsilon_2}\left(c_1\epsilon_1k(\epsilon_1)+\frac{1}{2}a_1(\epsilon_2^2-c_1^2\epsilon_1^2)+b_1(\epsilon_2-c_1\epsilon_1)+\frac{1}{2}a_2\left(c_2^2\epsilon_2^2-\epsilon_2^2\right)+b_2\left(c_2\epsilon_2-\epsilon_2\right)
\right)\\
&=\frac{1}{c_2\epsilon_2}\left(c_1\epsilon_1k(\epsilon_1)+\frac{1}{2}a_1(2c_2\epsilon_2^2-\epsilon_2^2-c_1^2\epsilon_1^2)+b_1(c_2\epsilon_2-c_1\epsilon_1)+\frac{1}{2}a_2\left(c_2\epsilon_2-\epsilon_2\right)^2
\right)\\
&=\frac{1}{c_2\epsilon_2}\left(a_1c_2\epsilon_2^2+b_1c_2\epsilon_2\right)=a_1\epsilon_2+b_1=\VaR_{\epsilon_2}(X).
\end{align*}
Thus, \eqref{eq:VaR-ES} is satisfied and we have $\Pi_{X}(\epsilon_2)=c_2$.

\item \underline{Case 5}, $\epsilon_2/\epsilon_1<c_1\le 1/\epsilon_1$ and $\frac{c_1\epsilon_1}{\epsilon_2}<c_2\le1/\epsilon_2$.
The equality between $\VaR_{\epsilon_1}(X)$ and $\ES_{c_1\epsilon_1}(X)$ can be checked by
\begin{align*}
\ES_{c_1\epsilon_1}(X)&=\frac{1}{c_1\epsilon_1}\left(\int_{0}^{\epsilon_1}k(\epsilon) \d \epsilon+\int_{\epsilon_1}^{\epsilon_2} (a_1\epsilon+b_1)\d\epsilon+(a_1\epsilon_2+b_1)(c_1\epsilon_1-\epsilon_2)
\right)\\
&=\frac{1}{c_1\epsilon_1}\left(k+ \frac{1}{2}a_1(\epsilon_2^2-\epsilon_1^2)+(k(\epsilon_1)-a_1\epsilon_1)(\epsilon_2-\epsilon_1)+(a_1\epsilon_2+b_1)(c_1\epsilon_1-\epsilon_2)
\right)\\
&=\frac{1}{c_1\epsilon_1}\left(k+ a_1(\epsilon_2-\epsilon_1)(c_1\epsilon_1-1/2(\epsilon_2+\epsilon_1))+k(\epsilon_1)(c_1\epsilon_1-\epsilon_1)
\right)\\
&=\frac{1}{c_1\epsilon_1}\left(k+ k(\epsilon_1)\epsilon_1-k+k(\epsilon_1)(c_1\epsilon_1-\epsilon_1)
\right) =k(\epsilon_1)=G_{\mathbf z}(\epsilon_1)=\VaR_{\epsilon_1}(X).
\end{align*}
The equality between $\VaR_{\epsilon_2}(X)$ and $\ES_{c_2\epsilon_2}(X)$ can be checked by
\begin{align*}
\ES_{c_2\epsilon_2}(X)&=\frac{1}{c_2\epsilon_2}\left(\int_{0}^{c_1\epsilon_1}k(\epsilon) \d \epsilon+\int_{c_1\epsilon_1}^{c_2\epsilon_2} (a_2\epsilon+b_2)\d\epsilon
\right)\\
&=\frac{1}{c_2\epsilon_2}\left(c_1\epsilon_1 k(\epsilon_1)+\frac{1}{2}a_2(c_2^2\epsilon_2^2-c_1^2\epsilon_1^2)+b_2(c_2\epsilon_2-c_1\epsilon_1)\right)\\
&=\frac{1}{c_2\epsilon_2}\left(c_1\epsilon_1 k(\epsilon_1)+\frac{1}{2}a_2(c_2\epsilon_2-c_1\epsilon_1)^2+(a_1\epsilon_2+b_1)(c_2\epsilon_2-c_1\epsilon_1)\right)\\
&=\frac{1}{c_2\epsilon_2}\left(c_1\epsilon_1 k(\epsilon_1)+c_1\epsilon_1(a_1\epsilon_2+b_1-k(\epsilon_1))+(a_1\epsilon_2+b_1)(c_2\epsilon_2-c_1\epsilon_1)\right)\\
&=a_1\epsilon_2+b_1=G_{\mathbf z}(\epsilon_2)=\VaR_{\epsilon_2}(X)
\end{align*}
Hence, \eqref{eq:VaR-ES} is satisfied, and $\Pi_X(\epsilon_1)=c_1$ and $\Pi_X(\epsilon_2)=c_2$.
\end{enumerate}
Therefore,  it is checked that $X$ satisfies $\Pi_{X}(\epsilon_1)=c_1$ and $\Pi_{X}(\epsilon_2)=c_2$ for all five cases.
\end{proof}

The following propositions address  the issue discussed in Remark \ref{rem:cannot}  by showing that the boundary cases of $(\epsilon_1,c_1,\epsilon_2,c_2) $ cannot be achieved by strictly decreasing quantile functions, and hence our construction of quantiles with a flat region in Figure \ref{fig:calibration} are needed.

\begin{proposition}\label{lem:low_bound}
For any $X \in L^1$, let $\epsilon_1,\epsilon_2 \in (0,1)$ be such that $\E[X]\le \VaR_{\epsilon_2}(X)$ and  $\epsilon_1<\epsilon_2$. Then, $\Pi_{X}(\epsilon_2)=\max\left\{1,\Pi_{X}(\epsilon_1)\epsilon_1/\epsilon_2\right\}$ if and only if $\VaR_{\epsilon_1}(X)=\VaR_{\epsilon_2}(X)$.
\end{proposition}

\begin{proof}
Using the same logic as in Lemma \ref{lem:bound_1}, we have that $ \Pi_{X}(\epsilon_1)$ and $\Pi_{X}(\epsilon_2)$ are finite.

We first show the ``if" statement.
Assume $\VaR_{\epsilon_1}(X)=\VaR_{\epsilon_2}(X)$. As $\VaR_{\epsilon}(X)$ is decreasing, we know that $\VaR_{\epsilon}(X)$ is a constant on $[\epsilon_1,\epsilon_2]$.

If $\VaR_{\epsilon}(X)=\VaR_{\epsilon_1}(X)$ for $\epsilon\in (0,\epsilon_1)$, then $\VaR_{\epsilon}(X)$ is a constant on $(0,\epsilon_2]$. Therefore, we can get $\Pi_{X}(\epsilon_1)=\Pi_{X}(\epsilon_2)=1$. Note that $\Pi_{X}(\epsilon_1)\epsilon_1/\epsilon_2=\epsilon_1/\epsilon_2<1$. Thus,  we obtain $\Pi_{X}(\epsilon_2)=\max\left\{1,\Pi_{X}(\epsilon_1)\epsilon_1/\epsilon_2)\right\}$.

If there exists $\epsilon \in (0,\epsilon_1)$ such that $\VaR_{\epsilon}(X)>\VaR_{\epsilon_1}(X)$, then $\ES_{\epsilon}(X)$ is strictly decreasing on $[\epsilon_1,1]$.
By the equivalent condition between VaR and ES, $\VaR_{\epsilon_1}(X)=\VaR_{\epsilon_2}(X)$ implies $\ES_{\epsilon_1\Pi_{X}(\epsilon_1)}(X)=\ES_{\epsilon_2\Pi_{X}(\epsilon_2)}(X)$.
 Thus, $\epsilon_1 \Pi_{X}(\epsilon_1)=\epsilon_2 \Pi_{X}(\epsilon_2)$.
 Furthermore, we have $$\VaR_{\epsilon_1\Pi_{X}(\epsilon_1)}(X)<\ES_{\epsilon_1\Pi_{X}(\epsilon_1)}(X)=\VaR_{\epsilon_1}(X)=\VaR_{\epsilon_2}(X).$$
 Thus, $\epsilon_1\Pi_{X}(\epsilon_1)>\epsilon_2$ and we get $\Pi_{X}(\epsilon_2)=\max\left\{1,\Pi_{X}(\epsilon_1)\epsilon_1/\epsilon_2\right\}$.

 Next, we show the ``only if" statement.  Assume $\Pi_{X}(\epsilon_2)=\max\left\{1,\Pi_{X}(\epsilon_1)\epsilon_1/\epsilon_2\right\}$.

 If $\Pi_{X}(\epsilon_2)=1$, then $\VaR_{\epsilon_2}(X)=\ES_{\epsilon_2}(X)$. This implies that $\VaR_{\epsilon}(X)$ is a constant on $(0,\epsilon_2]$, which gives $\VaR_{\epsilon_1}(X)=\VaR_{\epsilon_2}(X)$.

 If $\Pi_{X}(\epsilon_2)=\Pi_{X}(\epsilon_1)\epsilon_1/\epsilon_2$, then $\epsilon_2 \Pi_{X}(\epsilon_2)=\epsilon_1\Pi_{X}(\epsilon_1)$.
Hence, we have $$\VaR_{\epsilon_1}(X) =\ES_{\epsilon_1\Pi_{X}(\epsilon_1)}(X)=\ES_{\epsilon_2\Pi_{X}(\epsilon_2)}(X)=\VaR_{\epsilon_2}(X).$$
Thus, we complete the proof.
\end{proof}

\begin{proposition}\label{pro:c_2_bound}
For any $X \in L^1$, let $\epsilon_1,\epsilon_2 \in (0,1)$ be such that $\E[X]\le \VaR_{\epsilon_2}(X)$ and  $\epsilon_1<\epsilon_2$. Let $c_1=\Pi_{X}(\epsilon_1)$ and $c_2=\Pi_{X}(\epsilon_2)$.
If $\VaR_{\epsilon_1}(X)>\VaR_{\epsilon_2}(X)$, then
$$
 \hat c \!\le\! c_2\!\le\! \left\{\begin{aligned} &\!\min\left\{\!\frac{1}{\epsilon_2},\frac{c_1\epsilon_1}{\epsilon_2}\left(\frac{\VaR_{\epsilon_1}(X)-\VaR_{c_1\epsilon_1}(X)}{\VaR_{\epsilon_2}(X)-\VaR_{c_1\epsilon_1}(X)}\right) \!\right\}, & \!\!\VaR_{c_1\epsilon_1}(X) \!<\! \VaR_{\epsilon_2}(X),\\
 &\!\frac{1}{\epsilon_2}, &\!\!\VaR_{c_1\epsilon_1}(X) \!\ge\! \VaR_{\epsilon_2}(X),
 \end{aligned}\right.
$$
where
$$
\hat c=\inf\left \{  t \in (
1,1/\epsilon_2] : \frac{ (t \epsilon_2-c_1\epsilon_1)\left(\VaR_{\epsilon_2}(X)-\VaR_{t \epsilon_2}(X)\right)}{c_1\epsilon_1\left(\VaR_{\epsilon_1}(X)-\VaR_{\epsilon_2}(X)\right)}\ge 1\right\}.
 $$
 Moreover, $\hat c\ge\max\{1,c_1\epsilon_1/\epsilon_2\}$.
\end{proposition}

\begin{proof}
As $\E[X]\le \VaR_{\epsilon_2}(X)$, we have $c_1<\infty$ and $c_2<\infty$.
 By definition,  $c_2\le 1/\epsilon_2$.
From Lemma \ref{lem:low_bound}, we get $c_2>\max\left\{1,c_1\epsilon_1/\epsilon_2\right\}$.
Thus, the value of $c_2$ should be in  $\left(\max\left\{1,c_1\epsilon_1/\epsilon_2\right\}, {1/\epsilon_2}\right]$.

Note that $c_1, \epsilon_1, c_2, \epsilon_2$ satisfy the equivalent condition \eqref{eq:VaR-ES}.
We can rewrite \eqref{eq:VaR-ES} as
$$
\int_{0}^{c_1\epsilon_1} \VaR_{\epsilon}(X) \d\epsilon=c_1\epsilon_1\VaR_{\epsilon_1}(X)~~~\text{and} ~~~\int_{0}^{c_2\epsilon_2} \VaR_{\epsilon}(X) \d\epsilon=c_2\epsilon_2\VaR_{\epsilon_2}(X).
$$
Therefore, we have
$$\int_{c_1\epsilon_1}^{c_2\epsilon_2}\VaR_{\epsilon}(X) \d\epsilon=c_2\epsilon_2\VaR_{\epsilon_2}(X)-c_1\epsilon_1\VaR_{\epsilon_1}(X).$$
Furthermore, by the monotonicity of $\VaR$, we have
$$(c_2\epsilon_2-c_1\epsilon_1)\VaR_{c_2\epsilon_2}(X)\le \int_{c_1\epsilon_1}^{c_2\epsilon_2}\VaR_{\epsilon}(X) \d\epsilon \le (c_2\epsilon_2-c_1\epsilon_1)\VaR_{c_1\epsilon_1}(X).$$
The two inequality will provide an upper bound and a lower bound for $c_2$.

\underline{An upper bound on $c_2$}.
Using $c_2\epsilon_2\VaR_{\epsilon_2}(X)-c_1\epsilon_1\VaR_{\epsilon_1}(X)\le(c_2\epsilon_2-c_1\epsilon_1)\VaR_{c_1\epsilon_1}(X),$
we have
\begin{equation}\label{eq:neq_u}
c_2\epsilon_2 \left(\VaR_{\epsilon_2}(X)-\VaR_{c_1\epsilon_1}(X)\right)\le c_1\epsilon_1\left(\VaR_{\epsilon_1}(X)-\VaR_{c_1\epsilon_1}(X)\right).
\end{equation}
If $\VaR_{c_1\epsilon_1}(X)\ge \VaR_{\epsilon_2}(X)$, the left side of \eqref{eq:neq_u} is less or equal to 0 and the right side of \eqref{eq:neq_u} is larger or equal to 0 because $\VaR_{\epsilon_1}(X)\ge \VaR_{c_1\epsilon_1}(X)$. Therefore, \eqref{eq:neq_u} is satisfies for any $c_2 \in \left(\max\left\{1,c_1\epsilon_1/\epsilon_2\right\}, {1/\epsilon_2}\right]$. The upper bound for $c_2$ is unchanged.

On the other hand, if $\VaR_{c_1\epsilon_1}(X)<\VaR_{\epsilon_2}(X)$,
we have
$$c_2\le \frac{c_1\epsilon_1}{\epsilon_2 }\left(\frac{\VaR_{\epsilon_1}(X)-\VaR_{c_1\epsilon_1}(X)}{\VaR_{\epsilon_2}(X)-\VaR_{c_1\epsilon_1}(X)}\right).$$
Thus, an upper bound for $c_2$ is $\min\left\{\frac{1}{\epsilon_2},\frac{c_1\epsilon_1}{\epsilon_2 }\left(\frac{\VaR_{\epsilon_1}(X)-\VaR_{c_1\epsilon_1}(X)}{\VaR_{\epsilon_2}(X)-\VaR_{c_1\epsilon_1}(X)}\right)\right\}$.


\underline{A lower bound on $c_2$}.
It holds that
$$
(c_2\epsilon_2-c_1\epsilon_1)\VaR_{c_2\epsilon_2}(X)\le c_2\epsilon_2\VaR_{\epsilon_2}(X)-c_1\epsilon_1\VaR_{\epsilon_1}(X).
$$
Subtracting $(c_2\epsilon_2-c_1\epsilon_1)\VaR_{\epsilon_2}(X)$ from both sides, we get
\begin{equation}\label{eq:lower_bound}
(c_2\epsilon_2-c_1\epsilon_1)\left(\VaR_{\epsilon_2}(X)-\VaR_{c_2\epsilon_2}(X)\right)\ge c_1\epsilon_1\left(\VaR_{\epsilon_1}(X)-\VaR_{\epsilon_2}(X)\right).
\end{equation}
For $t \in (0,1/\epsilon_2)$, let
$$f(t)= (t\epsilon_2-c_1\epsilon_1)\left(\VaR_{\epsilon_2}(X)-\VaR_{t\epsilon_2}(X)\right).$$
As we can see, $f(1)=0$, $f(c_1\epsilon_1/\epsilon_2)=0$ and $f(t)\le 0$
if
$t \in[\min\{1,c_1\epsilon_1/\epsilon_2\},\max\{1,c_1\epsilon_1/\epsilon_2\}]$.  The $f$ is increasing in the interval $(\max\{1,c_1\epsilon_1/\epsilon_2\}, 1/\epsilon_2)$, decreasing in $(0,\min\{1,c_1\epsilon_1/\epsilon_2\})$.
Hence, by \eqref{eq:lower_bound}, the lower bound for  $c_2$ is
$$\begin{aligned}
\hat c=\inf \left\{ t\in (1, 1/\epsilon_2]: \frac{(t\epsilon_2-c_1\epsilon_1)\left(\VaR_{\epsilon_2}(X)-\VaR_{t\epsilon_2}(X)\right)}{c_1\epsilon_1\left(\VaR_{\epsilon_1}(X)-\VaR_{\epsilon_2}(X)\right)}\ge 1 \right\}. 
\end{aligned}$$ 
As $c_1\epsilon_1\left(\VaR_{\epsilon_1}(X)-\VaR_{\epsilon_2}(X)\right)>0$, we have  $\hat c\ge\max\{1,c_1\epsilon_1/\epsilon_2\}$.
\end{proof}

\section{Omitted proofs  in Section \ref{sec:constant}}\label{app:proof_sec5}
\begin{proof}[Proof of Theorem \ref{thm:constant}]

By Proposition \ref{prop:abstract}, for any $X \in \X$, we can find $f \in \mathcal{C}$ satisfying \eqref{eq:abstract}  such that $z_f(y)=1/\pi_{X}(y)=c$ and $X=f(U)$. As $z(y)=c$ is a continuously differentiable function, we know that all such $f$ is characterized by the advanced differential equation \eqref{eq:DDE}.
First, we show for any strictly decreasing solution $f$ to \eqref{eq:DDE} can be represented as
$$ f\left(y\right)  =C_{0}+C_{1}y^{\alpha}+O\left(y^{\zeta}\right).$$

Let us start with \eqref{eq:DDE}. If $z(y)=c$, we need to solve $f$ from
$$f(y)  =f(cy)+cyf'(cy),\quad y\in(0,1].$$
Even though $f$ in the first place is considered on $(0,1]$,   given that $c<1$, and this final equation, one can expand it to the whole positive line:
$$f(y) =f(cy)+cyf'(cy),\quad y>0.$$
Next, let $x\left(t\right)=e^{-t}f\left(e^{-t}\right)$ for $t\in\R$ and $a=-\log\left(c\right)>0$. This is equivalent to say that $f\left(y\right)=x\left(-\log\left(y\right)\right)/y$.
This changing variable simply gives the following delayed differential equation:
$$x'\left(t\right) =-e^{-a}x\left(t-a\right),\,\,\,t\in\mathbb{R}.\label{eq:DDE-1}$$
Since we have assumed that $f$ is strictly decreasing, i.e., $f'<0$, we have an extra restriction on $x$. Note that
$$x'\left(t\right) =-e^{-t}f\left(e^{-t}\right)-e^{-2t}f'\left(e^{t}\right)=-x\left(t\right)-e^{-2t}f'\left(e^{t}\right).$$
Thus, we have $f'<0\Leftrightarrow x'+x>0$. Therefore, we are looking for a solution to the following delay differential equation (DDE):
\begin{align}
\begin{cases}
x'\left(t\right)=-e^{-a}x\left(t-a\right),\\
x'\left(t\right)+x\left(t\right)>0,
\end{cases}t\in\R & .\label{eq:x}
\end{align}
A standard approach of finding the solutions is to assume that they are in the form of a characteristic function $t\mapsto e^{mt}$. Putting this solution inside the equation, we get
\begin{align}
    me^{mt} & =-e^{-a}e^{m\left(t-a\right)}
    ~~\Longrightarrow~~   ame^{am}=\left(-a\right)e^{\left(-a\right)}.\nonumber
\end{align}
This means any solution is given by $x\left(t\right)=e^{mt}$ where $m$ solves the characteristic equation
\begin{equation}
    l\left(ma\right)=l\left(-a\right),\label{eq:roots-1}
\end{equation}
where $l\left(x\right)=xe^{x}$. Let $b=l(-a)$. As $a>0$, we have $b \in [-1/e,0)$.
This equation has one obvious real solution at $m_{1}=-1$.
To find $m$, we need to know about the inverse of $l$.
The inverse of the function $l$ is known as the Lambert $W$ function and plays an essential role in solving delayed and advanced differential equations.

From the Lambert $W$ function, we know that $l\left(z\right)=ze^{z}=b$ has two real solutions when $b \in (-1/e, 0)$ and one real solution when $b=-1/e$.
As illustrated by Figure \ref{Lambert}, if $0<c<1/e$, the two real solutions are  $z_1=-a<-1$ and  $z_2=m_2 a>-1$; thus, $-1<m_2<0$.  If $0<c<1/e$, the two real solutions are  $-1<z_1=-a<0$ and  $z_2=m_2 a<-1$; thus, $m_2<-1$.  If $c=1/e$, there is only one real solution $z_1=z_2=-1$; thus $m_2=m_1=-1$.
\begin{figure}[htbp]
    \centering{}\includegraphics[scale=0.6]{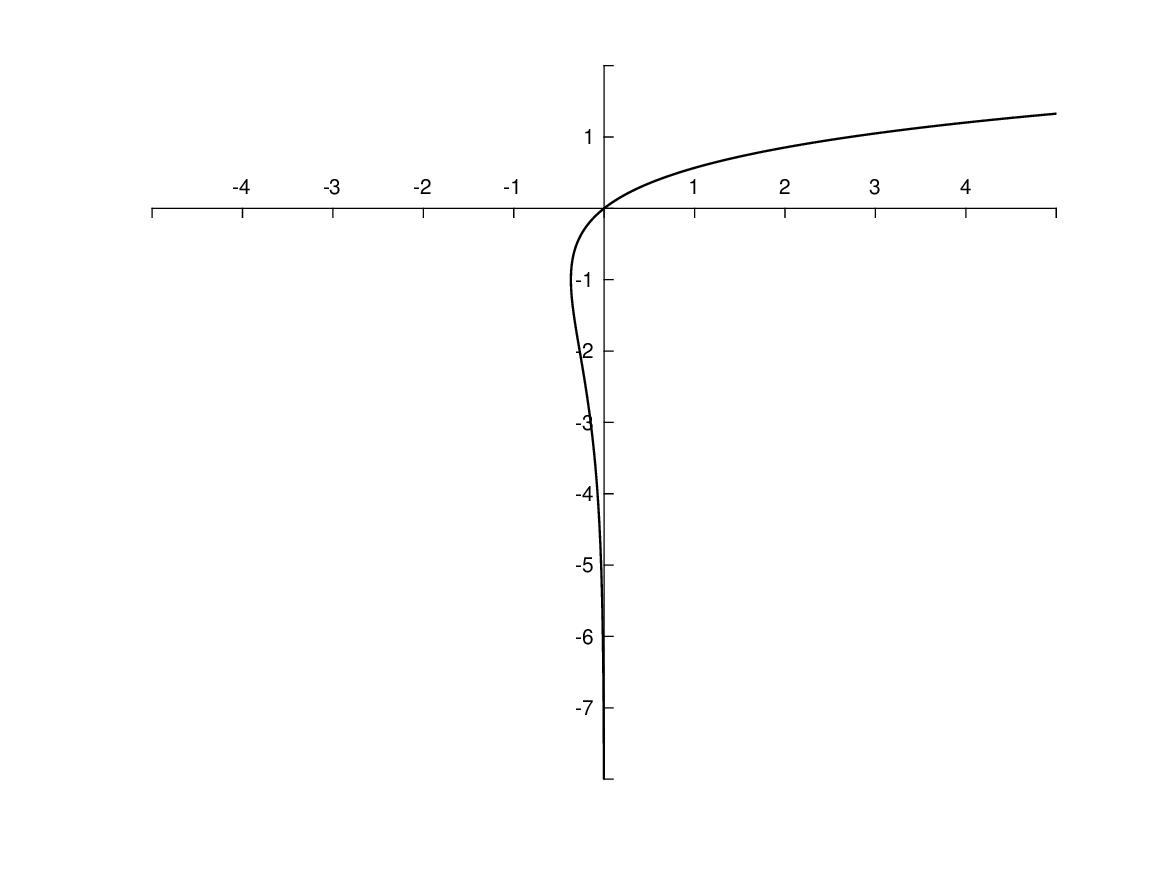} \caption{Lambert $W$ function.}
    \label{Lambert}
\end{figure}

It is important to note that in general an equation like $l\left(z\right)=ze^{z}=b$ has infinite  complex roots. Let $z=\theta+i\eta$, and $b\in\left[-1/e,0\right)$.
In that regard, we have
\begin{align*}
    b & =ze^{z}=\left(\theta+i\eta\right)e^{\theta+i\eta}\\
     & =\left(\theta+i\eta\right)\left(\cos\left(\eta\right)+i\sin\left(\eta\right)\right)\\
     & =\left(\theta\cos\left(\eta\right)-\eta\sin\left(\eta\right)\right)+i\left(\theta\sin\left(\eta\right)+\eta\cos\left(\eta\right)\right).
\end{align*}
This implies that $\theta\sin\left(\eta\right)+\eta\cos\left(\eta\right)=0$, and $b=e^{\theta}\left(\theta\cos\left(\eta\right)-\eta\sin\left(\eta\right)\right)$, leading to
\begin{align}
    \eta=0, b=\theta e^{\theta}\nonumber ~~~\text{ or }~~~
    \theta=-\frac{\eta}{\tan\left(\eta\right)}, b=-\frac{\eta\exp\left(-\frac{\eta}{\tan\left(\eta\right)}\right)}{\sin\left(\eta\right)}.\label{eq:roots}
\end{align}
 We plot the curves $b=\theta e^{\theta}$ and $\left(-\frac{\eta\exp\left(-\frac{\eta}{\tan\left(\eta\right)}\right)}{\sin\left(\eta\right)},-\frac{\eta}{\tan\left(\eta\right)}\right)$ to find out the relation between $b$ and the real part of the solution in Figure \ref{complex}. The $x$-axis is $b$ and the $y$-axis is $\theta$. The blue curve is associated with $b=\theta e^{\theta}$, which is essentially the principle branch of the Lambert $W$ function.
For any $b$, one can find the real values of the roots by fixing $b$.
For instance, the green dashed line is associated with $b=-0.12$.
As one can see, the curves intersect this line in infinite negative values.
  For $b\in [-1/e,0)$, we can see that the real roots are greater than the real part of the complex roots. For more explanation of this, see \cite{SB73}.
\begin{figure}[htpb]
    \centering{}\includegraphics[scale=0.7]{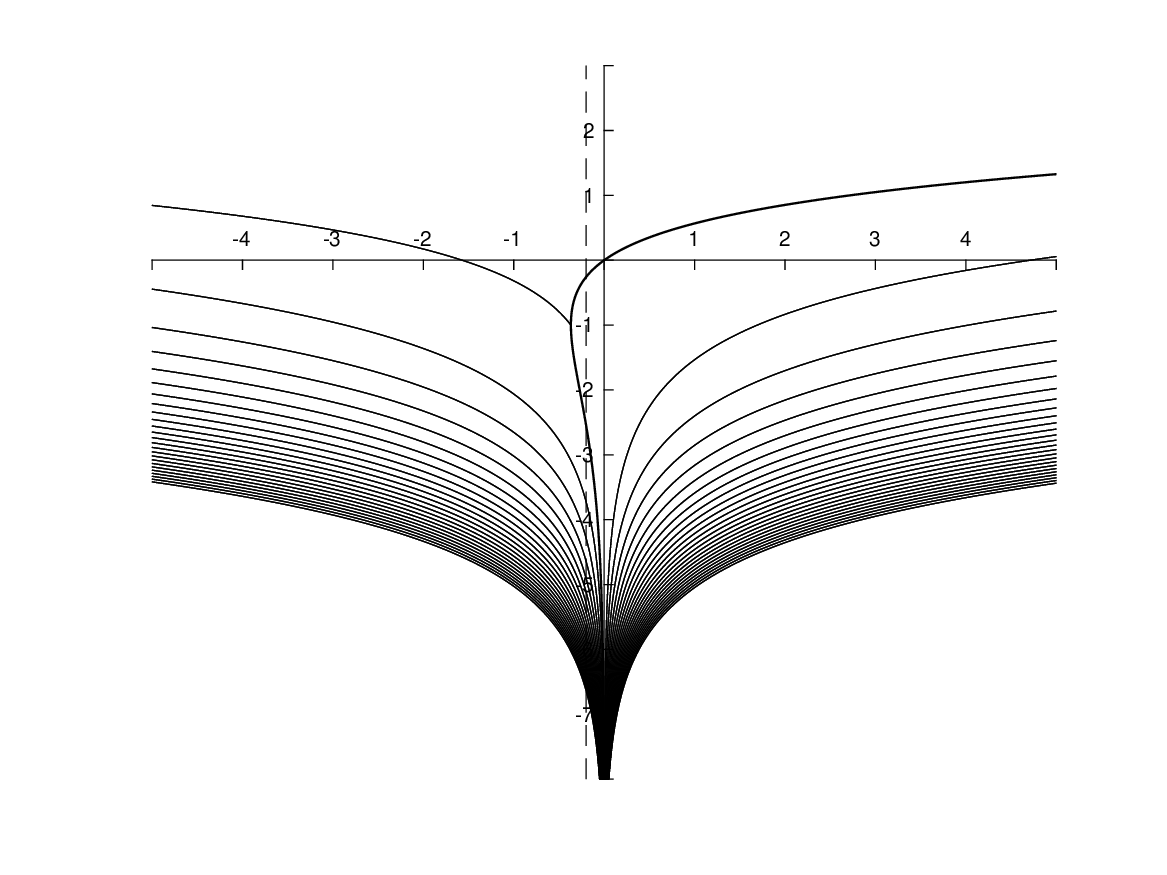} \caption{The real part of the Lambert $W$ roots.}
    \label{complex}
\end{figure}

Now  assume that all the complex solutions for $\left(am\right)e^{am}=\left(-a\right)e^{\left(-a\right)}$
 are $m_k=\lambda_k+\sigma_k i$ for $k=1,2,3,\dots$, where $(\lambda_1,\sigma_1)=(-1,0)$ and $(\lambda_2,\sigma_2)=(m_2,0)$.
Based on the above discussions, we have
$$\lambda_1=-1>\lambda_2=m_2>\lambda_{3}>\lambda_{4}>\dots \text{ when } c \in (1/e,1),$$
$$0>\lambda_2=m_2>\lambda_1=-1>\lambda_{3}>\lambda_{4}>\dots \text{ when }c \in (0,1/e),$$
and
$$\lambda_1=\lambda_2=-1>\lambda_{3}>\lambda_{4}>\dots \text{ when }c=1/e.$$

Let $$\mathbb{C}=\left\{\begin{pmatrix}
\lambda & -\sigma\\
\sigma  & \lambda\\
\end{pmatrix}=\lambda+\sigma i  \mid \lambda,\sigma \in \R\right\}$$ be the set of all complex numbers.
Then,
$$x_\mathbb{C}(t)=\exp\left\{\begin{pmatrix}
\lambda & -\sigma\\
\sigma  & \lambda\\
\end{pmatrix}t\right\}$$
 is a complex solution that solves \eqref{eq:x}. It is clear that \eqref{eq:x} still holds for the linear transform of $x(t)$. Therefore, for any two $2 \times 1$ vector $A$ and $B$,
 $$x(t)=A' x_{\mathbb{C}}(t) B=e^{\lambda t}\left(C_1 \cos(\sigma t)+C_2\sin(\sigma t)
\right)$$
is also a solution to \eqref{eq:x}.

In \cite{BC63}, it is shown that all complex solutions to \eqref{eq:x} can be represented as follows:
$$x_{\mathbb{C}}\left(t\right) =\sum_{k=1}^{\infty}C_{k}e^{m_{k}t}.
$$
Putting it in the real-valued context, we have that all real-valued solutions are in the following form:
$$x\left(t\right) =C_{1}e^{-t}+C_{2}e^{m_2 t}+\sum_{k=3}^{\infty}e^{\lambda_{k}t}\left(C_{k,1}\cos\left(\sigma_{k}t\right)+C_{k,2}\sin\left(\sigma_{k}t\right)\right).$$
Now let us check $x'+x>0$. This means that $C_2$, $\left\{ C_{k,1}\right\} _{k\ge 3}$ and $\left\{ C_{k,2}\right\} _{k\ge 3}$ for all $t$ must satisfy
\begin{multline*}
    x'\left(t\right)+x\left(t\right)=\left(1+m_2\right)C_{2}e^{m_2 t}\\
    +\sum_{k=3}^{\infty}e^{\lambda_{k}t}\left(\left(\lambda_{k}C_{k,1}+\sigma_{k}C_{k,2}+C_{k,1}\right)\cos\left(\sigma_{k}t\right)+\left(\lambda_{k}C_{k,2}-\sigma_{k}C_{k,1}+C_{k,2}\right)\sin\left(\sigma_{k}t\right)\right)>0.
\end{multline*}
As $\lambda_k<m_2$ for $k\ge 3$, we have $\lim_{t \to \infty} \left(x'(t)+x(t)\right)/e^{m_2 t}=(1+m_2)C_2$. Therefore, we need
$C_2<0$ if $m_2<-1$ or $C_2>0$ if $m_2>-1$. That is $C_2(1+m_2)>0$.
Then the solution can be written as
$$x\left(t\right) =C_{1}e^{-t}+C_{2}e^{m_2 t}+O\left(e^{\lambda_{3}t}\right).$$
By a change of variable, we get
$$
f\left(y\right) =C_{1}+C_{2}y^{\alpha}+O\left(y^{\zeta}\right),
$$
where $\alpha=-\left(1+m_2\right)$ and $\zeta=-\left(1+\lambda_{3}\right)$. As $\lambda_3<\min\{-1,m_2\}$ and $C_2(1+m_2)>0$, we have $\zeta>\max\{0,\alpha\}$ and $C_2\alpha<0$.
Also note that since $m_2$ solves $\left(am\right)e^{am}=\left(-a\right)e^{\left(-a\right)}$, by replacing $m_2=-1-\alpha$ and  $a=-\log\left(c\right)$, we get $ {(\alpha+1)^{-1/\alpha}}=c$.
\end{proof}

\section{Omitted proofs in Section \ref{sec:property}}\label{app:proof_sec6}

\begin{proof}[Proof of Proposition \ref{property}]
Before proving the statements in the proposition,
we   introduce the following linking function, for $X\in\X$,
$$
    \Gamma_{X}(\epsilon)=1-F_{X}(\ES_{\epsilon}(X)),~~~~\epsilon\in[0,1].
$$
As $X \in X$ 
it is easy to check that $\Gamma_{X}$ {satisfies Assumption \ref{ass:1} }
for any $X \in \X$. The domain of $\Gamma_{X}(\epsilon)$ is $[0,1]$ and its range is $[0,1-F_{X}(\E[X])]$.

As $\ES_{\epsilon}(X)=\VaR_{\Gamma_X(\epsilon)}(X)$, we have
$$
\VaR_{\Gamma_{X}(\epsilon)}(X)=\ES_{\epsilon}(X)=\VaR_{\epsilon/\pi_{X}(\epsilon)}(X)\quad\text{for}\quad\epsilon\in(0,1].
$$
Hence we have the simple relationship $\Gamma_{X}(\epsilon)=\epsilon/\pi_{X}(\epsilon).$
Therefore, $\Pi_{X}(\epsilon)=\pi_{X}\ensuremath{\left(\Gamma_{X}^{-1}(\epsilon)\right)}$ and $\pi_{X}(\epsilon)=\Pi_{X}(\Gamma_{X}(\epsilon))$.
The function $\Gamma_X$ yields an association  between a point on the  PELVE on $(0,1-F_{X}(\E[X])]$ and a point on the dual PELVE curve on $(0,1]$ with the same value. Furthermore, we have $\pi_X$ is continuous on $(0,1)$  as $\pi_X(\epsilon)=\epsilon/\Gamma_X(\epsilon)$ and $\Gamma_X$ is continuous.

 Next, we show the statements (i)-(iv).
The equivalence (i) of monotonicity of $\Pi_{X}$ and that of $\pi_X(\cdot)$ follows from $\Pi_{X}(\epsilon)=\pi_{X}\ensuremath{\left(\Gamma_{X}^{-1}(\epsilon)\right)}$, $\pi_{X}(\epsilon)=\Pi_X(\Gamma_{X}(\epsilon))$ and that $\Gamma_{X}$ is increasing.

For (ii), (iii) and (iv), we first show that $\Gamma_X$ is location-scale invariant and shape relevant (in the sense of \eqref{eq:shape-relevance}).
Assume that $f:\R \to\R$ is a strictly increasing concave function such that $f(X) \in \X$.
By Jensen's inequality and the dual representation of $\ES_p$, we have
$$\ES_p\left(f(X)\right)\le f\left(\ES_p(X)\right)$$ for all $p \in (0,1)$. This statement can be found in Appendices A.2 in \cite{LW22}.
Therefore,
\begin{equation}  \label{eq:shape-relevance}
\begin{array}{rl}
  &   \Gamma_{f(X)}(\epsilon)=1-F_{f(X)}\left(\ES_{\epsilon}\left(f(X)\right)\right) \\
                           &\qquad \qquad\ge 1-F_{f(X)}\left(f\left(\ES_{\epsilon}(X)\right)\right)
                         =1-F_{X}\left(\ES_{\epsilon}(X)\right)=\Gamma_{X} (\epsilon ).
                         \end{array}
                         \end{equation}
Then, we have $\Gamma_{f(X)}(\epsilon)\ge\Gamma_{X}(\epsilon)$ for all strictly increasing concave functions: $f:\R \to\R$ with $f(X) \in \X$.

For any strictly increasing convex function $g: \R \to \R$ with $g(X) \in \X$, we can take $f(x)=g^{-1}(X)$, which is  a strictly increasing concave function. Therefore, we have $ \Gamma_{g(X)}(\epsilon)\le \Gamma_{X}(\epsilon) $ for all strictly increasing convex functions $g$.

For $\lambda>0$ and $a\in \R$, we have that $f(x)=\lambda X+a$ is both convex and concave. Therefore,   $\Gamma_{\lambda X+a}(\epsilon)=\Gamma_X(\epsilon)$ for all $\epsilon \in [0,1]$. In conclusion, we have the following results for $\Gamma$.
    \begin{enumerate}[(1)]
     \item For all $\lambda>0$ and $a\in\R$, $\Gamma_{\lambda X+a}(\epsilon)=\Gamma_{X} (\epsilon)$.
     \item $\Gamma_{f(X)}(\epsilon)\ge\Gamma_{X}(\epsilon)$ for all strictly increasing concave functions: $f:\R \to\R$ with $f(X) \in \X$.
      \item $\Gamma_{g(X)}(\epsilon)\le\Gamma_{X}(\epsilon)$ for all strictly increasing convex functions: $g:\R \to\R$ with $g(X) \in \X$.
    \end{enumerate}
    Then, we have (ii), (iii) and (iv) from $\pi_X(\epsilon)=\epsilon/\Gamma_X(\epsilon)$.
\end{proof}

\begin{proof}[Proof of Theorem \ref{hazard_rate}]
    The idea is to prove that if $1/\eta$ is convex (concave), then $x\mapsto F^{-1}((1-p)F(x)+p)$ is convex (concave) for all $p\in(0,1)$. Then, we can get the desired result by Proposition \ref{convex}. We will use the following steps to show this statement.

    \underline{Step 1.}
    Let $s(x)=\log\left(1-F(x)\right)$ for $x\in(\essinf(X),\esssup(X))$.
    Then, $s$ is a continuous and strictly decreasing function and $s(x)<0$.
    Let $s^{-1}$ be the inverse function of $s$. Now, we have
    $$ F(x)=1-e^{s(x)},\quad x\in(\essinf(X),\esssup(X)) $$
    and
    $$ F^{-1}(t) =s^{-1}(\log(1-t)),\quad t\in(0,1).$$
    Therefore,
    $$
    \begin{aligned}
     F^{-1}\big((1-p)F(x)+p\big)
     & =s^{-1}\big(\log(1-(1-p)F(x)-p)\big)\\
     & =s^{-1}\Big(\log\big(1-F(x)\big)+\log(1-p)\Big)\\
     & =s^{-1}\Big(\log\big(e^{s(x)}\big)+\log(1-p)\Big)\\
     & =s^{-1}\Big(s(x)+\log(1-p)\Big).
    \end{aligned}
    $$
    Let $\theta=\log(1-p)$. It follows that the statement that $x\mapsto F^{-1}\big((1-p)F(x)+p\big)$ is convex (concave) for all $p\in(0,1)$ is equivalent to the statement
    that $x\mapsto s^{-1}\left(s(x)+\theta\right)$ is convex (concave) for all $\theta<0$.

    \underline{Step 2.}
    Let $g(x):=-s^{-1}(x)$. Then, $g$ is strictly increasing. We will show that if $1/\eta$ is convex (concave), $\log(g'(x))$ is convex (concave).

    As $g(x)=-s^{-1}(x)=-S^{-1}(e^{x}),$ we have
    $$g'(x) =\frac{e^{x}}{f\left(S^{-1}(e^{x})\right)}.$$
    Let $H(x):=\log(g'(x))=x-\log\left(f\left(S^{-1}(e^{x})\right)\right)$.
    We have
    \begin{equation}
    H\left(\log(S(x))\right)=\log S(x)-\log f(x)=-\log\eta(x).\label{H}
    \end{equation}
    Then, taking the derivative on both sides of \eqref{H}, we get
    $$
    \begin{aligned}
    & -H'\big(\log(S(x))\big)\eta(x)=-\frac{\eta'(x)}{\eta(x)}\\
  \Longleftrightarrow ~~& H'\big(\log(S(x))\big)=\frac{\eta'(x)}{\eta^{2}(x)}=-\frac{\d}{\d x}\left(\frac{1}{\eta(x)}\right).
    \end{aligned}
    $$
    Taking a derivative in both sides again, we get
    $$
    -H''\big(\log(S(x))\big)\eta(x)  =-\frac{\d^{2}}{\d x^{2}}\left(\frac{1}{\eta(x)}\right).
    $$
    Then, $1/\eta$ is a convex (concave) function means $H''(x)\ge0$ ($H''(x)\le0$), which gives that $\log\ensuremath{g'(x)}$ is convex
    (concave).

    \underline{Step 3.} For $\theta<0$, let $G_{\theta}(x):=s^{-1}(s(x)+\theta)$.
    We are going to show
    \begin{equation}\label{limG}
    \lim_{z\to 0}\frac{G_{\theta}(x+z)-G_{\theta}(x)}{z}\le\lim_{z\to0}\frac{G_{\theta}(x'+z)-G_{\theta}(x')}{z}
    \end{equation}
    for all $x<x'$.

    We take $z>0$ first. As $s$ is strictly decreasing, $s^{-1}$ is also strictly decreasing. Then, $G_{\theta}$ is a continuous and strictly increasing function. As $\theta<0$, we also have $G_{\theta}(x)>x$.
    Take arbitrary $x,x',y$ and $z$ such that $x<x'$, $x<y$ and $z>0$.
    Let $\theta=s(y)-s(x)$. Then, we have $G_{\theta}(x)=y$.
    Define
    $$
    h=G_{\theta}(x+z)-y,\quad y'=G_{\theta}(x')\quad\text{and}\quad h'=G_{\theta}(x'+z)-y'.
    $$
    By the definition of $G_{\theta}$, we have $s(y+h)=s(x+z)+\theta$, $s(y')=s(x')+\theta$ and $s(y'+h')=s(x'+z)+\theta$. As a result, we have
    $$
    s(y+h)-s(y)  =s(x+z)-s(x)\quad\text{and}\quad s(y'+h')-s(y')=s(x'+z)-s(x').
    $$

    By the mean-value theorem, there exists $\zeta\in(y,y+h)$, $\zeta'\in(y',y'+h')$, $\xi\in(x,x+z)$ and $\xi'\in(x',x'+z)$ such that
    $$
    s'(\zeta)h =s'(\xi)z\text{ and }s'(\zeta')h'=s'(\xi')z.
    $$
    Furthermore, $x,x',y$ and $y'$ satisfy $x<x'<y'$ and $x<y<y'$.
    If $z$ is small enough, then $h$ and $h'$ will also be small enough as $G_{\theta}$ is continuous.
    Therefore, we have $\xi<\xi'<\zeta'$ and $\xi<\zeta<\zeta'$ when $z$ is small enough.

    In Step 2, we have that $\log\ensuremath{g'(x)}$ is convex when $1/\eta$ is convex. Therefore, we get
    $$\log(g'(a))+\log(g'(b))  \ge\log(g'(a'))+\log(g'(b')),$$
    for all $a<a'<b$ and $a<b'<b$, which means
    $$g'(a)g'(b)  \ge g'(a')g'(b').$$
    As $g(x)=-s^{-1}(x)$, we have
    $$\frac{1}{s'(s^{-1}(a))s'(s^{-1}(b))} \ge\frac{1}{s'(s^{-1}(a'))s'(s^{-1}(b'))}.$$
    As $s^{-1}(x)$ is strictly decreasing, it means that
    $$s'(\alpha)s'(\beta)  \le s'(\alpha')s'(\beta')$$
    for $\alpha>\alpha'>\beta$ and $\alpha>\beta'>\beta$.
    Therefore, we have $s'(\xi)s'(\zeta')\le s'(\xi')s'(\zeta)$ as $\zeta'>\zeta>\xi$ and $\zeta'>\xi'>\xi$.
    That is,
    $$h  =\frac{s'(\xi)z}{s'(\zeta)}\le\frac{s'(\xi')z}{s'(\zeta')}=h'.$$

    On the other hand, $h=G_{\theta}(x+z)-G_{\theta}(x)$ and $h'=G_{\theta}(x'+z)-G_{\theta}(x')$.
    Therefore,
    $$G_{\theta}(x+z)-G_{\theta}(x)  \le G_{\theta}(x'+z)-G_{\theta}(x')\label{G}$$
    when $z$ is small enough. If $z<0$, we can also get \eqref{limG} by an analogous argument.

    Hence,   the second-order derivative of $G_{\theta}$ is increasing for each $\theta<0$, which means that $x\mapsto F^{-1}\big((1-p)F(x)+p\big)$ is convex for all $p\in(0,1)$ if $1/\eta$ is convex.

    An analogous argument yields that $x\mapsto F^{-1}\big((1-p)F(x)+p\big)$ is concave for all $p\in(0,1)$ when $1/\eta$ is concave.
\end{proof}
\end{document}